\documentclass[%
nofootinbib,
 preprint,
superscriptaddress,
preprintnumbers,
 amsmath,amssymb,
 aps,pra,
longbibliography,
]{revtex4-1}

\usepackage{hyperref}
\hypersetup{
colorlinks=true,
linkcolor=blue,
filecolor=blue,
citecolor=blue,  
urlcolor=black,
}

\renewcommand{\thesection}{\arabic{section}}
\renewcommand{\thesubsection}{\thesection.\arabic{subsection}}
\renewcommand{\thesubsubsection}{\thesubsection.\arabic{subsubsection}}

\usepackage{titlesec}
\titleformat{\section}
  {\normalfont\fontsize{16}{15}\bfseries}{\thesection}{1em}{}
\titleformat{\subsection}
  {\normalfont\fontsize{14}{15}\bfseries}{\thesubsection}{1em}{}
\titleformat{\subsubsection}
  {\normalfont\fontsize{12}{15}\bfseries}{\thesubsubsection}{1em}{}

\usepackage{setspace}

\usepackage{float}
\restylefloat{table}

\usepackage{multirow}
\usepackage{color} 
\usepackage{graphicx}
\usepackage{dcolumn}
\usepackage{bm}

\usepackage{verbatim}
\usepackage{appendix}
\usepackage{amsthm}
\usepackage{hyperref}
\hypersetup{
colorlinks=true,
linkcolor=black,
filecolor=blue,
citecolor=blue,  
urlcolor=black,
}

\theoremstyle{plain}
\newtheorem{thm}{Theorem}
\newtheorem{lem}[thm]{Lemma}

\newtheorem{cor}[thm]{Corollary}

\theoremstyle{definition}
\newtheorem{defn}{Definition}
\newtheorem{conj}{Conjecture}

\newcommand{\bra}[1]{{\langle#1|}}
\newcommand{\ket}[1]{{|#1\rangle}}

\newcommand{\ketbra}[2]{{\ket{#1}\!\bra{#2}}}
\newcommand{\norm}[1]{\left\lVert#1\right\rVert}

\newcommand{\tr}{\operatorname{tr}}
\DeclareMathOperator{\bbE}{\mathbb{E}}
\newcommand{\wg}{\mathrm{Wg}}

\newcommand{\orb}{\mathrm{orb}}

\DeclareMathOperator{\Sym}{Sym}
\DeclareMathOperator{\Hom}{Hom}

\def\vec#1{\bm{#1}} 

\newcommand{\rmd}{\mathrm{d}}
\newcommand{\rme}{\mathrm{e}}

\newcommand{\rmU}{\mathrm{U}}

\def\<{\langle}  
\def\>{\rangle}  

\newcommand{\bbC}{\mathbb{C}}

\newcommand{\cat}{\mathrm{Cat}}

\begin{document}
\preprint{MIT-CTP/4874}

\linespread{1}
\title{\singlespacing\Large Entanglement, quantum randomness, and \\ complexity beyond scrambling}

\author{Zi-Wen Liu}\email{zwliu@mit.edu}
\affiliation{Center for Theoretical Physics, Massachusetts Institute of Technology, Cambridge, MA 02139, USA}
\affiliation{Department of Physics, Massachusetts Institute of Technology, Cambridge, MA 02139, USA}
\author{Seth Lloyd}
\affiliation{Department of Mechanical Engineering, Massachusetts Institute of Technology, Cambridge, MA 02139, USA}
\affiliation{Department of Physics, Massachusetts Institute of Technology, Cambridge, MA 02139, USA}
\author{Elton Zhu}
\affiliation{Center for Theoretical Physics, Massachusetts Institute of Technology, Cambridge, MA 02139, USA}
\affiliation{Department of Physics, Massachusetts Institute of Technology, Cambridge, MA 02139, USA}
\author{Huangjun Zhu}
\affiliation{Institute for Theoretical Physics, University of Cologne, 50937 Cologne, Germany}
\affiliation{Department of Physics and Center for Field Theory and Particle Physics, Fudan University, Shanghai 200433, China}
\affiliation{Institute for Nanoelectronic Devices and Quantum Computing, Fudan University, Shanghai 200433, China}
\affiliation{State Key Laboratory of Surface Physics, Fudan University, Shanghai 200433, China}
\affiliation{Collaborative Innovation Center of Advanced Microstructures, Nanjing 210093, China}

\makeatletter
\renewcommand{\abstractname}{\small\bf Abstract}
\makeatother

\vspace{\fill}
\begin{abstract}
\begin{spacing}{1}
Scrambling is a process by which the state of a quantum system is effectively randomized due to the global entanglement that ``hides'' initially localized quantum information. Closely related notions include quantum chaos and thermalization. Such phenomena play key roles in the study of quantum gravity, many-body physics, quantum statistical mechanics, quantum information etc. 
Scrambling can exhibit different complexities depending on the degree of randomness
it produces. For example, notice that the complete randomization implies scrambling, but the converse does not hold; in fact, there is a significant complexity gap between them.   In this work, we lay the mathematical foundations of studying randomness complexities beyond scrambling by entanglement properties. We do so by analyzing the generalized (in particular R\'enyi) entanglement entropies of designs, i.e.~ensembles of unitary channels or pure states that mimic the uniformly random distribution (given by the Haar measure) up to certain moments.  
A main collective conclusion is that the R\'enyi entanglement entropies averaged over designs of the same order are almost maximal. This links the orders of entropy and design, and therefore suggests R\'enyi entanglement entropies as diagnostics of the randomness complexity of corresponding designs. Such complexities form a hierarchy between information scrambling and Haar randomness.
As a strong separation result, we prove the existence of (state) 2-designs such that the R\'enyi entanglement entropies of higher orders can be bounded away from the maximum.
However, we also show that the min entanglement entropy is maximized by designs of order only logarithmic in the dimension of the system.
In other words, logarithmic-designs already achieve the complexity of Haar in terms of entanglement, which we also call max-scrambling. 
This result leads to a generalization of the fast scrambling conjecture, that max-scrambling can be achieved by physical dynamics in time roughly linear in the number of degrees of freedom.  
\end{spacing}
\end{abstract}

\maketitle

\begin{spacing}{1.3}
\par

\makeatletter
\renewcommand\@dotsep{10000}
\makeatother

\tableofcontents

\section{Introduction}
Scrambling describes a property of the dynamics of isolated quantum systems, in which initially localized quantum information spreads out over the whole system, thereby becoming inaccessible to local observers. 
The notion of scrambling originates from the study of black holes in quantum gravity \cite{mirror,fast,add}. The thermal nature of the Hawking radiation \cite{hawking1974,hawking1975,thermo} indicates that the state of any matter and information falling into the black hole has been scrambled and so gets lost from the perspective of an external observer.
In particular, the ``fast scrambling conjecture'' \cite{fast} states that the fastest scramblers take time logarithmic in the system size to scramble information, and that black holes are the fastest scramblers.  

Scrambling and similar notions play important roles in other areas of physics as well.
For example, scrambling is closely related to many-body localization and quantum thermalization (see \cite{mblrev} for a recent review): quantum systems that exhibit localization clearly do not scramble or thermalize, since local quantum information may fail to spread, and so remains accessible to certain local measurements.  By contrast, a many-body system that undergoes scrambling evolves to states that appear random with respect to local measurements: here, the notion of scrambling can be seen as a form of thermalization at infinite temperature.
Quantum chaos is also a close relative of scrambling.  Under chaotic dynamics, initially local operators grow to overlap with the whole system (the butterfly effect). That is, chaotic quantum systems are scramblers \cite{chaos}. In particular, the behaviors of the so-called out-of-time-order (OTO) correlators can probe the growth of local perturbations. 
Their role as diagnostics of chaos has led to the active application of OTO correlators to the study of scrambling \cite{Shenker2014,Shenker20142,Roberts2015,Shenker2015,PhysRevLett.115.131603,Maldacena2016,chaos,wochaos,PhysRevA.94.040302,ry} and many-body localization \cite{ANDP:ANDP201600318,FAN2017707,ANDP:ANDP201600332}. 

This work is mainly motivated by two key features of scrambling. First, scrambling of quantum information and the growth of entanglement go hand in hand: information initially present in local perturbations ends up being irretrievable by local or simple measurements even though closed-system (unitary) evolutions do not actually erase any information, since it gets encoded in global entanglement. Entanglement captures the nonclassical essence of scrambling, and could be a natural and powerful probe of scrambling properties.
Second, scrambling is intimately connected to the generation of randomness. Loosely speaking, scrambling and chaos describe the phenomenon that the system is effectively randomized.   Indeed, the effects of information scrambling such as local indistinguishability \cite{Lashkari2013} and the decay of OTO correlators \cite{chaos} can be achieved by random dynamics given by a random unitary channel drawn from the group-invariant Haar measure.  A key idea of the seminal Hayden-Preskill work \cite{mirror} is to use random dynamics to model the scrambling behaviors of black holes.   However, such observations are essentially ``one-way'': scrambling do not necessarily imply full randomness.  As we shall further clarify, there is in fact a large gap of complexity between information scrambling and complete randomness. The notion of ``scrambling'' needs to be refined since it can correspond to vastly different randomness complexities.



 The major goal of this paper is to connect these two features and lay the mathematical foundations of diagnosing the randomness complexities associated with scrambling by entanglement.  This is achieved by studying the interplay between the degrees of entanglement and quantum randomness.  Note that studies along this line are also of great interest to many areas in quantum information.
 A basic result in this direction is that the expected entanglement entropy of a Haar random pure state is almost maximal, which is usually known as the Page's theorem \cite{Page93,FoonK94, Sanc95, Sen96}. However, this result is not tight in the sense that there is a large gap between the complexities of the Haar randomness and entanglement entropy conditions: the complexity of the Haar measure (given by e.g.~the optimal depth of local circuits that approximate it) grows exponentially in the number of qubits \cite{1995quant.ph..8006K}, while the near-maximal entanglement entropy only needs finite moments of the Haar measure, which have only polynomial complexity and can be efficiently implemented \cite{brandaoharrow1,brandaoharrow2,PhysRevA.75.042311,nakata}.  This also illustrates the separation between the loss of local information or information scrambling and Haar randomness as large entanglement entropy indicates that local information is spread out (which will be discussed in more detail later).
 The regime in between information loss and complete randomness is not well understood in the contexts of both the dynamical behaviors of scrambling or chaos, and the kinematic entanglement properties.

 To fill this gap, we consider more stringent entanglement measures and pseudorandom ensembles of quantum states and processes.  In particular, we analyze the generalized entanglement entropies of pseudorandom ensembles of pure states and unitary channels known as designs, both parametrized by an order index. 
 Generalized entanglement entropies of order $\alpha$ are entropic functions of the $\alpha$-th power of the reduced density matrix.  The higher the order of the generalized entropy, the more sensitive that entropy is to nonuniformity (such as sharp peaks) in the spectrum of the density matrix and so the harder it is to maximize.  (A particular family known as the R\'enyi entropy is most ideal for our purpose.)
An $\alpha$-design is an ensemble of pure states or unitary operators whose first $\alpha$ moments are indistinguishable from the Haar random states or unitaries. The higher the order of the design, the better it emulates the completely random Haar distribution. 
We establish a strong connection between the order of the generalized entanglement entropies and the order of designs, in both the random unitary channel and random state settings. (We note that a recent paper \cite{ry} establishes a related connection between $2\alpha$-point OTO correlators and $\alpha$-designs via frame potentials.)   Our analysis indicates that $\alpha$-designs induce almost maximal R\'enyi-$\alpha$ entanglement entropies, thereby tightening (in a complexity-theoretical sense) known results relating entanglement entropy and quantum randomness, such as Page's theorem for random states and similar results for random unitaries by Hosur/Qi/Roberts/Yoshida \cite{chaos}.    
This result reveals a fine-grained hierarchy of randomness complexities between information and Haar scrambling defined relative to the moments of the Haar measure, and suggests R\'enyi entanglement entropies of the corresponding order as useful diagnostics. For example, if the R\'enyi-$\alpha$ entanglement entropy for some way of partitioning the system does not meet the maximality condition, then one can argue that the system has not reached the complexity of $\alpha$-designs.   Since our characterization of such complexities of designs rely on entropy, we also refer to the joint notions as ``entropic scrambling/randomness complexities''. 

Interestingly, there cannot be infinitely many different orders of designs that can be separated by R\'enyi entanglement entropies.  This is seen by analyzing the min entanglement entropy, i.e.~the infinite order limit of R\'enyi entropy, which only depends on the largest eigenvalue and lower bounds all R\'enyi entropies.  Large min entanglement entropy indicates that the entanglement spectrum is almost completely uniform, and therefore the local information is totally lost and the system looks completely random even if one has access to the whole reduced density matrix.
That is, the system essentially becomes indistinguishable from being Haar random by entanglement.
This corresponds to a strong form of information scrambling, which we call ``max-scrambling''.
We show that the min entanglement entropy (and therefore all R\'enyi entanglement entropies) becomes almost maximal, for designs of an order that is only logarithmic in the dimension of the system. In terms of entanglement properties, there can be at most logarithmic ``nontrivial'' orders of designs or moments of the Haar measure. Designs of higher orders all behave like completely random and are essentially the same. 
This result leads to a strong estimate of the shortest max-scrambling time, which generalizes the fast scrambling conjecture, that max-scrambling can be achieved by physical dynamics in  time roughly linear in the number of degrees of freedom.  


Now we summarize the mathematical techniques and results more specifically. 
We first focus on the intrinsic scrambling and randomness properties of physical processes, which are represented by unitary channels.  We map unitary channels to a dual state via the Choi isomorphism, and study the entanglement associated with this dual state. 
As in \cite{chaos}, we partition the input register of the Choi state into two parts, $A$ and $B$, and the output register into $C$ and $D$.
Our results rely on the calculation of average $\mathrm{tr}\{\rho_{AC}^\alpha\}$, the defining element of order-$\alpha$ entanglement entropies between $AC$ and $BD$ of the Choi state.  We mainly employ tools from combinatorics and Weingarten calculus to compute the Haar integrals of $\mathrm{tr}\{\rho_{AC}^\alpha\}$ in various cases, which are equal to the average over unitary $\alpha$-designs due to their defining properties. 
The convexity of R\'enyi entropies in the trace term allows us to use these results to lower bound the R\'enyi entanglement entropies by Jensen's inequality. The asymptotic result is that the R\'enyi-$\alpha$ entanglement entropies for equal partitions averaged over unitary $\alpha$-designs are almost maximal, or more precisely, at most smaller than the maximal value by a constant that is independent of the dimension and the order. In other words, the difference is vanishingly small.  This conclusion relies on a lemma on the number of cycles associated with permutations.  In other words, a random unitary sampled from a unitary $\alpha$-design is very likely to exhibit nearly maximal R\'enyi-$\alpha$ entanglement entropies, which supports the idea of using R\'enyi-$\alpha$ entanglement entropies as witnesses of the complexity of $\alpha$-designs.
For finite dimensions, we also derive explicit bounds on the $\alpha$-design-averaged R\'enyi-$\alpha$ entanglement entropy using modern tools developed for Haar integrals.
It is natural to ask how robust the above results are against small deviations from exact unitary designs.  We derive error bounds for two common but slightly different ways to define approximate unitary designs. 
The extreme cases are actually quite interesting.   In particular, we find that finite-order designs are sufficient to maximize the entanglement entropy given by the R\'enyi entropy of infinite order, namely the min entropy. As mentioned above, we show that, rather surprisingly, unitary designs of an order that scales logarithmically in the dimension of the unitary induce min entanglement entropy that is at most a constant away from the maximum, which implies that they are already indistinguishable from Haar by the entanglement spectrum alone. 

Then we study the mathematically more straightforward and more well-known problem of entanglement in random states.   The main results are very analogous to those in the random unitary setting, but the derivations are simpler since there are only two subsystems involved.   Most importantly, we show that (projective) $\alpha$-designs exhibit almost maximal R\'enyi-$\alpha$ entanglement entropies, which can be regarded as a collection of tight Page's theorems. And similarly, designs of logarithmic order maximize the min entanglement entropy.  In addition, we are able to obtain the following separation result which is not there yet in the unitary setting.
We show by representation theory that there exist 2-designs whose R\'enyi entanglement entropies of higher orders are bounded away from the maximum.  The existence of such  2-designs can be regarded as the indicator of a separation between the complexity of 2-designs and those of higher orders as diagnosed by R\'enyi entanglement entropies.   
The paper also includes several other results related to e.g.\ R\'enyi entropies, designs, and Weingarten calculus, which may be of independent interest.
These mathematical results may find applications in many other relevant areas, such as quantum cryptography and quantum computing. 

The paper is organized as follows. In Sec.~\ref{sec:prelim}, we formally define the central concepts of this paper---the generalized quantum entropies, and projective and unitary designs.
In Sections \ref{sec:RandU} and \ref{sec:RandS}, we study the Choi model of unitary channels and pure states respectively. We conclude in Sec.~\ref{sec:dis} with open problems and some discussions on the connections and possible extensions of our results to several other topics.
The appendix contains several peripheral results and technical tools. See e.g.~\cite{doubly} for a comprehensive introduction of standard and soft notations of asymptotics (e.g.\ big-O and soft big-O) that will be used throughout this paper. 
This paper provides the technical details of the results in \cite{PhysRevLett.120.130502}.

\section{Preliminaries}\label{sec:prelim}
The theme of this paper is to establish connections between generalized quantum entropies and quantum designs, which we shall formally introduce in this section. 

\subsection{Generalized quantum entropies}\label{entropies}

\subsubsection{Definitions of unified and R\'enyi entropies}
Some parametrized generalizations of the Shannon and von Neumann entropy, most importantly the R\'enyi and Tsallis entropies, are found to be useful in both classical and quantum regimes. 
Here we focus on the entropies defined on a quantum state represented by density matrix $\rho$ living in a finite-dimensional Hilbert space. 
A unified definition of generalized quantum entropies is given in \cite{huye,Rastegin2011}:
\begin{defn}[Quantum unified entropies]
The quantum unified $(\alpha,s)$-entropy of a density matrix $ \rho $ is defined as
\begin{equation}
S^{(\alpha)}_s(\rho)=\frac{1}{s(1-\alpha)}\left[(\mathrm{tr}\{\rho^\alpha\})^s-1\right].
\end{equation}
The two parameters $\alpha$ and $s$ are respectively referred to as the order and the family of an entropy.  In this paper, we mostly care about the cases where $\alpha$ is a positive integer and $s$ is a nonnegative integer.
\end{defn}
The $\mathrm{tr}\{\rho^\alpha\}$ element plays a key role in this paper. Entropies specified by a certain order $\alpha$ are collectively called $\alpha$ entropies. The $\alpha\rightarrow 1$ limit gives the von Neumann entropy.
By fixing $s$, one obtains a family of entropies parametrized by order $\alpha$. We define the following function to be the characteristic function of an entropy:
\begin{equation}
f^{(\alpha)}_s(x) = -\frac{x^s-1}{s(1-\alpha)},
\end{equation}
which is obtained by treating $\mathrm{tr}\{\rho^\alpha\}$ as the argument $x$. The convexity of characteristic functions is important to many of our results.

The most representative families of quantum entropies are R\'enyi (the limiting case $s\rightarrow 0$) and Tsallis ($s=1$) entropies. 
In this work, we shall mostly focus on the R\'enyi entropies:
\begin{defn}[Quantum R\'enyi entropies]
The quantum R\'enyi-$ \alpha $ entropy of a density matrix $ \rho $ is defined as
\begin{equation}\label{renyidef}
S_R^{(\alpha)}(\rho)=\frac{1}{1-\alpha}\log\mathrm{tr}\{\rho^\alpha\}.
\end{equation}
\end{defn}
For $\alpha=0,1,\infty$, $S_R^{(\alpha)}$ is singular and defined by taking a limit. 
$S_R^{(0)}(\rho)=\log{\rm rank}(\rho)$ is called the max/Hartley entropy; $S_R^{(1)}=-{\rm tr}\rho\log\rho$ is just the von Neumann entropy.  The $s\rightarrow\infty$ limit, which is called the min entropy, is particularly important for our study:
\begin{defn}[Quantum min entropy]
The quantum min entropy of a density matrix $ \rho $ is defined as
\begin{equation}
S_{\text{min}}(\rho)=-\log\norm{\rho} = -\log\lambda_{\text{max}}(\rho),
\end{equation}
where $\norm{\rho}$ denotes the operator norm of $\rho$, and $\lambda_{\text{max}}(\rho)$ is the largest eigenvalue of $\rho$.
\end{defn}
Other R\'enyi entropies are well defined by Eq.\ (\ref{renyidef}). The $\alpha=2$ case $S_R^{(2)}(\rho)=-\log\mathrm{tr}\{\rho^2\}$, also called the second R\'enyi entropy or collision entropy for classical probability distributions, is also a widely used and highly relevant quantity. In the context of scrambling, a key result of \cite{chaos} is that the R\'enyi-2 entanglement entropy is directly related to the 4-point OTO correlators, which has become a widely concerned quantity in recent years as a probe of chaos. 
Also notice that $S_R^{(2)}$ is directly related to the quantum purity $\mathrm{tr}\{\rho^2\}$ (recall that less pure subsystems dictate entanglement), and is thus frequently employed in the study of entanglement \cite{entrev,measureent}.

Fig.~\ref{fig:entropy} summarizes the important generalized entropies in the relevant regime.
\begin{figure}
    \centering
    \includegraphics[scale=0.8]{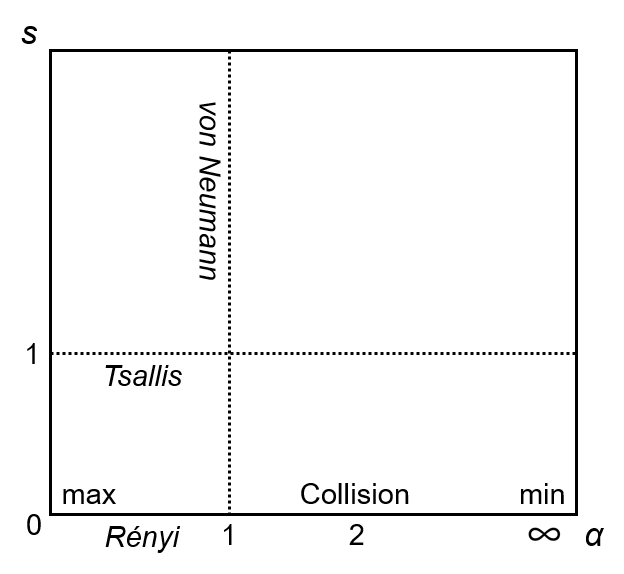}
    \caption{Unified $(\alpha,s)$-entropies, $\alpha > 0, s\geq 0$. Italicized names refer to the whole line.}
    \label{fig:entropy}
\end{figure}

\subsubsection{Important features of R\'enyi entropies}\label{sec:renyifeature}

We are particularly interested in the family of R\'enyi entropies since they have several desirable features that play important roles in our arguments throughout.

The following properties of each R\'enyi entropy are important for our purposes:
\begin{enumerate}
    \item They have the same maximal value $n$ for systems of $n$ qubits (attained by the uniform spectrum). This allows meaningful comparisons with the maximal value and between different orders; 
    \item They are additive on product states, i.e.,~$S_R^{(\alpha)}(\rho\otimes\sigma) = S_R^{(\alpha)}(\rho) + S_R^{(\alpha)}(\sigma)$ for all $\alpha$ and density matrices $\rho,\sigma$. Otherwise it is not natural to define extensive quantities such as mutual information and tripartite information;
     \item Their characteristic functions $f_R^{(\alpha)}$ are convex, i.e.,~$S_R^{(\alpha)}(\rho)$ is convex in $\mathrm{tr}\{\rho^\alpha\}$. This allows us to use Jensen's inequality to lower bound the design-averaged values by Haar integrals.
\end{enumerate}
These properties are all straightforward to verify.
These properties do not simultaneously hold for other families.  For example, it is easy to see that the first two fail for Tsallis entropies.  Later we shall further explain why these properties are desirable in explicit contexts.   However, we note that the calculations are essentially only about the trace term, so it is straightforward to obtain results for all families if one wishes.

In this work, we are particularly interested in the regimes where certain R\'enyi entropies are nearly maximal.
The following ``cutoff'' phenomenon concerning the maximality is an important foundation of our scheme of characterizing the complexity of scrambling by R\'enyi entropies.
First, notice that the unified entropy of a certain family, such as the R\'enyi entropy, is monotonically nonincreasing in the order: $S_R^{(\alpha)} \geq S_R^{(\beta)}$ if $\alpha<\beta$. (In particular, the min entropy sets a lower bound on all R\'enyi entropies: $S_{\text{min}}\leq S_R^{(\alpha)}$ for all $\alpha$.)  So if the R\'enyi entropy of some order is almost maximal, then those of lower orders are all almost maximal. 
Moreover, asymptotically, the values of R\'enyi entropies of different orders can be well separated, and for each order there exist inputs that attain almost maximal R\'enyi entropy of this order but those of all higher orders are small.  As will become clearer later, this allows for the possibility of distinguishing between different complexities by the asymptotic maximality of R\'enyi entropies of certain orders.  This feature can be illustrated by the following simple example. Given some order $\tilde\alpha$. Consider a density operator in the $d$-dimensional Hilbert space which has one large eigenvalue $1/d^{\frac{\tilde\alpha-1}{\tilde\alpha}}$, and the rest of the spectrum is uniform/degenerate. That is, the spectrum reads
\begin{equation}
\vec{\lambda} = \big(\frac{1}{d^{\frac{\tilde\alpha-1}{\tilde\alpha}}}, \underbrace{\frac{1-\frac{1}{d^{\frac{\tilde\alpha-1}{\tilde\alpha}}}}{d-1},\cdots,\frac{1-\frac{1}{d^{\frac{\tilde\alpha-1}{\tilde\alpha}}}}{d-1}}_{d-1}\big).
\end{equation}
The R\'enyi-$\tilde\alpha$ entropy (and thus all lower orders) is insensitive to this single peak:
\begin{eqnarray}
S^{(\tilde\alpha)}_R(\vec\lambda) &=& \frac{1}{1-\tilde\alpha}\log\left(\frac{1}{d^{\tilde\alpha-1}}+(d-1)\left(\frac{1-\frac{1}{d^{\frac{\tilde\alpha-1}{\tilde\alpha}}}}{d-1}\right)^{\tilde\alpha}\right)\\
&=& \log{d} - \frac{1}{\tilde\alpha-1}\log\left(1+\frac{d^{\tilde\alpha-1}\left(1-\frac{1}{d^{\frac{\tilde\alpha-1}{\tilde\alpha}}}\right)^{\tilde\alpha}}{(d-1)^{\tilde\alpha-1}}\right)\\&\geq& \log{d}-1,
\end{eqnarray}
that is, $S^{(\tilde\alpha)}_R(\vec\lambda)$ is almost maximal, up to a small residual constant.
However, the R\'enyi entropies of higher orders can detect this peak and become small. For $\beta > \tilde\alpha$,
\begin{equation}
S^{(\beta)}_R(\vec\lambda) = \frac{1}{1-\beta}\log\left(\frac{1}{d^{\frac{\beta(\tilde\alpha-1)}{\tilde\alpha}}}+(d-1)\left(\frac{1-\frac{1}{d^{\frac{\tilde\alpha-1}{\tilde\alpha}}}}{d-1}\right)^\beta\right)
\approx \frac{\tilde\alpha\beta - \beta}{\tilde\alpha\beta-\tilde\alpha}\log{d},
\end{equation}
which is $\Theta(\log{d})$ (linear in the number of qubits) smaller than the maximal value $\log d$.  In fact, $\vec\lambda$ produces $\Theta(\log{d})$ gaps between all higher orders.
The extreme case min entropy only cares about the largest eigenvalue by definition:
\begin{equation}
S_{\text{min}}(\vec\lambda) = -\log\lambda_\mathrm{max} = \frac{\tilde\alpha-1}{\tilde\alpha}\log{d},
\end{equation}
which is small for all finite $\tilde\alpha$.  
That is, the slope of $S^{(\beta)}_R(\vec\lambda)$ in $\log d$ decreases with $\beta$. It equals one for $\beta = \tilde\alpha$, and approaches $\frac{\tilde\alpha-1}{\tilde\alpha}$ in the $\beta\rightarrow\infty$ limit.
So there can be an asymptotic separation between R\'enyi entropies of any orders. 
The intuition is simply that promoting the power of eigenvalues essentially amplifies the nonuniformity of the spectrum.  We shall construct a similar separation for certain 2-designs, which indicates that R\'enyi entropies can distinguish low-degree pseudorandom states from truly random states.

In our calculations we often assume equal partitions for simplicity.   Since the subsystems contain half the total degrees of freedom, the equal partitions admit the largest possible entanglement entropy.
Also, the following simple argument ensures that as long as the (R\'enyi) entanglement entropies between all equal partitions are close to the maximum, then that between generic partitions must be close to the maximum as well.  Notice that the quantum R\'enyi divergence/relative entropy (either the non-sandwiched or sandwiched/non-commutative version, see e.g.~\cite{doi:10.1063/1.4838856} for definitions)
between $\rho$ and the maximally mixed state yields the gap between the R\'enyi entropy of $\rho$ and the maximum:
\begin{equation}
    D_R^{(\alpha)}(\rho\|I/d) = \frac{1}{\alpha-1}\log(\tr\{\rho^\alpha\}d^{\alpha-1}) = \log d - S_R^{(\alpha)}(\rho),
\end{equation}
For sandwiched R\'enyi divergence with $\alpha\geq 1/2$ (which covers the parameter range of interest in this paper), it is shown in \cite{franklieb,beigi} that the data processing inequality holds, which implies that the divergence is monotonically nonincreasing under partial trace. So the gap can only be smaller when we look at smaller subsystems.

In the appendix, we derive more properties of R\'enyi entropies, including inequalities relating different orders of R\'enyi entropies (Appendix \ref{app:renyiineq}), and a weaker form of subadditivity (Appendix \ref{app:renyisub}).  The above discussions on R\'enyi entropies are more or less tailored to our needs. We refer the interested readers to \cite{doi:10.1063/1.4838856} for a more comprehensive discussion of the motivations and properties of quantum R\'enyi entropies and divergences.
We also note that a close variant of the quantum R\'enyi entropy known as the ``modular entropy'', given by $\tilde S_R^{(\alpha)}(\rho) = \frac{1}{\alpha^2}\partial_\alpha(\frac{\alpha-1}{\alpha}S_R^{(\alpha)}(\rho))$, is found to be meaningful in the context of holography and admits a natural thermodynamic interpretation \cite{xidong,HoloEnt}.  



\subsection{Designs}\label{designs}
In quantum information theory, the notion of $t$-designs characterizes distributions of pure states or unitary channels that mimic the uniform distribution up to the first $t$ moments, and so can be considered as good approximations to Haar randomness. Analogous classical notions such as $t$-wise independence and $t$-universal hash functions are also found to be very useful in computer science and combinatorics.  We shall formally introduce the definitions of state and unitary designs relevant to this work in the following. 

  \subsubsection{Complex projective designs}
  Complex projective $t$-designs, which we may call ``$t$-designs'' for short throughout the paper, are distributions of vectors on the complex unit sphere that are good approximations to the uniform distribution, or pseudorandom, in the sense that they reproduce the first $t$ moments of the uniform distribution \cite{Hogg82,ReneBSC04,rep}. They are of interest in many research areas, such as  approximation theory,  experimental designs, signal processing, and quantum information.  
There are many equivalent definitions of exact designs (see \cite{rep} for as introduction). Here we mention a few that are directly relevant to the current study. 
  
  The canonical definition based on polynomials of vector entries will be directly used in deriving our results.
  Define  $\Hom_{(t,t)}(\bbC^d)$ as the space of  polynomials homogeneous of degree $t$  both in the coordinates of vectors in $\bbC^d$
  and in their complex conjugates.
  \begin{defn}[$t$-designs by polynomials]
An ensemble $\nu$ of pure state vectors in dimension $d$ is a (complex projective) \emph{$t$-design} if 
  	\begin{equation}
\bbE_\nu p(\psi)= \int p(\psi)\rmd \psi\quad \forall p\in \Hom_{(t,t)}(\bbC^d),
  	\end{equation}
  	where the integral is taken with respect to the (normalized) uniform  measure on the complex unit sphere in $\bbC^d$. 
  \end{defn}

 The second definition, based on the frame operator, is also widely used.
 Let $\Sym_t(\bbC^d)$ be the $t$-partite symmetric subspace of $(\bbC^d)^{\otimes t}$ with corresponding projector $P_{[t]}$.
  The dimension of   $\Sym_t(\bbC^d)$ reads
  \begin{equation}
  	D_{[t]}=\binom{d+t-1}{t}.
  \end{equation}
  \begin{defn}[$t$-designs by frame]
  The $t$-th frame operator of $\nu$ is defined as
  \begin{equation}\label{eq:tFrameO}
  \mathcal{F}_t(\nu):= D_{[t]}\bbE_\nu
  (\ketbra{\psi}{\psi})^{\otimes t},
  \end{equation}  
   and the $t$-th frame potential  is 
   \begin{equation}\label{eq:tFrameP}
   \Phi_t(\nu):=\tr\left\{\mathcal{F}_t(\nu)^2\right\}. 
   \end{equation} 
The ensemble $\nu$ is a $t$-design if and only if $\mathcal{F}_t(\nu)=P_{[t]}$ or, equivalently, if $\Phi_t(\nu)=D_{[t]}$ \cite{rep}. 
\end{defn}

The above definitions for exact designs are equivalent.  However, they lead to slightly different ways to define approximate designs by directly considering the deviations from equality, which essentially represent different norms. We shall discuss the approximate designs in more detail later for error analysis.


  \subsubsection{Unitary designs}
  \label{sec:intro:unitary}
  In analogy to complex projective $t$-designs, 
  unitary $t$-designs are distributions on the 
   unitary group that  are good approximations to the Haar measure, in the sense that they reproduce the Haar measure up to the first $t$ moments \cite{DiViLT02, Dank05the,DankCEL09, GrosAE07, rep}. 
 They also play key roles in many research areas, such as randomized benchmarking, data hiding, and decoupling. As in the case of state designs there are also many equivalent definitions of exact unitary designs (see \cite{rep}). Similarly, we formally define unitary designs by polynomials and frame operators/potentials.

  Let  $\Hom_{(t,t)}(\rmU(d))$ be the space of polynomials homogeneous of degree $t$ both in the matrix elements of $U\in \rmU(d)$
  and in their complex conjugates. 
  \begin{defn}[Unitary $t$-designs by polynomials]
  	An ensemble $\nu$ of  unitary operators in dimension $d$  is a  \emph{unitary $t$-design} if 
  	\begin{equation}\label{eq:U2design}
  	\bbE_\nu p(U) =\int \rmd U p(U) \quad \forall p\in \Hom_{(t,t)}(\rmU(d)),
  	\end{equation}
  	where the integral is taken over the normalized Haar measure on ${\rm U}(d)$.
  \end{defn}

\begin{defn}[Unitary $t$-designs by frame]
The $t$-th frame operator of $\nu$ is defined as
\begin{equation}
\mathcal{F}_t(\nu):=\bbE_\nu \left[U^{\otimes t}\otimes {U^\dag}^{\otimes t}\right],
\end{equation}
and the $t$-th frame potential is 
  \begin{equation}
  	\Phi_t(\nu):=\tr\left\{\mathcal{F}_t(\nu)^2\right\}
  \end{equation}
The ensemble $\nu$ is a unitary $t$-design if and only if $\mathcal{F}_t(\nu)=\mathcal{F}_t(\rmU(d))$, where $\mathcal{F}_t(\rmU(d))$ is the $t$th frame operator of the unitary group $\rmU(d)$ with Haar measure \cite{rep}. In addition,
  \begin{equation}
  	\Phi_t(\nu)\geq \gamma(t,d):=\int \rmd U |\tr\{U\}|^{2t},
  \end{equation}
  and the lower bound is saturated if and only if $\nu$ is a unitary $t$-design \cite{GrosAE07,  RoyS09, rep}. When $t\leq d$, which is the case we are mostly interested in, 
  \begin{equation}\label{eq:FramePmin}
 \gamma(t,d)= t!.
  \end{equation}
  \end{defn}
  
  Again, the definitions are equivalent for exact unitary designs, but lead to different ways to define approximate unitary designs, which we shall look into later.
  




\section{Generalized entanglement entropies and random unitary channels}\label{sec:RandU}
Unitary channels describe the evolutions of closed quantum systems.  Here we study the entanglement and scrambling properties of random unitary channels, which directly motivates this work.
As suggested by \cite{chaos}, we employ the Choi isomorphism to map a unitary channel to a dual state, and study scrambling by the relevant entanglement properties of this state. In this section, we first briefly introduce the Choi state model, and then present the explicit calculations of generalized entanglement entropies averaged over unitary designs. The results lead to an entropic notion of scrambling or randomness complexities, which we shall discuss in depth.

\subsection{Model: entanglement in the Choi state}
Ref.~\cite{chaos} proposed that one can use the negativity of the tripartite information associated with the Choi state of a unitary channel to probe information scrambling.  The negative tripartite information is actually a measure of global entanglement that quantifies the degree to which local information in the input to the channel becomes non-local in the output.  We first introduce the definitions and motivations of this formalism to set the stage.

The Choi isomorphism (more generally, the channel-state duality) is widely used in quantum information theory to study quantum channels as states. It says that a unitary operator $U$ acting on a $d$-dimensional Hilbert space $U=\sum_{i,j=0}^{d-1}U_{ij}\ket{i}\bra{j}$ is dual to the pure state
\begin{equation}\label{choi}
    \ket{U} = \frac{1}{\sqrt{d}}\sum_{i,j=0}^{d-1}U_{ji}\ket{i}_{in}\otimes\ket{j}_{out},
\end{equation}
which is called the Choi state of $U$. 
Now consider arbitrary bipartitions of the input register into $A$ and $B$, and the output register into $C$ and $D$. Let $d_A,d_B,d_C,d_D$ be the dimensions of subregions $A,B,C,D$ respectively ($d_A d_B= d_C d_D = d$). One expects that, in a system that scrambles information, any measurement on local regions of the output cannot reveal much information about local perturbations applied to the input. In other words, the mutual information between local regions of the input and output $I(A:C)$ and $I(A:D)$ should be small. This suggests that the negative tripartite information
\begin{equation}
    -I_3(A:C:D) := I(A:CD) - I(A:C) - I(A:D)
\end{equation}
can diagnose scrambling, since it essentially measures the amount of information of $A$ hidden nonlocally over the whole output register. 
Here $I(A:C)=S(A)+S(C)-S(AC)$ is the mutual information, which measures the total correlation between $A$ and $C$. Since the input and output are maximally mixed due to unitarity, the four subregions are all maximally mixed. For example, here $I(A:C)$ is reduced to $\log{d_Ad_C} - S(AC)$, so we only need to analyze the entanglement entropy $S(AC)$.  
Note that $-I_3$ can be reduced to the conditional mutual information $I(A:B|C)$ \cite{ding}, which is a quantity of great interest in quantum information theory.

The Haar-averaged (completely random) values of the terms in the von Neumann $-I_3$ was computed in \cite{chaos}, as a baseline for scrambling. However, it is clear that a pseudorandom ensemble (such as 2-designs) can already reach these roof values \cite{ry}, which indicates that information scrambling only corresponds to randomness of low complexity in contrast to Haar. We are going to generalize the above quantities in the Choi state model using generalized entropies $S_s^{(\alpha)}$. 
Since the maximally mixed states have maximal generalized entropies, we only need to analyze $S_s^{(\alpha)}(AC)$.

\subsection{Relevant reduced density matrices of the Choi state}\label{Haarentropy}

To calculate the generalized entanglement entropies, we first need to derive the moments of the reduced density matrix of $AC$ and the expression for their traces. 

By using individual indices for different subregions, we rewrite the Choi state in Eq.\ (\ref{choi}) as
 \begin{equation}\label{choi2}
    \ket{U} = \frac{1}{\sqrt{d}}\sum_{klmo}U_{mo,kl}\ket{kl}_{AB}\otimes\ket{mo}_{CD},
\end{equation}
where $k,l,m,o$ are respectively indices for $A,B,C,D$. The corresponding density matrix is then
\begin{equation}
\rho_{ABCD}=\ketbra{U}{U}=\frac{1}{d}\sum_{\substack{klmo\\k'l'm'o'}}U_{mo,kl}U^*_{m^\prime o^\prime, k^\prime l^\prime}\ket{kl}_{AB}\bra{k^\prime l^\prime}\otimes\ket{mo}_{CD}\bra{m^\prime o^\prime}.
\end{equation}
By tracing out $BD$, we obtain the reduced density matrix of $AC$:
\begin{equation}
\rho_{AC}=\frac{1}{d}\sum_{\substack{klmo\\k'm'}}U_{mo, kl}U^*_{m^\prime o, k^\prime l}\ket{k}_A\bra{k^\prime}\otimes\ket{m}_C\bra{m^\prime}.
\end{equation}
The entropy of $\rho_{AC}$ measures the entanglement between $AC$ and $BD$. In order to compute the generalized $\alpha$ entanglement entropies, we need to raise $\rho_{AC}$ to the power $\alpha$:
\begin{eqnarray}
\rho_{AC}^\alpha=\frac{1}{d^{\alpha}}\sum_{\text{all indices}}&U_{m_1o_1,k_1l_1}U^*_{m_2o_1,k_2l_1}U_{m_2o_2,k_2l_2}U^*_{m_3o_2,k_3l_2}\cdots\nonumber\\
&U_{m_{\alpha}o_{\alpha},k_{\alpha}l_{\alpha}}U^*_{m_{\alpha+1} o_{\alpha},k_{\alpha+1}l_{\alpha}}\ket{k_1}\bra{k_{\alpha+1}}\otimes\ket{m_1}\bra{m_{\alpha+1}}.
\end{eqnarray}
Therefore,
\begin{equation}\label{tr}
\mathrm{tr}\left\{\rho_{AC}^\alpha\right\}=\frac{1}{d^{\alpha}}\sum_{\text{all indices}}U_{m_1o_1,k_1l_1}U^*_{m_2o_1,k_2l_1}U_{m_2o_2,k_2l_2}U^*_{m_3o_2,k_3l_2}\cdots U_{m_{\alpha}o_{\alpha},k_{\alpha}l_{\alpha}}U^*_{m_1 o_{\alpha},k_1l_{\alpha}}.
\end{equation}
This result can also take more concise operator forms:
\begin{equation}
   \mathrm{tr}\left\{\rho_{AC}^\alpha\right\}=\frac{1}{d^{\alpha}}\mathrm{tr}\left\{(U\otimes U^*)^{\otimes\alpha}X_\alpha\right\}=\frac{1}{d^{\alpha}}\mathrm{tr}\left\{(U\otimes U^\dagger)^{\otimes\alpha}Y_\alpha\right\},
\end{equation}
where
\begin{align}
X_\alpha:=\sum_{\text{all indices}}&\ket{m_1 o_1}\bra{k_1 l_1}\otimes \ket{m_2 o_1}\bra{k_2 l_1}\otimes\ket{m_2 o_2}\bra{k_2 l_2}\otimes\ket{m_3 o_2}\bra{k_3 l_2}\otimes\cdots\nonumber\\&\otimes\ket{m_\alpha o_\alpha}\bra{k_\alpha l_\alpha}\otimes\ket{m_1 o_\alpha}\bra{k_1 l_\alpha}  ,      \label{x}      \\
Y_\alpha:=\sum_{\text{all indices}}&\ket{m_1 o_1}\bra{k_1 l_1}\otimes \ket{k_2 l_1}\bra{m_2 o_1}\otimes\ket{m_2 o_2}\bra{k_2 l_2}\otimes\ket{k_3 l_2}\bra{m_3 o_2}\otimes\cdots\nonumber\\&\otimes\ket{m_\alpha o_\alpha}\bra{k_\alpha l_\alpha}\otimes\ket{k_1 l_\alpha}\bra{m_1 o_\alpha} 
= X_\alpha^{\Gamma_{\text{even}}},\label{y}
\end{align}
where $\Gamma_{\text{even}}$ denotes partial transpose on even parties. Notice that $Y_\alpha Y_\alpha^\dagger=I$ so $Y_\alpha$ is unitary.

Other density matrices can be derived in a similar way. Again note that the input and output are maximally entangled due to unitarity, so all four individual subregions are maximally mixed.

\subsection{Haar random unitaries}\label{sec:haarint}

\subsubsection{General trace formula}
We first employ tools from random matrix theory, combinatorics, and in particular Weingarten calculus, to compute the Haar integrals of the trace term in generalized entanglement entropies. 

It is known that the Haar-averaged value of each monomial of degree $\alpha$ can be written in the following form \cite{wg}:
\begin{equation}\label{alphawg}
\int{\rm d}U U_{i_1j_1}\cdots U_{i_\alpha j_\alpha}U^*_{i^\prime_1j^\prime_1}\cdots U^*_{i^\prime_\alpha j^\prime_\alpha}
=\sum_{\sigma,\gamma\in S_\alpha}\delta_{i_1i^\prime_{\sigma(1)}}\delta_{j_1j^\prime_{\gamma(1)}}\cdots \delta_{i_\alpha i^\prime_{\sigma(\alpha)}}\delta_{j_\alpha j^\prime_{\gamma(\alpha)}}\mathrm{Wg}(d,\sigma\gamma^{-1}),
\end{equation}
where $S_\alpha$ is the symmetric group of $\alpha$ symbols, and 
\begin{equation}
\mathrm{Wg}(d,\sigma) = \frac{1}{(\alpha!)^2}\sum_{\lambda\vdash\alpha}\frac{\chi^\lambda(1)^2\chi^\lambda(\sigma)}{s_{\lambda,d}(1,\cdots,1)}
\end{equation}
are called Weingarten functions of $\mathrm{U}(d)$. Here $\lambda\vdash\alpha$ means $\lambda$ is a partition of $\alpha$, $\chi^\lambda$ is the corresponding character of $S_\alpha$, and $s_\lambda$ is the corresponding Schur function/polynomial. Notice that $s_{\lambda,d}(1,\cdots,1)$ is simply the dimension of the irrep of $\mathrm{U}(d)$ associated with $\lambda$.
The Weingarten function can be derived by various tools in representation theory, such as Schur-Weyl duality \cite{collins,collins2} and Jucys-Murphy elements \cite{zinn}. Therefore, we obtain the following general result:
\begin{thm}\label{thm:trace}
\begin{align}
\int{\rm d}U \mathrm{tr}\left\{\rho_{AC}^\alpha\right\}
=&\frac{1}{d^{\alpha}}\sum_{\mathrm{all\;indices}}\int{\rm d}U U_{m_1o_1,k_1l_1}U^*_{m_2o_1,k_2l_1}U_{m_2o_2,k_2l_2}U^*_{m_3o_2,k_3l_2}\cdots U_{m_{\alpha}o_{\alpha},k_{\alpha}l_{\alpha}}U^*_{m_1 o_{\alpha},k_1l_{\alpha}}\nonumber\\
=&\frac{1}{d^{\alpha}}\sum_{\sigma,\gamma\in S_\alpha}d_A^{\xi(\sigma\tau)}d_B^{\xi(\sigma)}d_C^{\xi(\gamma\tau)}d_D^{\xi(\gamma)}\mathrm{Wg}(d,\sigma\gamma^{-1}),\label{general}
\end{align}
where $\xi(\sigma)$ is the number of disjoint cycles associated with $\sigma$ \footnote{Every element of the symmetric group can be uniquely decomposed into a product of disjoint cycles (up to relabeling).}, and $\tau:=(1~2~\cdots ~\alpha)$ is the 1-shift (canonical full cycle).
\end{thm}
One can easily recover the $\alpha = 2$ results given in \cite{chaos} from Eq.\ (\ref{general}) as follows.  The Weingarten functions for $\sigma\in S_2$ are
\begin{equation}
\wg(d,\sigma)=\begin{cases}
\frac{1}{d^2-1}  &\sigma =I,\\
-\frac{1}{d(d^2-1)} & \sigma=(1~2).
\end{cases}
\end{equation}
There are 4 terms corresponding to two different Weingarten functions:
\begin{table}[H]
    \centering
    \begin{ruledtabular}
    \begin{tabular}{ccccccc}
        $\sigma\gamma$ & $\sigma$ & $\gamma$ & $\xi(\sigma\tau)$ & $\xi(\sigma)$ & $\xi(\gamma\tau)$ & $\xi(\gamma)$   \\ \hline
         \multirow{2}{*}{$I$} & $I$ & $I$ & 1 & 2 & 1 & 2 \\
          & $(1~2)$ & $(1~2)$ & 2 & 1 & 2 & 1 \\ 
           \multirow{2}{*}{$(1~2)$} & $I$ & $(1~2)$ & 1 & 2 & 2 & 1 \\
           & $(1~2)$ & $I$ & 2 & 1 & 1 & 2 
    \end{tabular}
    \end{ruledtabular}
    \label{tab:2}
\end{table}

Plugging them into Eq.\ (\ref{general}) yields
\begin{align}
\int{\rm d}U{\rm tr}\{\rho_{AC}^2\}&=\frac{1}{d^{2}}\left(\frac{1}{d^{2}-1}\left(d_A d_B^2 d_C d_D^2+d_A^2 d_B d_C^2 d_D\right)-\frac{1}{d(d^{2}-1)}\left(d_A d_B^2 d_C^2 d_D+ d_A^2 d_B d_C d_D^2\right)\right)\\
&\approx  d_A^{-1} d_C^{-1} + d_B^{-1} d_D^{-1} - d^{-1}d_A^{-1}d_D^{-1} - d^{-1}d_B^{-1}d_C^{-1},
\end{align}
which confirms Eq.\ (66) of \cite{chaos}. A series of results of \cite{chaos} such as an $O(1)$ gap between the Haar-averaged and maximal R\'enyi-2 entanglement entropies are obtained based on this formula.

More generally, we have 
\begin{align}
\int{\rm d}U \left(\mathrm{tr}\left\{\rho_{AC}^\alpha\right\}\right)^s
=&\frac{1}{d^{s\alpha}}\sum_{\sigma,\gamma\in S_{s\alpha}}d_A^{\xi(\sigma\tau_{\alpha,s})}d_B^{\xi(\sigma)}d_C^{\xi(\gamma\tau_{\alpha,s})}d_D^{\xi(\gamma)}\mathrm{Wg}(d,\sigma\gamma^{-1}),\label{general2}
\end{align}
where $\tau_{\alpha,s}:=\prod_{r=0}^{s-1}(\alpha r+1 ~ \alpha r+2 ~\cdots~\alpha (r+1))$ is the product of canonical full cycles on each of the $s$ blocks with $\alpha$ symbols.

\subsubsection{Large $d$ limit asymptotics}
We now analyze the asymptotic behaviors of generalized entanglement entropies in the $d\rightarrow\infty$ limit to provide a big picture. Later we shall introduce some non-asymptotic bounds that hold for general $d$. To simplify the analysis, we consider equal partitions $d_A=d_B=d_C=d_D = \sqrt{d}$ here, which delivers the main idea.

\paragraph{Trace}
We first introduce a series of useful combinatorics lemmas, which play critical roles in the behavior of generalized entanglement entropies (in particular R\'enyi). 
These results are known in the contexts of random matrix theory and free probability theory. We refer the readers to \cite{lancien} (c.f.\ references therein) for a summary of related results or \cite{nica2006lectures} for a textbook on the subject. 


\begin{lem}[Cycle Lemma]\label{sumcycle}
	$ \xi(\sigma)+\xi(\sigma\tau)\leq \alpha+1$ for all $\sigma\in S_\alpha$, where $\xi$ counts the number of disjoint cycles.
\end{lem}
This result can be obtained by combining Lemmas A.1 and A.4 of \cite{lancien}. See Appendix \ref{app:lem} for our proof by induction.  

\begin{lem}\label{g}
Let $g(\alpha)$ be the number of $\sigma\in S_\alpha$ that saturate the inequality in Lemma \ref{sumcycle}. Then $g(\alpha) = \cat_\alpha := 2\alpha!/\alpha!(\alpha+1)!=\frac{1}{\alpha+1}\binom{2\alpha}{\alpha}$, i.e.,~the $\alpha$-th Catalan number.
\end{lem}
This result follows from Lemmas A.4 and A.5 of \cite{lancien}. Such permutations lie on the geodesic from identity to $\tau$.  The above lemmas guarantee that the gap between the Haar-averaged R\'enyi entropies and the maximum value is independent of the system size, as will become clear shortly. We note that Catalan numbers frequently occur in counting problems. The first few Catalan numbers are $1, 1, 2, 5, 14, 42, 132, 429, 1430, 4862, 16796, \cdots$. Some useful bounds on the Catalan numbers are derived in Appendix \ref{app:cat}.

\begin{cor}\label{cyclecor}
	$ \xi(\sigma)+\xi(\sigma\tau_{\alpha,s})\leq s\alpha+s$ for all $\sigma\in S_{s\alpha}$. The number of $\sigma\in S_{s\alpha}$ that saturate the inequality is $g(\alpha,s) =g(\alpha)^s= \cat_\alpha^s=\frac{1}{(\alpha+1)^s}\binom{2\alpha}{\alpha}^s$.
\end{cor}


We also need the large $d$ asymptotic behaviors of the Weingarten function:
\begin{lem}[Asymptotics of Wg \cite{collins,Collins2006}]\label{asym}
Given $\sigma\in S_\alpha$ with cycle decomposition $\sigma = C_1\cdots C_k$. Let $ |\sigma| $ be the minimal number of factors needed to write $\sigma$ as a product of transpositions. The M\"obius function of $\sigma$ is defined by
\begin{equation}
    \mathrm{Moeb}(\sigma):= \prod_{i=1}^k (-1)^{|C_i|}\cat_{|C_i|},
\end{equation}
where $\cat_n$ is the $n$-th Catalan number (defined in Lemma \ref{g}). (Note that $|C_i|$ here is often replaced by $|C_i|-1$ in literature, where $|\cdot|$ means the length of the cycle.)
Then, in the large $d$ limit, the Weingarten function has the asymptotic behavior
\begin{equation}
d^{\alpha+|\sigma|}\mathrm{Wg}(d,\sigma)=\mathrm{Moeb}(\sigma)+O(d^{-2}).
\end{equation}
\end{lem}
 \begin{cor}\label{wgcor}
We mainly need to distinguish the following two cases:
\begin{itemize}
    \item $ \sigma=I $: $|\sigma|=0$ and $\mathrm{Moeb}(\sigma)=1$, thus $\mathrm{Wg}(d,I)=d^{-\alpha}+O(d^{-(\alpha+2)}) $;
    \item $ \sigma\neq I $: $ |\sigma|\geq 1 $, thus $ \mathrm{Wg}(d,\sigma)=O(d^{-(\alpha+|\sigma|)})=O(d^{-(\alpha+1)}) $.
\end{itemize}
\end{cor}
Some bounds on the M\"obius function are derived in Appendix \ref{app:moeb}.

Now we are equipped to derive the asymptotic behaviors of the Haar-averaged traces, $\int{\rm d}U \mathrm{tr}\left\{\rho_{AC}^\alpha\right\}$:
\begin{thm}\label{thm:asym}
For equal partitions ($d_A=d_B=d_C=d_D = \sqrt{d}$), in the large $d$ limit,
\begin{align}
\int{\rm d}U \mathrm{tr}\left\{\rho_{AC}^\alpha\right\}&=\cat_\alpha d^{1-\alpha}(1+O(d^{-1})), \label{expsumcycle}\\
\int{\rm d}U \left(\mathrm{tr}\left\{\rho_{AC}^\alpha\right\}\right)^s
&= \cat_\alpha^s d^{(1-\alpha)s}(1+O(d^{-1})),\label{expsumcycle2}
\end{align}
\end{thm}
\begin{proof}
Starting from Eq.\ (\ref{general}),  Theorem \ref{thm:trace}
\begin{align}
&\int{\rm d}U \mathrm{tr}\left\{\rho_{AC}^\alpha\right\}\nonumber\\=&\frac{1}{d^{\alpha}}\sum_{\sigma}(d_Ad_C)^{\xi(\sigma\tau)}(d_Bd_D)^{\xi(\sigma)}\mathrm{Wg}(d,I)+\frac{1}{d^{\alpha }}\sum_{\sigma\neq \gamma}d_A^{\xi(\sigma\tau)}d_B^{\xi(\sigma)}d_C^{\xi(\gamma\tau)}d_D^{\xi(\gamma)}\mathrm{Wg}(d,\sigma\gamma^{-1})\\
=&\frac{1}{d^{\alpha}}\sum_{\sigma}d^{\xi(\sigma\tau)+\xi(\sigma)}\mathrm{Wg}(d,I)+\frac{1}{d^{\alpha}}\sum_{\sigma}d^{(\xi(\sigma\tau)+\xi(\sigma))/2}\sum_{\gamma\neq\sigma}d^{(\xi(\gamma\tau)+\xi(\gamma))/2}\mathrm{Wg}(d,\sigma\gamma^{-1})\\
=&\sum_{\sigma}d^{\xi(\sigma\tau)+\xi(\sigma)}(d^{-2\alpha }+O(d^{-(2\alpha+2)}))+\sum_{\sigma}d^{(\xi(\sigma\tau)+\xi(\sigma))/2}\sum_{\gamma\neq\sigma}d^{(\xi(\gamma\tau)+\xi(\gamma))/2}O(d^{-(2\alpha+1)})\\
=&\cat_\alpha d^{1-\alpha}(1+O(d^{-1})), \label{expsumcycle11}
\end{align}
where the second line follows from the equal bipartition assumption, the third line follows from Lemma \ref{asym} and Corollary \ref{wgcor}, and the fourth line follows from Lemmas \ref{sumcycle}, \ref{g} and some simple scaling analysis. 
Similarly, the asymptotic behavior of 
 $\int{\rm d}U \left(\mathrm{tr}\left\{\rho_{AC}^\alpha\right\}\right)^s$ follows from
by Corollary \ref{cyclecor}.
\end{proof} 
 

\paragraph{$s>0$ entropies}

The calculations of $s>0$ entropies (e.g. Tsallis) are straightforward, since the term $(\mathrm{tr}\{\rho^\alpha\})^s$ linearly appears in the definition. 
By Theorem \ref{thm:asym}, for positive integers $\alpha,s$:
\begin{equation}
\int{\rm d}U S^{(\alpha)}_s(\rho_{AC}) = \frac{1}{s(1-\alpha)}\left(\int{\rm d}U (\mathrm{tr}\left\{\rho_{AC}^\alpha\right\})^s-1\right) = \frac{1-\cat_\alpha^s d^{(1-\alpha)s}(1+O(d^{-1}))}{s(\alpha-1)}. 
\end{equation}
Notice that the maximum value of $S^{(\alpha)}_s$ for a $d$-dimensional state is (achieved by the maximally mixed state $I/d$)
\begin{equation}\label{eq:smax}
S^{(\alpha)}_s(I/d) = \frac{1-d^{(1-\alpha)s}}{s(\alpha-1)}.
\end{equation}
So we see a gap between the Haar-averaged and the maximal value:
\begin{equation}
\bar\Delta S^{(\alpha)}_s:=S^{(\alpha)}_s(I/d)-\int{\rm d}U S^{(\alpha)}_s(\rho_{AC})  = \frac{\cat_\alpha^s(1+O(d^{-1}))-1}{s(\alpha -1)}d^{(1-\alpha)s}, 
\end{equation}
which is vanishingly small in $d$.

As mentioned above, $s>0$ entropies are less ideal than R\'enyi entropies for our study since they do not exhibit the three nice properties. Here we elaborate on the resulting problems one by one:
\begin{enumerate}
\item We see from Eq.~(\ref{eq:smax}) that the roof (maximally mixed) values of $s>0$ entropies vary with the order $\alpha$.  Therefore, it does not make much sense to compare $s>0$ entropies of different orders $\alpha$ or with the roof value, on which our entropic characterization of scrambling and randomness complexities and several other arguments rely. 
\item The $s>0$ entropies are not even additive on maximally mixed states. So the derived quantities of mutual information and tripartite information in terms of $s>0$ entropies do not make good sense. 
Recall that all partitions are in the maximally mixed state $I/\sqrt{d}$. However, the generalized mutual information $\int{\rm d}U I^{(\alpha)}_s(A:C)$ given by $I^{(\alpha)}_s(A:C) := S^{(\alpha)}_s(A)+S^{(\alpha)}_s(C)-S^{(\alpha)}_s(AC)$ is not directly given by $\bar\Delta S^{(\alpha)}_s$. Define
\begin{equation}
\delta^{(\alpha)}_s := 2S^{(\alpha)}_s(I/\sqrt{d})- S^{(\alpha)}_s(I/d) = \frac{1}{s(\alpha-1)} (1-d^{(1-\alpha)s/2+1}+d^{(1-\alpha)s}).
\end{equation}
then
\begin{equation}
\int{\rm d}U I^{(\alpha)}_s(A:C) = 2S^{(\alpha)}_s(I/\sqrt{d})-\int{\rm d}U S^{(\alpha)}_s(\rho_{AC}) = \delta^{(\alpha)}_s+\bar\Delta S^{(\alpha)}_s,
\end{equation}
which is dominated by the irrelevant $\delta^{(\alpha)}_s$ ($\bar\Delta S_s^{(\alpha)}$ is vanishingly small). 

\item The characteristic function for $s\geq 1$ entropies are not convex (linear for Tsallis). Although Theorem \ref{thm:asym} enables us to directly calculate the Haar-averaged $s>1$ entanglement entropies, the nonconvexity prevents us from using Jensen's inequality to lower bound their design-averaged values.
\end{enumerate}

\paragraph{R\'enyi entropy} 
Now we analyze the behaviors of the R\'enyi entropies, the $s\rightarrow 0$ limit.
Compared to $s>0$ entropies, the calculations of R\'enyi are trickier because of the logarithm, which nevertheless directly leads to the desirable properties---constant roof value, additivity, and convexity.
We are able to establish the following result:
\begin{thm}
In the large $d$ limit,
\begin{equation}
\int{\rm d}U S_R^{(\alpha)}(\rho_{AC})\geq \log d-O(1),
\end{equation}
\end{thm}
\begin{proof}
By definition,
\begin{equation}
\int{\rm d}U S_R^{(\alpha)}(\rho_{AC}) = \int{\rm d}U f^{(\alpha)}_R(\mathrm{tr}\left\{\rho_{AC}^\alpha\right\}),
\end{equation}
where
\begin{equation}
f^{(\alpha)}_R(x)=\frac{1}{1-\alpha}\log x
\end{equation}
is the characteristic function.
Since
\begin{equation}
\frac{d^2{f^{(\alpha)}_R}(x)}{dx^2}=\frac{1}{(\alpha-1)\ln 2}\frac{1}{x^2}\geq 0
\end{equation}
when $\alpha>1$, $ f^{(\alpha)}_R(x) $ is convex.
So
\begin{equation}
\int{\rm d}U S_R^{(\alpha)}(\rho_{AC})\geq f^{(\alpha)}_R\left(\int{\rm d}U\mathrm{tr}\left\{\rho_{AC}^\alpha\right\}\right)
\end{equation}
by Jensen's inequality. We note that this Jensen's lower bound due to convexity ($\bbE f_R\geq f_R\bbE$) will be repeatedly used to establish bounds for R\'enyi entropies.
Then according to Eq.\ (\ref{expsumcycle}),
\begin{align}
&f^{(\alpha)}_R\left(\int{\rm d}U\mathrm{tr}\left\{\rho_{AC}^\alpha\right\}\right)=\frac{1}{1-\alpha}\log\left(\int{\rm d}U\mathrm{tr}\left\{\rho_{AC}^\alpha\right\}\right)
=\frac{1}{1-\alpha}\log\left(\cat_\alpha d^{1-\alpha}(1+O(d^{-1}))\right)\nonumber\\
=&\log d-\frac{1}{\alpha-1}\log \cat_\alpha+O(d^{-1})
\geq \log d-\frac{2\alpha}{\alpha-1}+O(d^{-1}). \label{eq:o1}
\end{align}
Notice that the Cycle Lemma guarantees that the leading correction term (the second term) is independent of $d$ asymptotically.
In fact
\begin{equation}
\frac{1}{\alpha-1}\log \cat_\alpha\leq\frac{2\alpha}{\alpha-1}\leq 4 \quad\forall \alpha\geq 2.
\end{equation}
In conclusion, in the limit of large $d$, we have
\begin{equation}
\int{\rm d}U S_R^{(\alpha)}(\rho_{AC})\geq f^{(\alpha)}_R\left(\int{\rm d}U\mathrm{tr}\left\{\rho_{AC}^\alpha\right\}\right)=\log d-O(1).
\end{equation}
\end{proof}
So the gap between the Haar-averaged and maximal value of $ S_R $ (the ``residual entropy'') is
\begin{equation}
\bar\Delta S_R^{(\alpha)}:=\log d-\int{\rm d}U S_R^{(\alpha)}(\rho_{AC})\leq O(1).
\end{equation}
That is, the average R\'enyi entanglement entropies of the Haar measure are only bounded by a constant from the maximum.  Recall the discussion in Sec.~2.1.2: this $O(1)$ gap holds for non-equal partitions as well.
The result implies that a Haar random unitary typically has almost maximal R\'enyi entanglement entropies for any partition.  
Rigorous probabilistic arguments require more careful analysis using concentration inequalities, which we leave for future work.

Now consider the R\'enyi mutual information and tripartite information based on the entanglement entropy results. First, we can directly obtain
\begin{equation}
    \int{\rm d}U I^{(\alpha)}_R(A:C) = \log d - \int{\rm d}U S_R^{(\alpha)}(\rho_{AC})  \leq O(1),
\end{equation}
which is equal to $\bar\Delta S_R^{(\alpha)}$ by additivity. The results hold similarly for $AD$. That is, the R\'enyi mutual information between any two local regions of the input and output is vanishingly small compared to the system size. On the other hand, for any partition size, notice that
\begin{align}
I_R^{(\alpha)}(A:CD) &= S_R^{(\alpha)}(A) + S_R^{(\alpha)}(CD) - S_R^{(\alpha)}(ACD) \\
&= S_R^{(\alpha)}(A) + S_R^{(\alpha)}(CD) - S_R^{(\alpha)}(B) \\
&= \log d_A+\log d-\log d_B \\
&= 2\log d_A,
\end{align}
where the second line follows from $S_R^{(\alpha)}(ACD) = S_R^{(\alpha)}(B)$ since the whole Choi state is pure, the third line follows from the fact that the three subregions involved are maximally mixed, and the fourth line follows from that $d_A d_B = d$.  Under the equal partition assumption, $I_R^{(\alpha)}(A:CD) = \log d$.
This is consistent with the fact that all information of $A$ is kept in the whole output $CD$ because of unitarity. As a result:
\begin{equation}
-{I_3}_R^{(\alpha)}(A:C:D):= I_R^{(\alpha)}(A:CD) - I_R^{(\alpha)}(A:C) - I_R^{(\alpha)}(A:D) \geq \log d-O(1),
\end{equation}
by plugging in all relevant terms. So the negative R\'enyi tripartite information of Haar scrambling is indeed close to the maximum. However, we note that the R\'enyi-$\alpha$ entropy is not subadditive except for $\alpha=1$, thus $-{I_3}_R^{(\alpha)}(A:C:D)$ is not necessarily nonnegative. A weaker form of subadditivity of R\'enyi entropies is given in Appendix \ref{app:renyisub}.

\subsubsection{Non-asymptotic bounds}
Here we prove some explicit bounds on the Haar-averaged trace, R\'enyi entropies, and in particular the min entropy, in the non-asymptotic regime.  These bounds sharpen the asymptotic results.
Many useful lemmas are proved in the Appendices.

\paragraph{Trace and R\'enyi entropies}
We directly put the results of trace and R\'enyi entropies together.  We need the following refined cycle lemma:
\begin{lem}\label{lem:sumBound}
	Suppose $q:=\alpha^3/(32d_B^2)<1$, and $d_A\leq d_B$.	
	Then
	\begin{align}
	&\sum_{\sigma\in S_\alpha}d_A^{\xi(\sigma\tau)}d_B^{\xi(\sigma)}\leq h(q)\cat_\alpha d_Ad_B^{\alpha}
\leq \frac{4^\alpha h(q)}{\sqrt{\pi}\alpha^{3/2}}d_Ad_B^\alpha.
	\end{align}
	 where $h(q)=1+2q/[3(1-q)]$. 
\end{lem}
\begin{proof}
	Define $c_{\delta, \alpha}$ as the number of permutations in $S_\alpha$ with genus $\delta$, that is,
	\begin{equation}
	c_{\delta, \alpha}:=|\{\sigma\in S_\alpha|  \xi(\sigma)+ \xi(\sigma\tau)=\alpha+1 -2\delta \}|. 
	\end{equation}
	Note that $c_{0,\alpha}$ is the Catalan number $\cat_\alpha$ by Lemma \ref{g}.
	Then 
	\begin{equation}
	c_{\delta, \alpha}\leq \frac{2}{3}\left(\frac{\alpha^3}{32}\right)^\delta \cat_\alpha,
	\end{equation}
	according to Lemma~\ref{lem:NGpermutationT} in Appendix \ref{app:genus}. As a consequence of this inequality and the assumption $d_A\leq d_B$,
	\begin{align}
	&\sum_{\sigma\in S_\alpha}d_A^{\xi(\sigma\tau)}d_B^{\xi(\sigma)}\leq \sum_{\sigma\in S_\alpha}d_Ad_B^{\xi(\sigma)+\xi(\sigma\tau)-1}=
	\sum_{\delta=0}^{(n-1)/2}c_{\delta,\alpha}d_Ad_B^{\alpha-2\delta}=c_{0,\alpha}d_Ad_B^{\alpha} \sum_{\delta=0}^{(n-1)/2}\frac{c_{\delta,\alpha}}{\cat_\alpha}d_B^{-2\delta}\nonumber\\
	\leq& c_{0,\alpha}d_Ad_B^{\alpha}\left[1+ \frac{2}{3}\sum_{\delta=1}^{(n-1)/2}  \left(\frac{\alpha^3}{32d_B^2}\right)^{\delta}\right]\leq  \cat_\alpha d_Ad_B^{\alpha}\left[1+ \frac{2}{3}\sum_{\delta=1}^{\infty}  q^{\delta}\right]
	=\cat_\alpha d_Ad_B^{\alpha}\left[1+ \frac{2q}{3(1-q)}\right]\nonumber\\
	=&h(q)\cat_\alpha d_Ad_B^{\alpha}\leq \frac{4^\alpha h(q)}{\sqrt{\pi}\alpha^{3/2}}d_Ad_B^\alpha,
	\end{align}
	where the last inequality follows from  Lemma~\ref{lem:CatalanBound} in Appendix \ref{app:cat}, which sets an upper bound on the Catalan numbers. 
\end{proof}

In the following we still assume equal partitions so that $d_{AC} = d$ for simplicity. Recall that for generic partitions the residual entropy cannot be larger anyway.    By Lemma \ref{lem:sumBound}, we can obtain the following non-asymptotic bounds for the Haar integrals of traces and R\'enyi entanglement entropies:
\begin{thm}\label{thm:UniNonasymptotic}
	Suppose $d>\sqrt{6}\alpha^{7/4}$. Then
\begin{align}
\int{\rm d}U \mathrm{tr}\{\rho_{AC}^{\alpha}\}&\leq \frac{a_\alpha \cat_\alpha d^{1-\alpha}}{8}\left(1+\frac{2q}{3(1-q)}\right)\left(7+\cosh\frac{2\alpha(\alpha-1)}{d}\right),\label{finited1}\\
\int {\rm d}U S_R^{(\alpha)}(\rho_{AC})&\geq \log d-\frac{\log \cat_\alpha}{\alpha-1}-\frac{\log \left[\frac{a_\alpha }{8}\left(1+\frac{2q}{3(1-q)}\right)\left(7+\cosh\frac{2\alpha(\alpha-1)}{d}\right)\right]}{\alpha-1}, \label{finited2}
\end{align}	
where $a_\alpha:=\frac{1}{1-\frac{6\alpha^{7/2}}{d^2}}$.
\end{thm}
\begin{proof}
By Eq.\ (\ref{general}) (Theorem \ref{thm:trace}):
    \begin{eqnarray}
\int{\rm d}U \mathrm{tr}\{\rho_{AC}^{\alpha}\} &= &\sum_{\sigma, \gamma\in S_\alpha}d^{[\xi(\sigma\tau)+\xi(\sigma)+\xi(\gamma\tau)+\xi(\gamma)]/2}
\wg(d,\sigma\gamma^{-1})\\
&=&\sum_{\zeta\in S_\alpha}\left[
\sum_{ \gamma\in S_\alpha}d^{[\xi(\zeta\gamma\tau)+\xi(\zeta\gamma)+\xi(\gamma\tau)+\xi(\gamma)]/2}
	\wg(d,\zeta)\right]\\	
&\leq& \sum_{\zeta\in A_\alpha}\left[
\sum_{ \gamma\in S_\alpha}d^{[\xi(\zeta\gamma\tau)+\xi(\zeta\gamma)+\xi(\gamma\tau)+\xi(\gamma)]/2}
\wg(d,\zeta)\right]\\&\leq& \sum_{\zeta \in A_\alpha}
\sum_{ \gamma\in S_\alpha}
d^{\xi(\gamma\tau)+\xi(\gamma)}
\wg(d,\zeta)\\	
&\leq& \cat_\alpha d^{\alpha+1}\left(1+\frac{2q}{3(1-q)}\right)  \sum_{\zeta \in A_\alpha} \wg(d,\zeta)\\
&\leq&\frac{a_\alpha\cat_\alpha d}{8}\left(1+\frac{2q}{3(1-q)}\right)\left(7+\cosh\frac{2\alpha(\alpha-1)}{d}\right),\label{cosh}
\end{eqnarray}
where $A_\alpha$ is the set of even permutations, i.e.\ the alternating group. The first inequality follows from the fact that $\wg(d,\zeta)$ is negative when $\zeta$ is an odd permutation, the second inequality follows from the Cauchy-Schwarz inequality, noting that $\sum_{\gamma\in S_\alpha}d^{\xi(\zeta\gamma\tau)+\xi(\zeta\gamma)}=\sum_{\gamma\in S_\alpha}d^{\xi(\gamma\tau)+\xi(\gamma)}$, the third inequality follows from Lemma~\ref{lem:sumBound}, and the last inequality follows from Lemma~\ref{lem:WgSumBound} in Appendix \ref{app:weingarten}. By plugging Eq.\ (\ref{cosh}) into Eq.\ (\ref{general}), we immediately obtain the trace result Eq.\ (\ref{finited1}).  The R\'enyi result Eq.\ (\ref{finited2}) then follows from Jensen's inequality. 
\end{proof}
We see that the leading terms indeed match the asymptotic results. The overall observation is similar: the Haar integrals of R\'enyi entanglement entropies are very close to the maximum for sufficiently large $d$.


To gain intuition, we compute $\int{\rm d}U \mathrm{tr}\{\rho_{AC}^{\alpha}\}$ for $\alpha=2,3$ based on the explicit formulas for Weingarten functions. When $\sigma\in S_2$,
\begin{equation}
\wg(d,\sigma)=\begin{cases}
\frac{1}{d^2-1}  &\sigma =I,\\
-\frac{1}{d(d^2-1)} & \sigma=(1~2).
\end{cases}
\end{equation}
When  $\sigma\in S_3$,
\begin{equation}
\wg(d,\sigma)=\frac{1}{d(d^2-1)(d^2-4)}\begin{cases}
d^2-2 &\sigma= I,\\
-d & \sigma=(1\; 2),\\
2 & \sigma=(1\; 2\;3).
\end{cases}
\end{equation}
Therefore,
\begin{align}
\int{\rm d}U \mathrm{tr}\{\rho_{AC}^{2}\}&=\frac{2}{d+1}\leq \frac{2}{d},\label{eq:SecMomentChoi}\\
\int{\rm d}U \mathrm{tr}\{\rho_{AC}^{3}\}&=\frac{5d^3-7d^2-6d+2}{d^2(d+1)(d^2-4)}\leq \frac{5}{d^2}.  \label{eq:ThirdMomentChoi}
\end{align}

\paragraph{Min entropy}
The results so far only directly apply to positive integer $\alpha$. The min entanglement entropy, which corresponds to the special limit $\alpha\rightarrow\infty$, plays a crucial role in our framework of scrambling complexities. 
Now we examine the Haar integral of the min entanglement entropy.

\begin{thm}\label{thm:AveMinEntropy}
\begin{align}
\int{\rm d}U \|\rho_{AC}\|&\leq \frac{m_d}{d}, \label{eq:AveNormChoi}    \\
\int{\rm d}U S_{\min}(\rho_{AC})&\geq \log d-\log m_d.  \label{eq:AveMinEntropyChoi}   
\end{align}
where $m_d:=\min\{7,\; 4 (8\sqrt{d})^{1/\sqrt{d}}\}$.
\end{thm}

\begin{proof}
Suppose $d\leq 70$. Then we have
\begin{equation}\label{7}
d \int{\rm d}U \|\rho_{AC}\|\leq d\left(\int{\rm d}U\tr \{\rho_{AC}^{3}\}\right)^{1/3}\leq d\left({\frac{5d^3-7d^2-6d+2}{d^2(d+1)(d^2-4)}}\right)^{1/3}\leq 7.
\end{equation}
Now suppose $d\geq 50$. Let $\alpha=\lceil \sqrt{d}/2\rceil$.
Then
\begin{equation}
\frac{6\alpha^{7/2}}{d^2}\leq \frac{2}{5},\quad q=\frac{\alpha^3}{32d^2}\leq \frac{1}{210},\quad \frac{2\alpha(\alpha-1)}{d}\leq \frac{2}{3},
\end{equation}
so that 
\begin{equation}
a_\alpha\leq \frac{5}{3},\quad h(q)\leq \frac{301}{300},\quad \frac{1}{8}\left(7+\cosh\frac{2\alpha(\alpha-1)}{d}\right)\leq \frac{103}{100}.
\end{equation}
Consequently, 
\begin{equation}\label{eq:AveAlphaMomentProof}
\int{\rm d}U \tr \{\rho_{AC}^{\alpha}\}\leq \left[ \frac{a_\alpha}{8\sqrt{\pi}}\left(1+\frac{2q}{3(1-q)}\right)\left(7+\cosh\frac{2\alpha(\alpha-1)}{d}\right)\right] \frac{4^\alpha d^{1-\alpha}}{\alpha^{3/2}}
\leq \frac{ 4^\alpha d^{1-\alpha}}{\alpha^{3/2}},
\end{equation}
and thus 
\begin{equation}
d \int{\rm d}U \|\rho_{AC}\|\leq d\left(\int{\rm d}U\tr \{\rho_{AC}^{\alpha}\}\right)^{1/\alpha}\leq 4 \Bigl(\frac{d}{\alpha^{3/2}}\Bigr)^{1/\alpha}\leq 4 \Bigl(\frac{d}{(\sqrt{d}/2)^{3/2}}\Bigr)^{2/\sqrt{d}}=
4 (8\sqrt{d})^{1/\sqrt{d}}.
\end{equation}
The proof of Eq.\ \eqref{eq:AveNormChoi} 
is completed by observing that $4 (8\sqrt{d})^{1/\sqrt{d}}> 7$ when $50\leq d\leq 52$ and  $4 (8\sqrt{d})^{1/\sqrt{d}}< 7$  when $53\leq d\leq 70$.
Eq.\ \eqref{eq:AveMinEntropyChoi} then follows from the convexity of $-\log$. We note that slightly lower $m_d$ can in principle be obtained by computing to higher orders in Eq.\ (\ref{7}), which is nevertheless not important for the main idea.
\end{proof}

As $d$ gets large, $m_d$ approaches the limit 4, and  $\log m_d$ approaches the limit 2. As an implication of Lemma~\ref{lem:EntropyGap} in Appendix \ref{app:renyisub}, Theorem~\ref{thm:AveMinEntropy} with $d$ replaced by $d_{AC}$ also holds when the four subregions have different dimensions, as long as $d_A, d_C\leq \sqrt{d}$ and $d_{AC}\leq d$. The same remark also applies to Theorem~\ref{thm:AveMinEntropyLog} below.

Note that the above results essentially confirm the conjecture in \cite{karol2} that a Haar random unitary $U$ and its reshuffled matrix $U^R$ are asymptotically free, and the conjecture in \cite{karol1} (based on extensive numerical evidence) that $U^R$ converges to the Ginibre ensemble (of random non-Hermitian matrices) so that their moments will be asymptotically given by the Catalan numbers and the distribution of their spectra will be described by the
famous Marchenko-Pastur distribution.

\subsection{Unitary designs and their approximates}
\subsubsection{Average over unitary designs}
Now we state a key observation: the Haar integral of $\mathrm{tr}\{\rho_{AC}^\alpha\}$, the defining term of $\alpha$ entropies, only uses the first $\alpha$ moments of the Haar measure. In other words, pseudorandom unitary $\alpha$-designs are indistinguishable from Haar random by $\mathrm{tr}\{\rho_{AC}^\alpha\}$. More explicitly, let $\nu_\alpha$ be a unitary $\alpha$-design ensemble, then we have
\begin{equation}
    \int{\rm d}U U_{i_1j_1}\cdots U_{i_\alpha j_\alpha}U^*_{i^\prime_1j^\prime_1}\cdots U^*_{i^\prime_\alpha j^\prime_\alpha} = \mathbb{E}_{\nu_\alpha} \left[U_{i_1j_1}\cdots U_{i_\alpha j_\alpha}U^*_{i^\prime_1j^\prime_1}\cdots U^*_{i^\prime_\alpha j^\prime_\alpha}\right]
\end{equation}
by definition.
Therefore, all Haar integrals of $\mathrm{tr}\{\rho_{AC}^\alpha\}$ from Sec.~\ref{sec:haarint} (those derived from Eqs.\ (\ref{general}) and (\ref{expsumcycle})) directly carry over to $ \alpha $-designs. 

This observation is the essential basis for the order correspondence results and in turn the idea that $\alpha$ entropies can generically diagnose whether a scrambler is locally indistinguishable from random dynamics as powerful as $\alpha$-designs. The Haar-averaged Tsallis-$\alpha$ entropies ($s=1$) are exactly saturated by $\alpha$-designs due to the linearity in $\mathrm{tr}\{\rho^\alpha\}$. However, as mentioned, we cannot make analogous arguments for $s>1$: the exact saturation requires $f_s$ and the Haar integral to commute asymptotically, which is not known to hold; and the lower bound following from Jensen's inequality does not hold since $f_{s>1}$ becomes concave. In contrast, the R\'enyi entropies can be lower bounded because of the convexity. Due to the importance of the R\'enyi entropies, we state this result as a theorem:
\begin{lem}
\begin{equation}
    \mathbb{E}_{\nu_\alpha}\left[S_R^{(\alpha)}\left(\rho_{AC}\right)\right] 
\geq f^{(\alpha)}_R\left(\int{\rm d}U\mathrm{tr}\left\{\rho_{AC}^\alpha\right\}\right) = \frac{1}{1-\alpha}\log\left(\int{\rm d}U\mathrm{tr}\left\{\rho_{AC}^\alpha\right\}\right).
\end{equation}

\end{lem}
\begin{proof}
\begin{equation}
    \mathbb{E}_{\nu_\alpha}\left[S_R^{(\alpha)}\left(\rho_{AC}\right)\right]
= \mathbb{E}_{\nu_\alpha}\left[f_R^{(\alpha)}\left(\mathrm{tr}\left\{\rho_{AC}^\alpha\right\}\right)\right]\geq f^{(\alpha)}_R(\mathbb{E}_{\nu_\alpha}[\mathrm{tr}\{\rho_{AC}^\alpha\}])=f^{(\alpha)}_R\left(\int{\rm d}U\mathrm{tr}\left\{\rho_{AC}^\alpha\right\}\right),
\end{equation}
where the inequality follows from Jensen's inequality, and the last equality follows from the fact that $\nu_\alpha$ is an $\alpha$-design. 
\end{proof}
The lemma enables us to use the Haar integrals of traces to lower bound the design-averaged R\'enyi entanglement entropies in all dimensions. 
By combining this lemma and Theorem \ref{thm:asym}, we directly see that the $O(1)$ upper bound on the residual R\'enyi-$\alpha$ entropy still holds: 
\begin{thm}
In the large $d$ limit,
\begin{equation}
\mathbb{E}_{\nu_\alpha}\left[S_R^{(\alpha)}\left(\rho_{AC}\right)\right]
\geq \log d-O(1).
\end{equation}
\end{thm}
This is a key result of this work.  We conclude that R\'enyi-$\alpha$ entanglement entropies are very likely to be almost maximal when sampling from unitary $\alpha$-designs, as well as from the Haar measure. This result establishes the correspondence between the order of R\'enyi entropy and the order of designs, and lays the basis for the notion of entropic scrambling complexities.
The non-asymptotic bound in Theorem \ref{thm:UniNonasymptotic} carries over in a similar fashion:
\begin{thm}\label{thm:UniNonasymptotic2}
	Suppose $d>\sqrt{6}\alpha^{7/4}$. Then
\begin{equation}
\mathbb{E}_{\nu_\alpha}\left[S_R^{(\alpha)}\left(\rho_{AC}\right)\right]
\geq \log d-\frac{\log \cat_\alpha}{\alpha-1}-\frac{\log \left[\frac{a_\alpha }{8}\left(1+\frac{2q}{3(1-q)}\right)\left(7+\cosh\frac{2\alpha(\alpha-1)}{d}\right)\right]}{\alpha-1}, \label{finited2design}
\end{equation}	
where $a_\alpha:=\frac{1}{1-\frac{6\alpha^{7/2}}{d^2}}$.
\end{thm}

Later we analyze the min entanglement entropy of designs in particular, which leads to another main result.


\subsubsection{Error analysis: approximate unitary designs}\label{sssection:error}
The above analysis is based on exact unitary designs, but in most contexts we need to deal with the approximate versions of them.  
How robust or sensitive are these results under small deviations from exact unitary designs?   One would expect ensembles that are very close to exact unitary $\alpha$-designs to maintain near-maximal R\'enyi-$\alpha$ entanglement entropies. 
A subtlety is that different ways of measuring the deviation may lead to inequivalent definitions of approximate unitary designs, in contrast to the exact case.
Here we discuss the deviation bounds for two commonly used definitions of approximate unitary designs, based on polynomials and frame operators respectively.
This error analysis will be directly useful for e.g.~relating the entropic scrambling complexities to circuit depth. 

First, the canonical definition of unitary designs by polynomials leads to the following measure of deviation: 
\begin{defn}[m-approximate unitary designs \cite{brandaoharrow2}]
An ensemble $\nu$ is an $\epsilon$-m-approximate unitary $t$-design (``m'' represents monomial) if
\begin{equation}
\left|\int{\rm d}U q^{(k)}(U) - \mathbb{E}_{\nu} \left[q^{(k)}(U)\right]\right| \leq \epsilon \quad \forall q^{(k)},k\leq t.
\end{equation}
where $q^{(k)}(U) = U_{i_1j_1}\cdots U_{i_kj_k}U^*_{m_1n_1}\cdots U^*_{m_kn_k}$ is a monomial of degree $k$ both in the entries of $U$ and in their complex conjugates.
\end{defn}
Note that the bound is on each monomial with unit constant factor, otherwise the difference can be arbitrarily amplified by including more terms or changing the constant.
\begin{thm}
Let $\omega_{\alpha}$ be an $\epsilon$-m-approximate unitary $\alpha$-design. Then
\begin{align}
\mathbb{E}_{\omega_\alpha}[\mathrm{tr}\left\{\rho_{AC}^\alpha\right\}]&\leq\int{\rm d}U\mathrm{tr}\left\{\rho_{AC}^\alpha\right\}+d^\alpha\epsilon,\\
\mathbb{E}_{\omega_\alpha}\left[S_R^{(\alpha)}\left(\rho_{AC}\right)\right]
&\geq\frac{1}{1-\alpha}\log\left(\int{\rm d}U\mathrm{tr}\left\{\rho_{AC}^\alpha\right\}+d^{\alpha}\epsilon\right).
\end{align}
In the large $d$ limit,
\begin{equation}
\mathbb{E}_{\omega_\alpha}[S_R^{(\alpha)}\left(\rho_{AC}\right)]  \geq
\log{d} - O(1) - \frac{1}{(\alpha-1)\cat_\alpha\ln 2}d^{2\alpha-1}\epsilon\left(1+O\left(d^{-1}\right)\right).   \label{mlarged}
\end{equation}
\end{thm}
\begin{proof}
\begin{equation}\label{eq:trerror}
\mathbb{E}_{\omega_\alpha}[\mathrm{tr}\left\{\rho_{AC}^\alpha\right\}]-\int{\rm d}U\mathrm{tr}\left\{\rho_{AC}^\alpha\right\}\leq\left|\int{\rm d}U\mathrm{tr}\left\{\rho_{AC}^\alpha\right\}-\mathbb{E}_{\omega_\alpha}[\mathrm{tr}\left\{\rho_{AC}^\alpha\right\}]\right|\leq \frac{1}{d^\alpha}d^{2\alpha}\epsilon = d^{\alpha}\epsilon
\end{equation}
by triangle inequality, since $\mathrm{tr}\left\{\rho_{AC}^\alpha\right\}$ is the sum of $d^{2\alpha}$ monomials according to Eq.\ (\ref{tr}). 
Then 
\begin{align}
&\mathbb{E}_{\omega_\alpha}\left[S_R^{(\alpha)}\left(\rho_{AC}\right)\right]
\geq f_R^{(\alpha)}\left(\mathbb{E}_{\omega_\alpha}\left[\mathrm{tr}\left\{\rho_{AC}^\alpha\right\}\right]\right) \nonumber\\=& \frac{1}{1-\alpha}\log\mathbb{E}_{\omega_\alpha}[\mathrm{tr}\{\rho_{AC}^\alpha\}]\geq\frac{1}{1-\alpha}\log\left(\int{\rm d}U\mathrm{tr}\left\{\rho_{AC}^\alpha\right\}+d^{\alpha}\epsilon\right),
\end{align}
where the first inequality follows from Jensen's inequality, and the second inequality follows from Eq.\ (\ref{eq:trerror}) and the fact that $-\log$ is monotonically decreasing.
We can then use the $\int{\rm d}U\mathrm{tr}\left\{\rho_{AC}^\alpha\right\}$ results to analyze the perturbation.

Most importantly, in the large $d$ limit,
\begin{align}
&\frac{1}{1-\alpha}\log\left(\int{\rm d}U\mathrm{tr}\left\{\rho_{AC}^\alpha\right\}\right)-\mathbb{E}_{\omega_\alpha}\left[S_R^{(\alpha)}\left(\rho_{AC}\right)\right]
\leq -\frac{1}{1-\alpha}\log\left(1+\frac{d^{\alpha}\epsilon}{\int{\rm d}U\mathrm{tr}\left\{\rho_{AC}^\alpha\right\}}\right)  \\
\leq&\frac{1}{\alpha-1}\log\left(1+\frac{1}{\cat_\alpha}d^{2\alpha-1}\epsilon\left(1+O\left(d^{-1}\right)\right)\right)\\
\leq&\frac{1}{(\alpha-1)\cat_\alpha\ln 2}d^{2\alpha-1}\epsilon\left(1+O\left(d^{-1}\right)\right), \label{eq:renyierror}
\end{align}
where the second line follows from Eq.\ (\ref{expsumcycle11}) and the following analysis, and the third line follows from the inequality that $\ln(1+x)\leq x$ when $x>-1$. Then we directly obtain Eq.~(\ref{mlarged}), which says that the error in $S_R^{(\alpha)}(\rho_{AC})$ scales at most as $O(d^{2\alpha-1}\epsilon)$.
\end{proof}

Recall the other definition of exact designs by frame operators.
The deviation of an ensemble from a unitary $t$-design can also be quantified by a suitable norm of the deviation operator 
\begin{equation}
\Delta_t(\nu):=\mathcal{F}_t(\nu)-\mathcal{F}_t(\rmU(d)).
\end{equation}
The operator norm and trace norm of  $\Delta_t(\nu)$ are two common figures of merit. The latter choice is more convenient for the current study: 
\begin{defn}[FO-approximate unitary designs]
Ensemble $\nu$ is a $\lambda$-FO-approximate unitary $t$-design (FO represents frame operator) if
\begin{equation}
\|\Delta_t(\nu)\|_1\leq \lambda.
\end{equation}
\end{defn}

Note that this definition is very similar to the quantum tensor product expander (TPE) \cite{tpe}. TPEs conventionally use the operator norm, and the deviation operators relate to each other by partial transposes (like operators $X,Y$ in Eqs.~(\ref{x}), (\ref{y})).

Here we can directly use the operator form of local density operators derived earlier to do an error analysis of FO-approximate unitary designs. Let $\omega_\alpha$ be a  $\lambda$-FO-approximate unitary $\alpha$-design. 
We define $\tilde\Delta_\alpha$, and explicitly write out $\Delta_\alpha$:
\begin{align}
\tilde\Delta_\alpha(\omega_\alpha) &= \mathbb{E}_{\omega_\alpha}[(U\otimes U^\dagger)^{\otimes\alpha}] - \int{\rm d}U(U\otimes U^\dagger)^{\otimes\alpha},\\
\Delta_\alpha(\omega_\alpha) &= \mathbb{E}_{\omega_\alpha}[U^{\otimes\alpha}\otimes {U^\dagger}^{\otimes\alpha}] - \int{\rm d}UU^{\otimes\alpha}\otimes {U^\dagger}^{\otimes\alpha}.
\end{align}
\begin{thm}
Let $\omega_{\alpha}$ be a $\lambda$-FO-approximate unitary $\alpha$-design. Then
\begin{align}
\mathbb{E}_{\omega_\alpha}[\mathrm{tr}\left\{\rho_{AC}^\alpha\right\}]&\leq\int{\rm d}U\mathrm{tr}\left\{\rho_{AC}^\alpha\right\}+\frac{1}{d^{\alpha}}\lambda,\\
\mathbb{E}_{\omega_\alpha}[S_R^{(\alpha)}\left(\rho_{AC}\right)] & \geq
\frac{1}{1-\alpha}\log\left(\int{\rm d}U\mathrm{tr}\left\{\rho_{AC}^\alpha\right\}+\frac{1}{d^{\alpha}}\lambda\right).
\end{align}
In the large $d$ limit, 
\begin{equation}
\mathbb{E}_{\omega_\alpha}[S_R^{(\alpha)}\left(\rho_{AC}\right)]  \geq
\log{d} - O(1) - \frac{1}{(\alpha-1)\cat_\alpha\ln 2}d^{-1}\lambda\left(1+O\left(d^{-1}\right)\right).   \label{mlarged2}
\end{equation}
\end{thm}
\begin{proof}
\begin{align}
&\mathbb{E}_{\omega_\alpha}[\mathrm{tr}\left\{\rho_{AC}^\alpha\right\}] - \int{\rm d}U\mathrm{tr}\left\{\rho_{AC}^\alpha\right\} = \frac{1}{d^{\alpha}}\mathrm{tr}\{\tilde\Delta_\alpha(\omega_\alpha) Y_\alpha\} \leq \frac{1}{d^{\alpha}}\norm{\tilde\Delta_\alpha(\omega_\alpha)}_1\norm{Y_\alpha}\nonumber \\=& \frac{1}{d^{\alpha}}\norm{\tilde\Delta_\alpha(\omega_\alpha)}_1 = \frac{1}{d^{\alpha}}\norm{\Delta_\alpha(\omega_\alpha)}_1\leq\frac{1}{d^{\alpha}}\lambda,
\end{align}
where the first inequality follows from H\"older's inequality, and the second line follows from the unitarity of $Y_\alpha$ defined by Eq.~(\ref{y}).
The large $d$ limit calculation simply resembles the above.
\end{proof}

The essential difference between the m- and FO-approximate unitary designs is that the deviation is measured by different norms \cite{nakata}. Letting $\epsilon,\lambda=0$ recovers equivalent definitions of exact designs.   However, we can see from the asymptotic error bounds that they pose constraints of different strengths. 
The $\epsilon$-m-approximation condition is quite loose, in the sense that the deviation $\epsilon$ needs to be vanishingly small to guarantee that the residual entropy remains small. Or one could say that the R\'enyi entanglement entropy results can be very sensitive to this type of error.
In contrast, the $\lambda$-FO-approximation condition is more stringent and suitable:  the residual entropy remains $O(1)$ as long as $\lambda\leq O(d)$, which implies that the FO-approximation may be a more suitable scheme.

\subsection{Hierarchy of entropic scrambling complexities}
\subsubsection{Scrambling complexities by R\'enyi entanglement entropy}
As motivated in the introduction, we expect that there is a hierarchy of scrambling complexities that lie in between information scramblers and Haar random unitaries, with different levels of the hierarchy indexed by the order of unitary designs needed to mimic the scrambler. Our results in the above link the randomness complexity of designs and the maximality of R\'enyi entanglement entropies of the corresponding order.   This suggests that we can use the generic maximality of R\'enyi-$\alpha$ entanglement entropy as i) a necessary indicator of the resemblance to an $\alpha$-design, and ii) a diagnostic of the entanglement complexity of $\alpha$-designs, or ``$\alpha$-scrambling''. The basic logic is that if a supposedly random unitary dynamics does not produce nearly maximal R\'enyi-$\alpha$ entanglement entropy in all valid partitions, as $\alpha$-designs must do, then it is simply not close to any unitary $\alpha$-design. This strategy is not directly relevant to testing designs at the global level, but it can probe the typical behaviors of entanglement between local regions of designs. Recall that R\'enyi entropy is monotonically nonincreasing in the order, and all orders share the same roof value. So $\alpha$-scrambling necessarily implies $\alpha'$-scrambling, for $\alpha\geq\alpha'$.  In scrambling dynamics, the R\'enyi-$\alpha$ entanglement entropy is expected to grow slower and saturate the maximum at a later time than R\'enyi-$\alpha'$ in general.


\subsubsection{Extreme orders: min- and max-scrambling}
Now we discuss the 1- and $\infty$-scrambling more carefully, which respectively correspond to the the weakest and strongest entropic scrambling complexities, given by the low and high ends of R\'enyi entropies. 

Recall that $\alpha\rightarrow 1$ gives the von Neumann entropy, which probes information scrambling. 
First notice that unitary 1-designs do not necessarily create nontrivial entanglement or scramble quantum information. For example, the ensemble of tensor product of Pauli operators acting on each qubit
\begin{equation}
    \{P_1\otimes P_2\otimes\cdots\otimes P_n\},~~P = I,\sigma_x,\sigma_y,\sigma_z
\end{equation}
forms a unitary 1-design \cite{ry}. However, this local Pauli ensemble clearly does not scramble in any sense, since it cannot create entanglement among qubits (so local operators do not grow). So any entanglement entropy will be zero. 
On the other hand, unitary 2-designs are sufficient to maximize R\'enyi-2 entropies, which lower bounds the corresponding von Neumann entropies. It is shown in \cite{ding} that there actually exists a clear gap between them.
So information scrambling is strictly weaker than 2-scrambling, but on the other hand strictly stronger than 1-designs.
More precise characterizations may depend on the specific signatures of min-scrambling one is using, and require more careful analysis of designs and generalized entropies in the non-integer order regime, which remains largely unclear and is left for future work.

The other end of the spectrum is $\alpha\rightarrow\infty$, which leads to the min entropy $S_{\text{min}}(\rho) =-\log\norm{\rho} = -\log\lambda_{\text{max}}(\rho)$. Large min entanglement entropy directly indicates that the spectrum of the reduced density matrix is almost completely uniform, since it only cares about the largest eigenvalue. As the example of $\vec\lambda$ in Section \ref{entropies} shows, the min entropy is extremely sensitive to even one small peak in the entanglement spectrum. So it can be regarded as the ``harshest'' entropy measure and the strongest entropic diagnostic of scrambling: if the min entanglement entropy is almost maximal, then the system must be very close to maximally entangled in any sense and we cannot effectively distinguish the scrambler from Haar random by any R\'enyi entanglement entropy. This corresponds to the highest entropic scrambling complexity in our framework and thus we call this ``max-scrambling''.   We shall see in a moment that designs of sufficiently high orders are simply indistinguishable from the Haar measure (also in the random state setting) by studying the min entanglement entropy of designs, which implies that max-scrambling is not an infinitely strong condition.

\subsubsection{Nontrivial moments and fast max-scrambling}
Given the definition of max-scrambling by the min entanglement entropy, one may wonder if the full Haar measure is needed to achieve this strongest form of entropic scrambling. Here we answer this question in the negative: for a given dimension, only a finite number of moments (which scales logarithmically in the dimension) are needed to maximize the min entanglement entropy, which we call nontrivial moments.

First we note that the same Haar-averaged min entanglement entropy results in Theorem~\ref{thm:AveMinEntropy} hold if the average is taken over a unitary $\alpha$-design with $\alpha\geq \lceil \sqrt{d}/2\rceil$. 
The conclusion is clear from the proof when $d\geq 50$. When $17\leq d\leq  49$, $\lceil \sqrt{d}/2\rceil\geq3$, so the conclusion also follows from the proof. The conclusion is obvious when $d\leq7$. It remains to consider the case $8\leq d\leq  16$, which means $\lceil \sqrt{d}/2\rceil=2$. Therefore, Eq.~\eqref{eq:SecMomentChoi} applies, so that
\begin{equation}
d \int{\rm d}U \|\rho_{AC}\|\leq d\left(\int{\rm d}U \tr\{\rho_{AC}^{2}\}\right)^{1/2}\leq \sqrt{2d}< 7.
\end{equation}
Therefore, Eqs.\ \eqref{eq:AveNormChoi} and \eqref{eq:AveMinEntropyChoi} hold.

We can further show that, in fact, a unitary $O(\log d)$-design is enough to achieve nearly maximal min entanglement entropy:

\begin{thm}\label{thm:AveMinEntropyLog}
	Let  $\nu_\alpha$ be a unitary $\alpha$-design, where	$1\leq\alpha= \lceil \log d/a\rceil \leq \sqrt{d}/2$ and $ a> 0$; then 
	\begin{align}
	d\bbE_{\nu_\alpha} \|\rho_{AC}\|&\leq 2^{2+a}, \label{eq:AveNormChoiLog2}    \\
	\bbE_{\nu_\alpha} S_{\min}(\rho_{AC})&\geq \log d-2-a.  \label{eq:AveMinEntropyChoiLog2}   
	\end{align}	
In particular, if $\alpha\geq \lceil \log d\rceil$, then 
	\begin{align}
	d\bbE_{\nu_\alpha} \|\rho_{AC}\|&\leq 8, \label{eq:AveNormChoiLog}    \\
	\bbE_{\nu_\alpha} S_{\min}(\rho_{AC})&\geq \log d-3.  \label{eq:AveMinEntropyChoiLog}   
	\end{align}		
\end{thm}
\begin{proof}
If 	$1\leq \alpha= \lceil \log d/a\rceil \leq \sqrt{d}/2$, then one can show that Eq.~\eqref{eq:AveAlphaMomentProof} holds as in the proof of Theorem~\ref{thm:AveMinEntropy} even without additional restrictions. 
Therefore,
\begin{equation}\label{eq:AveNormChoiLogProof} 
d \bbE_{\nu_\alpha} \|\rho_{AC}\|\leq d\left(\bbE_{\nu_\alpha} [\mathrm{tr}\{\rho_{AC}^{\alpha}\}]\right)^{1/\alpha}\leq 4 \Bigl(\frac{d}{\alpha^{3/2}}\Bigr)^{1/\alpha}\leq
4 d^{1/\alpha}\leq 4 d^{a/\log d}=2^{2+a},
\end{equation}
which confirms Eq.~\eqref{eq:AveNormChoiLog2} and implies Eq.~\eqref{eq:AveMinEntropyChoiLog2}.

Now suppose $\alpha=\lceil \log d\rceil$. If $\alpha\geq \lceil\sqrt{d}/2\rceil$, then Eqs.~\eqref{eq:AveNormChoiLog} and \eqref{eq:AveMinEntropyChoiLog} hold by Theorem~\ref{thm:AveMinEntropy} and the above analysis for unitary $\lceil\sqrt{d}/2\rceil$-designs.
Otherwise, if $\alpha\leq\sqrt{d}/2$, the two equations follow from Eqs.~\eqref{eq:AveNormChoiLog2} and  \eqref{eq:AveMinEntropyChoiLog2} with $a=1$. 
The same conclusion also holds when $\alpha\geq\lceil \log d\rceil$.
%
%
\end{proof}


This result is crucial to the understanding and characterization of max-scrambling.  In particular, the observation that log-designs can already achieve max-scrambling leads to an interesting argument about max-scrambling in physical dynamics. 
The studies of the dynamical scrambling behaviors of physical systems primarily care about the amount of time needed for the system to scramble under certain constraints. 
The fast scrambling conjecture \cite{fast} is the standard general argument about the limitation on this scrambling time, roughly saying that the fastest min-scramblers take $O(\log n)$ time, where $n\sim\log d$ is the number of degrees of freedom (and black holes, as in reason the most complex physical system and the fastest quantum information processor in nature, should achieve this bound).

Here we may ask similar questions for the complexities beyond min-scrambling: How fast can physical dynamics achieve certain scrambling complexities, in particular, max-scrambling?  To make the assumption of ``physical'' more explicit, one typically requires the Hamiltonian governing the evolution to be local (meaning that each interaction term involves at most a finite number of degrees of freedom) and time-independent. 
Ref.~\cite{nakata} introduces the notion of design Hamiltonian, and conjectures that there are physical Hamiltonians that approximate unitary $\alpha$-designs in time that scales roughly as $O(\alpha \log n)$. Note that the approximation scheme and error dependence will be important in translating it to the language of scrambling complexities.  For example, for m-approximation error $\epsilon$, an $\omega(\log\log(1/\epsilon))$ dependence is sufficient to dominate $\log n$ by the previous error analysis.  Based on the above nontrivial moments result and the design Hamiltonian conjecture, the fastest max-scrambling time scales roughly as $O(n \log n)$.  To absorb the non-primary effects, we state the conjecture using soft notations (absorbing polylogarithmic factors) as follows:
\begin{conj}[Fast max-scrambling conjecture]
    Max-scrambling can be achieved by physical dynamics in $\tilde{O}(n)$ time, i.e.~in time roughly linear in the number of degrees of freedom.
\end{conj}
To better formalize and study this fast max-scrambling conjecture, it would be important to further investigate the error dependency.     Fast scrambling is an active research topic that has led to many key developments in quantum gravity and quantum many-body physics in recent years, such as the SYK model \cite{PhysRevLett.70.3339,k}.  It could be interesting to generalize the studies about fast scrambling to this strong notion of max-scrambling.

\subsubsection{On the gaps between entropic scrambling complexities}
A further question then arises as to whether the entropic scrambling complexities form a strict hierarchy, i.e., whether different complexities are gapped. 

A straightforward but strong definition of a separation between $\alpha$- and $\alpha'$-scrambling ($\alpha'<\alpha$) is the following: There exist scramblers such that the associated R\'enyi-$\alpha'$ entropies are always near maximal, but some R\'enyi-$\alpha$ entropies can be bounded away from maximal.
Such separations are in principle possible according to the properties of R\'enyi entropies (recall $\vec{\lambda}$). 
However, by the nontrivial moments result, we already know that $O(\log d)$ and higher complexities are not truly separated.

We tried several approaches to establish general separations in the Choi model, with limited success.  In particular, we attempted to generalize the partially scrambling unitary model \cite{ding}, and attempted to extend the gap results in the random state setting (next section) to random unitaries. 
The partially scrambling unitary model is used in \cite{ding} to prove a large separation between von Neumann and R\'enyi-2 tripartite information in the Choi state setting. By contrast, as we analyze in Appendix \ref{app:partial}, this model is not likely to provide similar separations among generalized entropies. The analysis nevertheless reveals a rather interesting tradeoff between sensitivity and robustness between R\'enyi and $s>0$ entropies. 
However, we are able to establish gaps using projective designs in the random state setting (see next section), but the results cannot be directly generalized to unitary designs. The reasons will be explained in more detail in the next section.  We leave the gap problem in the Choi model open for the moment.  We note, however, that the absence of strict separations of this type is not indicating that the behaviors of R\'enyi entropies (of sublogarithmic orders) are not separated in physical scenarios.  We may still expect, for example, that the higher orders grow slower than lower orders, so that they still separate different complexities.








\subsubsection{Relating to other complexities}
It would be interesting to relate the entropic scrambling complexities to other traditional types of complexities, such as circuit complexity.  For example, consider the local random circuit model. It is shown in \cite{brandaoharrow1,brandaoharrow2} that $O(\alpha^9 n[\alpha n+\log({1/\epsilon})])$ Haar random local gates are sufficient to form an $\epsilon$-m-approximate $\alpha$-design of $n$ qubits.  By the error analysis result, one can easily see that the minimum number of gates/circuit depth needed to maximize R\'enyi-$\alpha$ entropies scales polynomially in $\alpha$ and $n$: Let $\epsilon = 2^{-3\alpha n}$ so that $\log(1/\epsilon)=3\alpha n$, then the number of gates scales as $O(\alpha^{10}n^2)$, but meanwhile the deviation $\epsilon$ is sufficiently small such that the error in $S_R^{(\alpha)}(\rho_{AC})$ is vanishingly small, which indicates that such circuit is a good $\alpha$-scrambler.
That is, the entropic scrambling complexity and the random circuit complexity (minimum number of random gates) can be polynomially related.
We note that the $O(\alpha^{10}n^2)$ scaling can be improved to $O(\alpha n^2)$ for $\alpha=o(\sqrt{n})$ \cite{nakata}.  
Moreover, the fast design and max-scrambling conjectures discussed in the last part can be regarded as connections to time complexity (in the physical sense).

\section{Generalized entanglement entropies and random states}\label{sec:RandS}

The previous section focused on Choi states, which are representations of the corresponding unitary channels.
Here we consider a more straightforward problem---the entanglement in random and pseudorandom states---to generalize the connections between generalized entropies and designs. 
Note that the Page-like results, that a truly random state should typically be highly entangled, have been playing important roles in many fields including quantum gravity, quantum statistical mechanics, and quantum information theory for a long time. 
In this pure state setting, we obtain analogous main results that designs maximize corresponding R\'enyi entanglement entropies, closing the complexity gap in the Page's theorem, and that there are at most logarithmic nontrivial moments. These results suggest a similar hierarchy of entropic randomness complexities of states, which we call Page complexities.  In addition, we are able to get solid results on the gap problem. We shall follow similar steps as in the random unitary setting, but with more focus on the different aspects. The presentation of similar arguments and derivations is going to be more compact.

\subsection{Setting}
The mathematical setting is as follows. Consider a bipartite system with Hilbert space 
$\mathcal{H}=\mathcal{H}_A\otimes \mathcal{H}_B$, where $\mathcal{H}_A, \mathcal{H}_B$ have dimensions $d_A, d_B$, respectively, assuming $d_A\leq d_B$.   
We essentially need to compute the generalized entropies of the reduced density operator $\rho_A$.
From here on we use $\bbE$ to denote the average over states drawn uniformly from the unit sphere in $\mathcal{H}$.  Note that this uniform distribution on pure states is equivalent to the distribution generated by a Haar random unitary acting on some fixed fiducial state, so the induced uniformly random pure state is also called a Haar random state.

More explicitly, the Page's theorem (originally conjectured by Page in \cite{Page93}, proved in \cite{FoonK94, Sanc95, Sen96}) states that the average entanglement entropy of each reduced state is given by 
\begin{equation}
\bbE S(\rho_A)=\bbE S(\rho_B)=\frac{1}{\ln 2}\left( \sum_ {j=d_B+1}^{d_Ad_B}-\frac{d_A-1}{2d_B}\right)>\log d_A-\frac{1}{2\ln2}\frac{d_A}{d_B}\geq \log d_A-\frac{1}{2\ln2}.
\end{equation}
The gap between the average entropy and the maximum $\log d_A$ is bounded by the dimension-independent  constant ${1}/(2\ln2)$. 
Similar observations were even earlier made by Lubkin \cite{lubkin} and Lloyd/Pagels \cite{lloydpagels}. In particular, \cite{lloydpagels} derived the distribution of the local eigenvalues of a random state, which may imply this result. Also see e.g.\ \cite{Hayden2006,stab} for further studies of this phenomenon.
In the following we shall strengthen this result by proving the gap between the average R\'enyi-$\alpha$ entropy of each reduced state and the roof value $\log d_A$ is also bounded by  a constant that is independent of the dimensions $d_A, d_B$ and the order $\alpha$.


\subsection{Haar random states}

Similarly, we first derive the integrals of the trace term and generalized entanglement entropies over the uniform measure.

\subsubsection{General trace formula}

Suppose $\ket{\psi}$ is drawn uniformly from the unit sphere in $\mathcal{H}$. The analytical formula for the average of the $\alpha$-moment $\tr\{\rho_A^\alpha\}$, where $\rho_A$ is the reduced density matrix of $\ket{\psi}$ for system $A$, is derived as follows. Expand $|\psi\rangle$ in the standard product basis $|\psi\rangle=\sum_{jk}\psi_{jk} |jk\rangle$, where $j=1, 2,\ldots, d_A$ label the basis elements for $\mathcal{H}_A$, and $k=1,2,\ldots, d_B$ label the basis elements for $\mathcal{H}_B$. Then 
\begin{equation}\label{eq:rhoa}
\rho_A=\sum_{j_1, j_2, k} \psi_{j_1 k}\psi_{j_2,k}^*|j_1\>\< j_2|.
\end{equation}
The general result on the Haar-averaged trace is as follows:
\begin{thm}
\begin{equation}\label{eq:tra}
\bbE \tr\{\rho_A^\alpha\}=\frac{1}{\alpha!D_{[\alpha]}}\sum_{\sigma\in S_\alpha}d_A^{\xi(\sigma\tau)}d_B^{\xi(\sigma)},
\end{equation}
where 
\begin{equation}
D_{[\alpha]}=\binom{d_Ad_B+\alpha-1}{\alpha}=\frac{d_Ad_B(d_Ad_B+1)\cdots (d_Ad_B+\alpha-1)}{\alpha !}
\end{equation}
is the dimension of the symmetric subspace of $\mathcal{H}^{\otimes \alpha}$.
\end{thm}
\begin{proof}
By Eq.\ (\ref{eq:rhoa}),
\begin{equation}
\tr\{\rho_A^\alpha\}=\sum_{\text{all  indices}}\psi_{j_1,k_1}\psi_{j_2,k_1}^*\psi_{j_2,k_2}\psi_{j_3,k_2}^*\cdots \psi_{j_\alpha,k_\alpha}\psi_{j_1,k_\alpha}^*=\tr\left[(|\psi\>\<\psi|)^\alpha Q_\alpha\right],
\end{equation}
where
\begin{equation}
Q_\alpha=\sum_{\text{all indices}} |j_2k_1\>\<j_1 k_1|\otimes |j_3 k_2\>\<j_2k_2|\otimes \cdots\otimes |j_1k_\alpha\>\<j_\alpha k_\alpha|. 
\end{equation}
Therefore,
\begin{equation}\label{eq:AveAlphaMomentSym1}
\bbE \tr\{\rho_A^\alpha\}=\frac{1}{D_{[\alpha]}}\tr\{P_{[\alpha]}Q_\alpha\},
\end{equation}
where
$P_{[\alpha]}$ is the projector onto the symmetric subspace of $\mathcal{H}^{\otimes \alpha}$, and $D_{[\alpha]}$ is its dimension.
Recall that the symmetric group $S_\alpha$ acts on $\mathcal{H}^{\otimes \alpha}$ by permuting the tensor factors,  and $P_{[\alpha]}$ can be expressed as follows 
\begin{equation}
P_{[\alpha]}=\frac{1}{\alpha !}\sum_{\sigma\in S_\alpha} U_\sigma,
\end{equation}
where $U_\sigma$ is the unitary operator associated with the permutation $\sigma$. Simple analysis shows that 
\begin{equation}
\tr\{U_\sigma Q_\alpha\}=d_A^{\xi(\sigma \tau)}d_B^{\xi(\sigma)}.
\end{equation}
Consequently,
\begin{equation}\label{eq:train}
\bbE \tr\{\rho_A^\alpha\}=\frac{1}{\alpha!D_{[\alpha]}}\sum_{\sigma\in S_\alpha}d_A^{\xi(\sigma\tau)}d_B^{\xi(\sigma)}. 
\end{equation}
\end{proof}
We noticed that similar results have been derived and rederived several times \cite{lubkin,ZyczS01,MalaML02,Collins2010,CollN11}. Compared to known approaches, our approach seems simpler; in addition, it admits easy generalization to states drawn from (approximate) complex projective designs, which is not obvious for other approaches of which we are aware.

To get an intuitive understanding of Eq.~\eqref{eq:tra}, it is worth taking a closer look at several concrete examples. When $\alpha=2$, we reproduce a formula derived by Lubkin \cite{lubkin}:
\begin{equation}
\bbE \tr\{\rho_A^2\}=\frac{d_A+d_B}{d_Ad_B+1}, 
\end{equation}
From this equation we can derive 
 a nearly-tight lower bound for the average R\'enyi-2 entanglement entropy,
\begin{equation}
\bbE S^{(2)}_R(\rho_A)\geq \log\frac{d_Ad_B+1}{d_A+d_B}>\log d_A-\log\frac{d_A+d_B}{d_B}\geq \log d_A-1.
\end{equation}
When $d_A=d_B$, the averages of the first few moments are given by
\begin{align}
\bbE\tr\{\rho_A^2\}&=\frac{2d_A}{d_A^2+1}\leq \frac{2}{d_A},\\
\bbE\tr\{\rho_A^3\}&=\frac{5d_A^2+1}{(d_A^2+1)(d_A^2+2)}\leq \frac{5}{d_A^2},\\
\bbE\tr\{\rho_A^4\}&=\frac{14d_A^3+10d_A}{(d_A^2+1)(d_A^2+2)(d_A^2+3)}\leq \frac{14}{d_A^3},
\end{align}
which imply that 
\begin{align}
\bbE S^{(2)}_R(\rho_A)&\geq\log d_A-1,\\
\bbE S^{(3)}_R(\rho_A)&\geq\log d_A -\frac{\log5}{2},\\
\bbE S^{(4)}_R(\rho_A)&\geq\log d_A -\frac{\log14}{3}.
\end{align}
Note that the gap of each R\'enyi entropy from the maximum is tied with the corresponding Catalan number. This is not a coincidence.

\subsubsection{Large $d$ limit}

When $d_A=d_B\rightarrow\infty$, the asymptotic results go as follows:
\begin{thm}
In the limit of large $d_A$,
\begin{align}
\bbE \tr\{\rho_A^\alpha\}&=\cat_\alpha d_A^{-\alpha+1}(1+O(d_A^{-2})). \\
\bbE S^{(\alpha)}_R(\rho_A)&\geq\log d_A-\frac{2\alpha}{\alpha-1} +O(d_A^{-2})\geq \log d_A-O(1). 
\end{align}
\end{thm}
\begin{proof}
The trace result also follows from the Cycle Lemma:
    \begin{equation}\label{eq:ARE1}
\bbE \tr\{\rho_A^\alpha\}=\frac{1}{\alpha!D_{[\alpha]}}\sum_{\sigma\in S_\alpha}d_A^{\xi(\sigma\tau)+\xi(\sigma)}=\frac{\cat_\alpha d_A^{\alpha+1} +O(d_A^{\alpha-1})}{d_A^{2\alpha}+O(d_A^{2\alpha-2})}=\cat_\alpha d_A^{-\alpha+1}(1+O(d_A^{-2})).
\end{equation}
Therefore,
\begin{equation}\label{eq:ARE2}
\bbE S^{(\alpha)}_R(\rho_A)\geq \frac{\log \cat_\alpha d_A^{-\alpha+1}}{1-\alpha}+O(d_A^{-2})
=\log d_A-\frac{\log \cat_\alpha}{\alpha-1}+O(d_A^{-2})
\geq\log d_A-\frac{2\alpha}{\alpha-1} +O(d_A^{-2}).
\end{equation}
So the residual R\'enyi entropy is $O(1)$.
\end{proof}
This theorem suggests that the gap between the average R\'enyi-$\alpha$ entropy and the maximum $\log d_A$ is bounded by a constant asymptotically.

\subsubsection{Non-asymptotic bounds}
The following bounds hold for any $d_A\leq d_B$:
%

\begin{lem}\label{lem:alphaMoment}
Let $q:=\alpha^3/(32d_B^2)<1, h(q):=1+2q/[3(1-q)]$. Then 
\begin{align}
\bbE \tr\{\rho_A^\alpha\}&\leq h(q)\cat_\alpha d_A^{1-\alpha}
\leq \frac{4^\alpha h(q)}{\sqrt{\pi}\alpha^{3/2}}d_A^{1-\alpha},\\
\bbE S^{(\alpha)}_R(\rho_A)
&\geq\log d_A-\frac{2\alpha-\frac{3}{2}\log\alpha +\log h(q) -\frac{1}{2}\log\pi  }{\alpha-1}.
\end{align}
\end{lem}
	

\begin{proof}
According to Lemma~\ref{lem:sumBound},
\begin{equation}\label{eq:ARE3}
\bbE \tr\{\rho_A^\alpha\}=\frac{1}{\alpha!D_{[\alpha]}}\sum_{\sigma\in S_\alpha}d_A^{\xi(\sigma\tau)}d_B^{\xi(\sigma)}\leq
\frac{h(q)\cat_\alpha d_Ad_B^{\alpha}}{d_A^\alpha d_B^\alpha}
\leq 
 h(q)\cat_\alpha d_A^{1-\alpha}
\leq \frac{4^\alpha h(q)}{\sqrt{\pi}\alpha^{3/2}}d_A^{1-\alpha},
\end{equation}
which in turn implies that
\begin{align}
\bbE S^{(\alpha)}_R(\rho_A)
&\geq \frac{1}{1-\alpha} \log \bbE \tr\{\rho_A^\alpha\}
\geq \frac{1}{1-\alpha} \log \frac{4^\alpha h(q)}{\sqrt{\pi}\alpha^{3/2}}d_A^{1-\alpha},\\
&=\log d_A-\frac{2\alpha-\frac{3}{2}\log\alpha +\log h(q) -\frac{1}{2}\log\pi  }{\alpha-1}.
\end{align}

\end{proof}

%


In fact we can show that the gap is at most 2:

\begin{thm}\label{thm:ARE}
For all $d_A\leq  d_B$ and $\alpha\geq 0$,
\begin{equation}
	\bbE S^{(\alpha)}_R(\rho_A)\geq \log d_A-2. 
	\end{equation}
\end{thm}
\begin{proof}
	Recall that R\'enyi-$\alpha$ entropy is nonincreasing with $\alpha$, so to establish the theorem, it suffices to  prove the lower bound for the min entropy. For all $\alpha$,
	\begin{align}
	& S^{(\alpha)}_R(\rho_A)\geq S_{\min}(\rho_A)=-\bbE \log \|\rho_A\|\geq- \log\bbE \|\rho_A\|
	\geq -\log(\bbE  \|\rho_A\|^\beta)^{1/\beta} \nonumber\\ \geq&  -\log (4d_A^{-1})
	= \log d_A-2,
	\end{align}
	where the second line follows from Lemma~\ref{lem:AveNormPower} below, by taking $0< \beta\leq \lfloor (29d_B^2)^{1/3}\rfloor$. 
\end{proof}

\begin{lem}\label{lem:AveNormPower}
For all $d_A\leq d_B$ and $0< \alpha\leq \lfloor (29d_B^2)^{1/3}\rfloor$,
	\begin{equation}
	(\bbE  \|\rho_A\|^\alpha)^{1/\alpha} \leq 4 d_A^{-1}.
	\end{equation} 
\end{lem}

\begin{proof}
The conclusion  is obvious when $d_A\leq 4$. When $d_B\geq d_A\geq 5$, note that	
 	$(\bbE \|\rho_A\|^\alpha)^{1/\alpha}$ is nondecreasing with $\alpha$ for $\alpha>0$, so it suffices to prove the lemma in the case 
 $\alpha=\lfloor 29d_B^{2/3}\rfloor$.   Then $\alpha^{3/2}\geq d_B\geq d_A$ and  $0.6< q=\alpha^3/(32d_B^2)\leq29/32$. According to Lemma~\ref{lem:alphaMoment},
\begin{align}
\bbE \|\rho_A\|^\alpha&\leq \bbE \tr\{\rho_A^\alpha\}
\leq \frac{4^\alpha h(q)}{\sqrt{\pi}\alpha^{3/2}}d_A^{1-\alpha}
=\left[\frac{d_A}{\sqrt{32\pi q}d_B}\left(1+\frac{2q}{3(1-q)}\right)\right]4^\alpha d_A^{-\alpha}\nonumber\\
&\leq \frac{3-q}{3(1-q)\sqrt{32\pi q}}4^\alpha d_A^{-\alpha}
\leq 4^\alpha d_A^{-\alpha},
\end{align}
which implies the lemma.
Here the last inequality follows from the observation that  $f(q):=(3-q)/[3(1-q)\sqrt{32\pi q}]<1$  for  $0.6< q\leq29/32$. This fact can be verified immediately if we notice that the derivative  $f'(q)$ has a unique zero at $q_0=4-\sqrt{13}$ in the interval $0<q<1$ and that $f(q)$ is monotonically decreasing for $0<q<q_0$ and monotonically increasing for $q_0<q<1$. 
\end{proof}


%

	We also obtain the following bound, which improves Theorem~\ref{thm:ARE} when $d_A \ll d_B$: 
\begin{thm}\label{thm:ARE2}
For all $d_A\leq  d_B$ and $\alpha\geq 0$,
	\begin{equation}
	\bbE S^{(\alpha)}_R(\rho_A)\geq \log d_A-2\log\left(1+\sqrt{\frac{d_A}{d_B}}\right)-\log c\geq\log
	d_A-\frac{2}{\ln2}\sqrt{\frac{d_A}{d_B}}-\log c,
	\end{equation}
	where $c=1$ if  $\mathcal{H}$ is real and   $c=2$ if $\mathcal{H}$ is complex. 	
\end{thm}
\begin{proof}
The proof goes similarly as Theorem \ref{thm:ARE}. For all $\alpha\geq 0$,
	\begin{align}
	& S^{(\alpha)}_R(\rho_A)\geq S_{\min}(\rho_A)=-\bbE \log \|\rho_A\|\geq- \log\bbE \|\rho_A\|\geq -\log\left(\bbE \sqrt{ \|\rho_A\|}\right)^2 \nonumber\\ 
	\geq&  \log d_A-2\log\left(1+\sqrt{\frac{d_A}{d_B}}\right)-\log c\geq\log
	d_A-\frac{2}{\ln2}\sqrt{\frac{d_A}{d_B}}-\log c,
	\end{align}
where $c=1$ if  $\mathcal{H}$ is real and   $c=2$ if $\mathcal{H}$ is complex. The second line follows from Lemma \ref{lem:AveRootNorm} stated below. 
\end{proof}
\begin{lem}\label{lem:AveRootNorm}
	\begin{equation}
	\bbE \sqrt{\|\rho_A\|}\leq \sqrt{c}\left(\frac{1}{\sqrt{d_A}}+\frac{1}{\sqrt{d_B}} \right),
	\end{equation}
	where $c=1$ if  $\mathcal{H}$ is real and   $c=2$ if $\mathcal{H}$ is complex.   
\end{lem}
The proof of this lemma is rather complicated, so we leave it in Appendix~\ref{app:lem24}. We believe that the constant $c$ in Theorem~\ref{thm:ARE2} and Lemma~\ref{lem:AveRootNorm} can be set to 1 in both real and complex cases.
We note that Hayden and Winter had a similar result \cite{HaydW08}, but they are not so explicit about the constant and the dimensions for which their result is applicable. 


%
\subsection{State designs and their approximates}


\subsubsection{Average over designs and Tight Page's theorems}
Recall Page's theorem, which states that Haar-averaged von Neumann entanglement entropies of small subsystems are almost maximal. This theorem is not tight from the perspectives of both entropy and randomness: by the results above, the Haar-averaged R\'enyi entanglement entropies of higher orders are generically close to maximum as well, and the complete randomness is an overkill to maximize the entanglement entropies in terms of randomness complexity. Our results imply that Page's theorem can be strengthened from both sides.
Similar to the random unitary setting, since $\bbE\mathrm{tr}\{\rho_A^\alpha\}$ only uses $\alpha$ moments of the uniform measure, all bounds on $\bbE\mathrm{tr}\{\rho_A^\alpha\}$ and $\bbE S_R^{(\alpha)}(\rho_A)$ from the last part still hold if the average is over $\alpha$-designs. So we arrive at the following bounds that can be regarded as tight Page's theorems for each order $\alpha$, by Theorems \ref{thm:ARE}, \ref{thm:ARE2}:
\begin{thm}[Tight Page's theorems]
Let $\nu_\alpha$ be an $\alpha$-design. Then
\begin{align}
\bbE_{\nu_\alpha}\mathrm{tr}\{\rho_A^\alpha\} &= \bbE\mathrm{tr}\{\rho_A^\alpha\}, \\
\bbE_{\nu_\alpha} S_R^{(\alpha)}(\rho_A) &\geq f_R^{(\alpha)}\left(\bbE\mathrm{tr}\{\rho_A^\alpha\}\right) = \frac{1}{1-\alpha}\log\bbE\mathrm{tr}\{\rho_A^\alpha\},
\end{align}
For all $d_A\leq d_B$ and all $\alpha\geq 0$, the following bounds hold:
\begin{equation}
    \bbE_{\nu_\alpha} S^{(\alpha)}_R(\rho_A)\geq \log d_A -2,
\end{equation}
and
	\begin{equation}
	\bbE_{\nu_\alpha} S^{(\alpha)}_R(\rho_A)\geq \log d_A-2\log\left(1+\sqrt{\frac{d_A}{d_B}}\right)-\log c\geq\log
	d_A-\frac{2}{\ln2}\sqrt{\frac{d_A}{d_B}}-\log c,
	\end{equation}
	where $c=1$ if  $\mathcal{H}$ is real and   $c=2$ if $\mathcal{H}$ is complex. 
\end{thm}
Obviously $\bbE_{\nu_\alpha} S^{(\alpha)}_R(\rho_A)\geq \log d_A - O(1)$ also hold in the limit of large $d_A$.

\subsubsection{Approximate designs}
Here we directly consider the more relevant notion of approximate $\alpha$-designs given by deviation in frame operators.    This error analysis is important for characterizing the randomness complexity by R\'enyi entropies, as will be explained later.

Given an ensemble $\nu$ of quantum states, define
\begin{equation}
\Delta_\alpha(\nu):= D_{[\alpha]}\bbE_\nu (\ketbra{\psi}{\psi})^{\otimes t}-P_{[\alpha]}. 
\end{equation}
\begin{defn}[FO-approximate designs]
An ensemble $\nu$ is an $\lambda$-approximate $\alpha$-design if 
\begin{equation}
\left\|\Delta_\alpha(\nu)\right\|_1\leq \lambda. 
\end{equation}
\end{defn}
\begin{thm}\label{lem:AppDesignEntropy}
Let $\omega_\alpha$ be an $\lambda$-FO-approximate $\alpha$-design with $\alpha\geq2$. Then 
\begin{align}
\bbE_{\omega_\alpha}\tr\{\rho_A^\alpha\}&\leq \bbE\tr\{\rho_A^\alpha\}+\frac{\lambda}{D_{[\alpha]}},\\
\bbE_{\omega_\alpha} S_R^{(\alpha)}(\rho_A)&\geq \frac{1}{1-\alpha}\log\left( \bbE \tr\{\rho_A^\alpha\}+\frac{\lambda}{D_{[\alpha]}}\right).
\end{align}
In the large $d_A$ limit,
\begin{equation}
    \bbE_{\omega_\alpha} S_R^{(\alpha)}(\rho_A) \geq \log d_A -O(1) -\frac{1}{(\alpha-1)\cat_\alpha \ln 2}\frac{d_A^{\alpha-1}\lambda}{D_{[\alpha]}}(1+O(d_A^{-2})).
\end{equation}
\end{thm}


\begin{proof}
According to the same argument that leads to Eq.\ \eqref{eq:AveAlphaMomentSym1},
\begin{align}\label{eq:AveAlphaMomentSym}
\bbE_{\omega_\alpha} \tr\{\rho_A^\alpha\}&=\tr\left\{\bbE_\nu (\ketbra{\psi}{\psi})^{\otimes t}Q_\alpha\right\}
=\frac{1}{D_{[\alpha]}}\tr\left\{(P_{[\alpha]}+\Delta_\alpha(\nu))Q_\alpha\right\},\nonumber\\
&=\bbE\tr\{\rho_A^\alpha\}+\frac{1}{D_{[\alpha]}}\tr\left\{\Delta_\alpha(\nu)Q_\alpha\right\}\leq \bbE\tr\{\rho_A^\alpha\}+\frac{1}{D_{[\alpha]}}\|\Delta_\alpha(\nu)\|_1 \|Q_\alpha\|\nonumber\\
&\leq \bbE\tr\{\rho_A^\alpha\}+\frac{\lambda}{D_{[\alpha]}},
\end{align}
where the last inequality follows from the assumption $\left\|\Delta_\alpha(\nu)\right\|_1\leq \lambda$ and the fact that $\|Q_\alpha\|=1$, since $Q_\alpha$ is unitary.
\end{proof}
We see that the residual entropy remains $O(1)$ as long as $\lambda/D_{[\alpha]}=O(d_A^{1-\alpha})$.    





\subsection{Hierarchy of Page complexities}
\subsubsection{Page Complexities by R\'enyi entanglement entropy}
Like the unitary case, our analysis of R\'enyi entanglement entropies lead to an entropic notion of randomness complexities: the complexity of $\alpha$-designs can be witnessed by whether the average R\'enyi-$\alpha$ entanglement entropies are close enough to the maximum.   Here we call them Page complexities as the foundation of this framework is the hierarchy of tight Page's theorems.

Here we provide an illustrating example based on the Clifford group.
As an application of Lemma~\ref{lem:AppDesignEntropy}, let us consider the average R\'enyi entanglement entropy of  Clifford orbits for a multiqubit system. For simplicity we assume $d_B=d_A\gg\alpha$, so that
\begin{equation}
\bbE\tr\{\rho_A^\alpha\}\approx \cat_\alpha d_A^{1-\alpha},\quad \bbE S_R^{(\alpha)}(\rho_A)\gtrsim \log d_A-\frac{\log \cat_\alpha}{\alpha-1}.
\end{equation}
Recall that the Clifford group is a unitary 3-design \cite{Zhu15MC,Webb15}, so any orbit of the Clifford group forms a 3-design. Consequently, the  average R\'enyi-$\alpha$ entanglement entropy for $\alpha\leq 3$ of any Clifford orbit is close to the maximum,
\begin{equation}\label{eq:AveEntropyCliOrbit}
\bbE_{\orb(\psi)}\tr\{\rho_A^\alpha\}\approx \cat_\alpha d_A^{1-\alpha},\quad \bbE_{\orb(\psi)} S_R^{(\alpha)}(\rho_A)\gtrsim \log d_A-\frac{\log \cat_\alpha}{\alpha-1},  
\end{equation}
for any $\psi$, where $\orb(\psi)$ denotes the Clifford orbit generated from $\psi$.

However, the Clifford group is not a 4-design, and Clifford orbits are in general not 4-designs \cite{Zhu15MC,Webb15,KuenG15}. If $\psi$ is a stabilizer state, then $\left\|\Delta_4(\orb(\psi))\right\|_1\approx d_A^6/12$ according to \cite{rep}. In this case the bounds for the fourth moment and R\'enyi-4 entropy provided by Theorem~\ref{lem:AppDesignEntropy} is not very informative, note that $\bbE\tr\{\rho_A^4\}\approx 14 d_A^{-3}$ and $D_{[4]}\approx (d_Ad_B)^4/24=d_A^8/24$. For a typical Clifford orbit, by contrast, $\left\|\Delta_\alpha(\nu)\right\|_1\approx d_A^2$ is much smaller \cite{rep}. Now Theorem~\ref{lem:AppDesignEntropy} implies that
\begin{align}
\bbE_{\orb(\psi)}\tr\{\rho_A^4\}&\leq \bbE\tr\{\rho_A^4\}+\frac{\left\|\Delta_\alpha(\orb(\psi))\right\|_1}{D_{[\alpha]}}\approx 14 d_A^{-3} +24d_A^{-6}\approx \bbE\tr\{\rho_A^4\}.
\end{align}
Therefore, Eq.\ \eqref{eq:AveEntropyCliOrbit} also holds for typical Clifford orbits when $\alpha=4$.  In our language, a Clifford orbit is very likely to have the Page complexity of 4-designs, although it is not really a 4-design in general.  This is a rather nontrivial example indicating that the Page complexity is a necessary but not sufficient condition for certifying designs. 

\subsubsection{Nontrivial moments}
Again, the min entanglement entropy witnesses the strongest Page complexity: if the average min entanglement entropies are always close to the maximum, then we simply cannot distinguish the ensemble from the completely random ensemble by the entanglement spectrum. The following theorem indicates that designs of order $O(\log d_A)$ maximize the min entanglement entropy and therefore achieve the max-Page complexity:
\begin{thm}\label{thm:Moment}
	Suppose $|\psi\rangle$ is drawn from an $\alpha$-design  in a bipartite Hilbert space $\mathcal{H}=\mathcal{H}_A\otimes \mathcal{H}_B$ of dimension $d_A\times d_B$, where  $\alpha=\lceil (\log d_A)/a\rceil\leq (16d_B^2)^{1/3}$ with $0<a\leq 1$. Let  $\rho_A$ be the  reduced state of subsystem A. 	
	Then
\begin{align}
\bbE \|\rho_A\|&\leq \frac{ 2^{2+a}}{d_A}, \\
\bbE S_{\min}(\rho_A)&\geq \log d_A-2-a.
\end{align}
In particular, $\bbE \|\rho_A\|\leq 8/d_A$ and $\bbE S_{\min}(\rho_A)\geq \log d_A-3$ if $\alpha=\lceil \log d_A\rceil$. 
\end{thm}

\begin{proof}

According to Lemma~\ref{lem:alphaMoment},
	\begin{align}
	\bbE \tr\{\rho_A^\alpha\}&
	\leq \frac{4^\alpha		
		 h(q)}{\sqrt{\pi}\alpha^{3/2}}d_A^{1-\alpha}
	\leq \frac{5\times 4^\alpha}{3\sqrt{\pi}\alpha^{3/2}}d_A^{1-\alpha}
\leq  4^\alpha d_A^{1-\alpha},
\end{align}
where the first inequality follows from the fact that $q=\alpha^3/(32d_B^2)\leq 1/2$ and $h(q)\leq 5/3$ given that $\alpha^3\leq 16d_B^2$ by assumption.
Consequently,
\begin{align}
\bbE \|\rho_A\|&\leq\left(\bbE \|\rho_A\|^\alpha\right)^{1/\alpha}\leq \left[\bbE \tr\{\rho_A^\alpha\}\right]^{1/\alpha}\leq d_A^{1/\alpha}\frac{4}{d_A}\leq 
d_A^{a/ \log d_A}\frac{4}{d_A}= \frac{ 2^{2+a}}{d_A},\\
\bbE S_{\min}(\rho_A)&\geq -\log \bbE \|\rho_A\|\geq -\log \frac{2^{2+a}}{d_A}\geq \log d_A-2-a.
\end{align}
In the case, $a=1$ and  $\alpha=\lceil \log d_A\rceil$, the  inequality   $\alpha\leq (16d_A^2)^{1/3}\leq (16d_B^2)^{1/3}$ holds automatically; therefore, $\bbE \|\rho_A\|\leq 8/d_A$ and $\bbE S_{\min}(\rho_A)\geq \log d_A-3$.
\end{proof}

So again the hierarchy of distinguishable Page complexities can only extend to logarithmic designs.

\subsubsection{Gaps between Page complexities}

Following the definition of gaps between the entropic scrambling complexities, one may wonder here whether there exist $\alpha$-designs such that R\'enyi entanglement entropies of orders larger than $\alpha$ are bounded away from the maximum, which we call ``gap $\alpha$-designs''.
In this random state setting, we are able to construct a family of gap 2-designs and so establish a strict gap between the second and $\alpha$-th Page complexities with all $\alpha\geq 3$. 
Our construction is based on the orbits of a  special subgroup of the unitary group on $\mathcal{H}=\mathcal{H}_A\otimes \mathcal{H}_B$. As mentioned before, any orbit of a unitary $2$-design is a complex projective $2$-design. What is interesting, 
our construction of projective 2-designs does not require unitary $2$-designs. In this way, we 
also provide a novel recipe for constructing projective  2-designs, which is particularly useful when the dimension is not a prime power.

Consider the group $G:=\rmU_A\otimes \rmU_B$, where $\rmU_A, \rmU_B$ are the unitary groups on $\mathcal{H}_A, \mathcal{H}_B$, respectively. It is irreducible, but does not form a 2-design. Simple analysis shows that $G$ has four irreducible components on $\mathcal{H}^{\otimes2 }$, with dimensions $d_A d_B(d_A\pm1)(d_B\pm 1)/4$, respectively. The symmetric subspace of $\mathcal{H}^{\otimes2 }$ contains two irreducible components with dimensions $d_A d_B(d_A+1)(d_B+1)/4$ and $d_A d_B(d_A-1)(d_B- 1)/4$. By a similar continuity argument as employed in \cite{rep}, there must exist an orbit of $G$ that forms a 2-design. Let $|\psi\rangle $ be a fiducial vector of a 2-design with reduced state $\rho_A$ for subsystem A. Then $\tr\{\rho_A^2\}$ is necessarily equal to the average over the uniform ensemble, that is,
\begin{equation}\label{eq:2-designCon}
\tr\{\rho_A^2\}=\frac{d_A+d_B}{d_Ad_B+1}.
\end{equation}
It turns out that this condition is also sufficient. To see this, note that the condition must be invariant under local unitary transformations and thus only depends on a symmetric polynomial of the eigenvalues of $\rho_A$ of degree 2, which is necessarily a function of $\tr\{\rho_A^2\}$ given the normalization condition $\tr\{\rho_A\}=1$.
It is worth pointing out that the same conclusion also holds if $\rmU_A, \rmU_B$
are replaced by groups that form unitary 2-designs on $\mathcal{H}_A, \mathcal{H}_B$, respectively. 

The following spectrum of $\rho_A$ with one large eigenvalue is a solution of Eq.~(\ref{eq:2-designCon}):
\begin{align}
\lambda_1=\frac{d_Ad_B+1+(d_A-1)\sqrt{(d_A+1)(d_Ad_B+1)}}{d_A(d_Ad_B+1)},\\ \lambda_2=\cdots=\lambda_{d_A} =\frac{d_Ad_B+1-\sqrt{(d_A+1)(d_Ad_B+1)}}{d_A(d_Ad_B+1)}.
\end{align}
If $d_B\geq d_A^2$, then 
\begin{equation}
\lambda_1\leq \frac{d_A^3+1+(d_A-1)\sqrt{(d_A+1)(d_A^3+1)}}{d_A(d_A^3+1)}<\frac{2}{d_A}. 
\end{equation}
Therefore, $S_{\min}(\rho_A)\geq \log d_A-1$, and the gap of all R\'enyi entropies from the maximum is bounded. 
The case in which the ratio $d_B/d_A$ is bounded by a constant, say $r$, has very similar features to the $\alpha=2$ single-peak spectrum discussed in Section \ref{entropies}. We have 
\begin{equation}
\lambda_1\geq \frac{\sqrt{(d_A+1)(d_Ad_B+1)}}{d_Ad_B+1}\geq d_B^{-1/2}\geq (rd_A)^{-1/2}.
\end{equation}
Consequently,
\begin{equation}
S_R^{(\alpha)}(\rho_A)\leq \frac{1}{1-\alpha}\log\lambda_1^\alpha\leq \frac{1}{1-\alpha}\log(rd_A)^{-\alpha/2}=\frac{\alpha}{2(\alpha-1)}(\log d_A +\log r). \label{gaprenyi}
\end{equation}
As $d_A$ increases, the gap of $S_R^{(\alpha)}(\rho_A)$ from the maximum is unbounded whenever $\alpha>2$.

We note that such construction cannot be directly generalized to establish gaps in the Choi setting.
As mentioned, any orbit of a unitary $t$-design is a complex projective $t$-design, but to construct a projective $t$-design, a unitary $t$-design is not required. Here the complex projective 2-design is constructed using a group that is a tensor product. However, such a group can never be a unitary 2-design.  Also, in the Choi setting, four parties are involved, and it is not easy to ensure unitarity using the idea for constructing projective designs. New approaches are necessary for such a construction.

\section{Concluding remarks}\label{sec:dis}

\subsection{Summary and open problems}

This paper explores the complexity of scrambling by connecting it to the degrees of quantum randomness via entanglement properties. In particular, we study the entanglement of state and unitary designs to lay the mathematical foundations for using R\'enyi and other generalized entanglement to probe the randomness complexities corresponding to designs, which we introduce as entropic scrambling complexities (or Page complexities in the state setting).  These complexities form a hierarchy that spans in between the most basic notions of scrambling and the max-scrambling which mimics the entanglement properties of Haar. In summary, our results mainly establish the following key features of entropic scrambling complexities:
\begin{figure}
    \centering
    \includegraphics[scale=0.82]{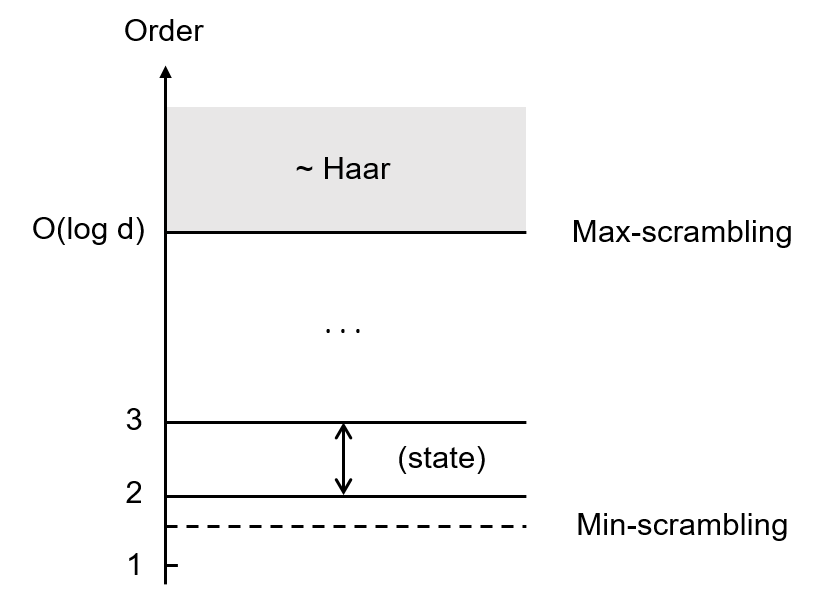}
    \caption{The hierarchy of entropic scrambling complexities.  Each order is given by near-maximality of corresponding R\'enyi entanglement entropies, which diagnose the complexity of corresponding designs.  The highest such complexity, corresponding to the notion of max-scrambling, is achieved at an order that roughly scales logarithmically in the dimension of the system.  The weakest form of scrambling, or min-scrambling, is weaker than order-2 but stronger than order-1.  Order-2 is strictly separated from higher orders in the state setting.}
    \label{fig:hierarchy}
\end{figure}
\begin{enumerate}
    \item $\alpha$-designs and close approximations induce almost maximal R\'enyi-$\alpha$ entanglement entropies.  This basic result links the maximality of R\'enyi entanglement entropies and the design complexity of corresponding orders. 
    \item $O(\log d)$-designs are sufficient to maximize the min entanglement entropy, which means that they achieve the highest entropic scrambling complexity, namely max-scrambling.  So all higher complexities collapse in the sense that they are simply indistinguishable from Haar scrambling by R\'enyi entanglement entropies.
    \item We show that there exist projective 2-designs with non-maximal R\'enyi-3 (and therefore higher order) entanglement entropies. This establishes a strict separation between the order-2 complexity and higher levels, at least in the state setting.
\end{enumerate}
The known structure of the entropic scrambling complexities based on our results is illustrated in Fig.~\ref{fig:hierarchy}.   
In summary, this study reveals the fine-grained complexity structure of the regime beyond information scrambling, and introduces a set of tools for studying it.  We also hope that this work initiates further research into this significant but relatively unknown regime.

There are several technical open problems, especially in the setting of unitary channels. For example, we are not yet able to give a construction that opens a strict gap between the entropic scrambling complexities.  Although we prove such gaps for projective 2-designs in the state setting, the similar techniques do not directly generalize to unitary channels. Moreover, due to the lack of subadditivity, we know that the negative tripartite information in terms of R\'enyi entropies are not necessarily nonnegative. It is worth looking into when this situation occurs, and further considering the meanings of such derived quantities.
Furthermore, this paper mostly concerns the expected values. It would be important to further analyze the variances and derive probabilistic bounds using concentration inequalities, in order to talk about ``typical'' behaviors in a more rigorous sense.  

\subsection{Outlook}

There are many interesting extensions to make.
For example, our results suggest that R\'enyi entanglement entropies could be powerful tools to further advance the study of quantum randomness and pseudorandomness.
A particularly interesting insight is that R\'enyi entropies of non-integer orders are naturally defined, which indicates that they can be helpful for characterizing and understanding the mysterious notion of designs of non-integer orders.  This problem is of interest in quantum information, and as explained earlier, is key to a more precise characterization of the min-scrambling complexity.   For example, it is reasonable to require that $\alpha$-designs (where $\alpha$ can be non-integer) by any definition must induce nearly maximal R\'enyi-$\alpha$ entanglement entropies. 
Then it is straightforward to see by Eq.~(\ref{gaprenyi}) that our gap 2-design induces small R\'enyi-$(2+\epsilon)$ entanglement entropy for any $\epsilon>0$, and so cannot be a $(2+\epsilon)$-design.
However, the attempts in properly defining non-integer designs and constructing such examples have mostly been negative so far.   We tried a few possible ways to construct random ensembles such that the maximal-nonmaximal ``phase transition'' of R\'enyi entropy occurs at some non-integer order which do not work well.  We also mention that the definition of designs by frame potential could be rather directly generalized to non-integer orders, but such generalization also suffers from fundamental problems \footnote{Learned from communications with Yoshifumi Nakata.}.   We hope to give more well-behaved definitions or constructions of non-integer designs, or find more fundamental reasons that they are not meaningful notions---either of which is very interesting. 

Also given that the entanglement properties of random states and channels play important roles in many areas in quantum information, including entanglement theory, quantum computing, and quantum cryptography, we expect the techniques and results here to find more interesting applications and advance the study of these fields.
It is worth mentioning that the recent study of pseudorandom quantum states and unitaries  from the perspective of computational indistinguishability \cite{2017arXiv171100385J}, which represents a different notion of quantum pseudorandomness that is more directly related to the practical requirements for cryptographic security.  It would be interesting to explore the role of entanglement in such computational quantum pseudorandomness, and find connections to our framework.


 

The current work focuses mostly on the kinematic or mathematical properties of unitary channels and states, which constitute a framework for further exploring the post-scrambling physics.     For example, it would be interesting to study the dynamical behaviors of R\'enyi entanglement entropies and randomness, and in particular investigate fast max-scrambling, in specific many-body or holographic systems.  By doing so we may extend existing studies of entanglement growth such as ``entanglement tsunami'' \cite{tsunami,PhysRevD.89.066012}), and eventually understand the whole process of scrambling and especially its relation to randomness and complexity generation.  In general, the study of randomness complexities may also shed new light on the fruitful idea of modeling complex systems (especially black holes \cite{mirror}) by random states or dynamics. 
 A recent study \cite{Gu2017} on (a 1d variant of) the strongly chaotic SYK model (which has drawn considerable interest as a solvable toy model of quantum black holes and holography) shows that, after a quench, there is a ``prethermal'' regime where light modes rapidly scramble, but the R\'enyi entanglement entropies do not reach thermal values, which confirms our expectation that the randomness complexity of the system is still low. However, the late-time behaviors remain unclear.      Another recent work  \cite{2018arXiv180310425Y} studies the R\'enyi entanglement entropies of random dynamics generated by Hamiltonians drawn from the Gaussian unitary ensemble (GUE). In general, the R\'enyi entanglement entropies are useful and analyzable quantities in the study of scrambling and chaos, and our work strengthens the motivation by connecting them to different randomness complexities.
 
 We also hope to establish more solid connections between the randomness complexities and the conventional ones, such as computational, gate and Kolmogorov complexities, which play active roles in recent studies of holographic duality and black holes \cite{cealong,CEA,secondlaw}, and are of independent interest.
 Note that an interesting recent paper \cite{secondlaw} directly concerns the evolution of complexity in generic physical dynamics. Here the complexity roughly means the computational/gate complexity, which is rather difficult to rigorously analyze. We feel that it is fruitful to combine their framework and viewpoints with ours.

  Moreover, the notion of scrambling and randomness discussed here is associated with the entire Hilbert space. It would be nice to extend the techniques and results to the finite temperature regime or systems with conserved quantities, so as to apply our ideas in more physical scenarios and in general the study of quantum thermalization and many-body localization.   We also hope to solidify the connections to several other relevant topics, including random tensor network holography \cite{Hayden2016} and OTO correlators. In summary, we believe that further research along the lines of research mentioned in this section could be essential to our understanding of quantum chaos,  quantum  statistical  mechanics,  quantum  many-body physics, and quantum gravity.

\begin{acknowledgments}
ZWL thanks David Ding, Yingfei Gu, Alan Guth, Aram Harrow, Linghang Kong, Hong Liu, Guang Hao Low, Yoshifumi Nakata, Kevin Thompson, Andreas Winter, Beni Yoshida, Quntao Zhuang, and Karol \.Zyczkowski for discussions related to this work or feedbacks on the draft. 
ZWL and SL are supported by AFOSR and ARO. EYZ is supported by the National Science Foundation under grant Contract Number CCF-1525130. HZ is supported by the Excellence Initiative of the German Federal and State Governments (ZUK~81) and the DFG in the early stage of this work. Research at MIT CTP is supported by DOE.
\end{acknowledgments}

\appendix

\section{Inequalities relating R\'enyi entropies of different orders}\label{app:renyiineq}
First, we present a series of inequalities relating R\'enyi entropies of different orders.
It is well known that the R\'enyi entropy is monotonically nonincreasing with the parameter $\alpha$, that is $ S^{(\alpha)}_R(\rho)\geq S^{(\beta)}_R(\rho)$ whenever $\beta\geq \alpha$. On the other hand, $S^{(\alpha)}_R(\rho)$ can also be used to construct a
lower bound for $S^{(\beta)}_R(\rho)$ when $\beta\geq \alpha \geq 1$ as shown below,
\begin{align}
S^{(\beta)}_R(\rho)&=-\frac{1}{\beta-1}\log\tr\{\rho^\beta\}=-\frac{\beta}{\beta-1}\log(\tr\{\rho^\beta\})^{1/\beta}\nonumber\geq -\frac{\beta}{\beta-1}\log(\tr\{\rho^\alpha\})^{1/\alpha}\\&=\frac{\beta}{\beta-1}\frac{\alpha-1}{\alpha}S^{(\alpha)}_R(\rho).
\end{align}
In particular, this equation yields a lower bound for the min entropy
\begin{equation}
S_{\min}(\rho)\geq \frac{\alpha-1}{\alpha}S^{(\alpha)}_R(\rho)=S^{(\alpha)}_R(\rho)-\frac{S^{(\alpha)}_R(\rho)}{\alpha}. 
\end{equation}
When  $\alpha\geq \log d$, we have  
\begin{equation}
S^{(\alpha)}_R(\rho)-1\leq S_{\min}(\rho)\leq S^{(\alpha)}_R(\rho),
\end{equation}
so the difference between $S^{(\alpha)}_R(\rho)$ and $S_{\min}(\rho)$ is less than 1. When $\beta=\alpha+1$, we have $S^{(\alpha+1)}_R(\rho)\geq\frac{\alpha^2-1}{\alpha^2}S^{(\alpha)}_R(\rho)$, so the difference between $S^{(\alpha+1)}_R(\rho)$ and $S^{(\alpha)}_R(\rho)$ is upper bounded by $S^{(\alpha)}_R(\rho)/\alpha^2$.

Next we derive another lower bound for $S^{(\beta)}_R(\rho)$ in terms of $S^{(\alpha)}_R(\rho)$ and the min entropy in the case  $\beta\geq \alpha \geq 1$. The following equation
\begin{align}
\tr\{\rho^\beta\}=\tr\left(\rho^\alpha\rho^{\beta-\alpha}\right)\leq \tr\{\rho^\alpha\} \|\rho\|^{\beta-\alpha}
\end{align}
implies that
\begin{equation}
S^{(\beta)}_R(\rho)\geq \frac{1}{\beta-1}[(\alpha-1)S^{(\alpha)}_R(\rho)+(\beta-\alpha)S_{\min}(\rho)].
\end{equation}
In particular, any R\'enyi $\beta$-entropy with $\beta\geq2$ is lower bounded by a convex combination of 
R\'enyi $2$-entropy and the min entropy,
\begin{equation}
S^{(\beta)}_R(\rho)\geq \frac{1}{\beta-1}[S^{(2)}_R(\rho)+(\beta-2)S_{\min}(\rho)].
\end{equation}

\section{Weak subadditivity of the R\'enyi entropies}\label{app:renyisub}
It is known that R\'enyi-$\alpha$ entropy is not subadditive except for the special case $\alpha=1$. The following lemma yields a weaker form of subadditivity:
\begin{lem}\label{lem:EntropyGap}
	Let $\rho_{AB}$ be any bipartite state on the product Hilbert space $\mathcal{H}_A\otimes \mathcal{H}_B$ with dimension $d_A\times d_B$. Let $\rho_A, \rho_B$ be the two reduced states.  Then 
	\begin{gather}
	\rho_{AB}\succ \rho_A\otimes\frac{I}{d_B},\\
	S^{(\alpha)}_R(\rho_{AB})\leq S^{(\alpha)}_R(\rho_{A})+\log d_B,\\
	\log(d_A d_B)-S^{(\alpha)}_R(\rho_{AB})\geq  \log d_A-S^{(\alpha)}_R(\rho_{A}).
	\end{gather}
\end{lem}
The first inequality in Lemma~\ref{lem:EntropyGap} means that the spectrum of $\rho_{AB}$ majorizes that of $\rho_A\otimes\frac{I}{d_B}$. 
The second and third  inequalities are immediate consequences of the first one, which are
are equivalent. The second one can be seen as a weaker form of subadditivity, while the third one means that the gap of R\'enyi entropy of a joint state from the maximum is no smaller than the corresponding  gap for each reduced state, which has already been discussed in a slightly different way.

\begin{proof}
	Let $|j\>$ for $j=1,2,\ldots, d_B$ be an orthonormal basis for $\mathcal{H}_B$ and $P_j=|j\>\<j|$ be the corresponding projectors. Let 
	\begin{equation}
	\sigma=\sum_j (I\otimes P_j)\rho_{AB}(I\otimes P_j)=\sum_j\rho_j\otimes P_j,
	\end{equation}
	where $\rho_j$ are subnormalized states that sum up to $\rho_A$.  Define
	\begin{equation}
	\sigma_k=\sum_j\rho_j\otimes P_{j+k},\quad k=1, 2, \cdots, d_B. 
	\end{equation}
	where the addition in the indices is modulo $d_B$; note that $\sigma_0=\sigma$. Then all $\sigma_k$ have the same spectrum, which is majorized by $\rho_{AB}$, that is, $\rho_{AB}\succ \sigma_k$. Consequently,
	\begin{equation}
	\rho_{AB}\succ \frac{1}{d_B}\sum_{k=0}^{d_B-1}\sigma_k=\rho_A\otimes\frac{I}{d_B}.
	\end{equation}
	Since R\'enyi $\alpha$-entropy is Schur concave for $0\leq \alpha\leq\infty$, it follows that 
	\begin{equation}
	S^{(\alpha)}_R(\rho_{AB})\leq S^{(\alpha)}_R\left(\rho_A\otimes\frac{I}{d_B}\right)=S^{(\alpha)}_R(\rho_{A})+\log d_B
	\end{equation}
	which confirms the second inequality in Lemma~\ref{lem:EntropyGap} and implies the third inequality.
	
\end{proof}

\section{Proof of the Cycle Lemma}\label{app:lem}
We include here an intuitive proof of Lemma \ref{sumcycle} (the Cycle Lemma), which plays a key role in our study, by induction. The intuition is that any element $ \sigma\in S_\alpha $ can be viewed as a local deformation of some element $ \sigma_-\in S_{\alpha-1} $, such that $ \xi(\sigma)+\xi(\sigma\tau)$ can only increase by at most 1. We formalize the argument below.

Suppose the statement is true for $\alpha=k$. That is, $ \xi(\sigma)+\xi(\sigma\tau) \leq k+1$ for all $\sigma\in S_k$.  Now for some $\sigma \in S_{k+1}$, look at element $k+1$. There are two possibilities:
\begin{enumerate}
\item $k+1$ appears in a 1-cycle (is mapped to itself): $\sigma[k+1]= k+1$.   So $\sigma = \sigma_-(k+1)$, for some $\sigma_-\in S_{k}$.

 $\xi(\sigma)$:  We directly see $\xi(\sigma) = \xi(\sigma_-)+1$.

 $\xi(\sigma\tau)$: Write $ \tau=(1~2~\cdots~ k+1)=\tau_-(k~k+1) $, where $ \tau_-=(1~2~\cdots~ k) $. Then $ \sigma\tau=\sigma_- (k+1)\tau_- (k~k+1)=\sigma_-\tau_-(k~k+1) $, with $ \sigma_-\tau_-\in S_k $. 
Now compare the action of $\sigma_-\tau_-$ and $\sigma\tau $ on individual elements. The only differences is $ \sigma\tau[k]=k+1 $ but $\sigma_-\tau_-[k]=\sigma_-\tau_-[k] $, and in addition $ \sigma\tau[k+1]=\sigma_-\tau_-[k] $.  So $\sigma\tau$ simply increases the length of a cycle in $\sigma_-\tau_-$ by one, and does nothing to other cycles. This is illustrated in Fig.~\ref{fig:cycle1}. 
	\begin{figure}
    \centering
    \includegraphics[scale=0.78]{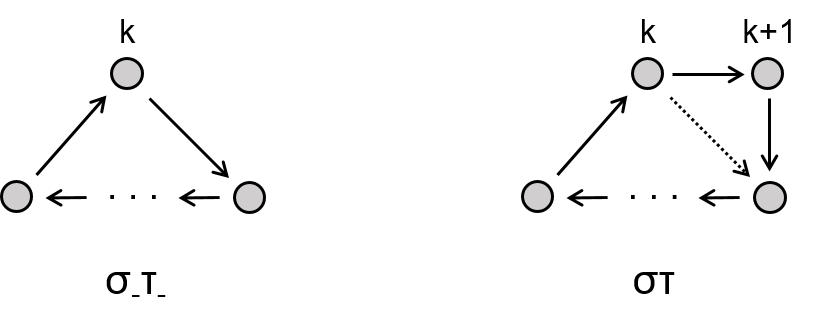}
    \caption{Comparison between $\sigma_-\tau_-$ and $\sigma\tau$, when $k+1$ is a 1-cycle in $\sigma$. Dashed arrows represent the mappings that are in $\sigma_-\tau_-$ but no longer there in $\sigma\tau$, and identical cycles are not shown.  We see that in $\sigma\tau$ the element $k+1$ is inserted in the cycle, but the total number of cycles does not change.}
    \label{fig:cycle1}
\end{figure}
So $ \xi(\sigma\tau)=\xi(\sigma_-\tau_-) $. 

From the induction hypothesis, $ \xi(\sigma_-)+\xi(\sigma_-\tau_-)\leq k+1 $, so $ \xi(\sigma)+\xi(\sigma\tau)=\xi(\sigma_-)+1+\xi(\sigma_-\tau_-)\leq k+2 $. Check.

\item $ k+1 $ appears in a cycle of length $ >1 $: $\sigma[a]=k+1 $, $\sigma[k+1]=b $ for some elements $a,b\in\{1,\dots,k\}$. Define $\sigma'_- \in S_k$ by $ \sigma'_-[i]=\sigma[i] $ for $ i\in \{1,\dots,k\}\setminus\{a\} $ and $ \sigma'_-[a]=b $.

$\xi(\sigma)$:  Clearly $\xi(\sigma) = \xi(\sigma'_-)$.

$\xi(\sigma\tau)$: Compare the action of $\sigma'_-\tau_-$ and $\sigma\tau$ on individual elements. Depending on the value of $ a $, there are two cases:
\begin{enumerate}
\item $ a\neq 1 $. The differences are: $\sigma'_-\tau_-[a-1]=b$ and $\sigma'_-\tau_-[k]=\sigma[1]$, but $\sigma\tau[a-1]=k+1$, $\sigma\tau[k]=b$, and in addition $\sigma\tau[k+1]=\sigma[1]$. They act identically on all other elements. There are two possible effects (see Fig.~\ref{fig:cycle2} for illustration):
\begin{enumerate}
	\item In $\sigma'_-\tau_-$, $ \{a-1,~ b\} $ and $ \{k,~\sigma[1]\} $ belong to the same cycle. Then $ \sigma\tau $ breaks this cycle into two disjoint ones involving $ \{a-1,~k+1,~\sigma[1]\} $ and $ \{k,~b\} $ respectively.  So $ \xi(\sigma\tau)=\xi(\sigma'_-\tau_-)+1 $;
	\item In $\sigma'_-\tau_-$, $ \{a-1,~ b\} $ and $ \{k,~\sigma[1]\} $ belong to two disjoint cycles. Then $ \sigma\tau $ glues these two cycles together into one.
	\begin{figure}[t]
    \centering
    \includegraphics[scale=0.78]{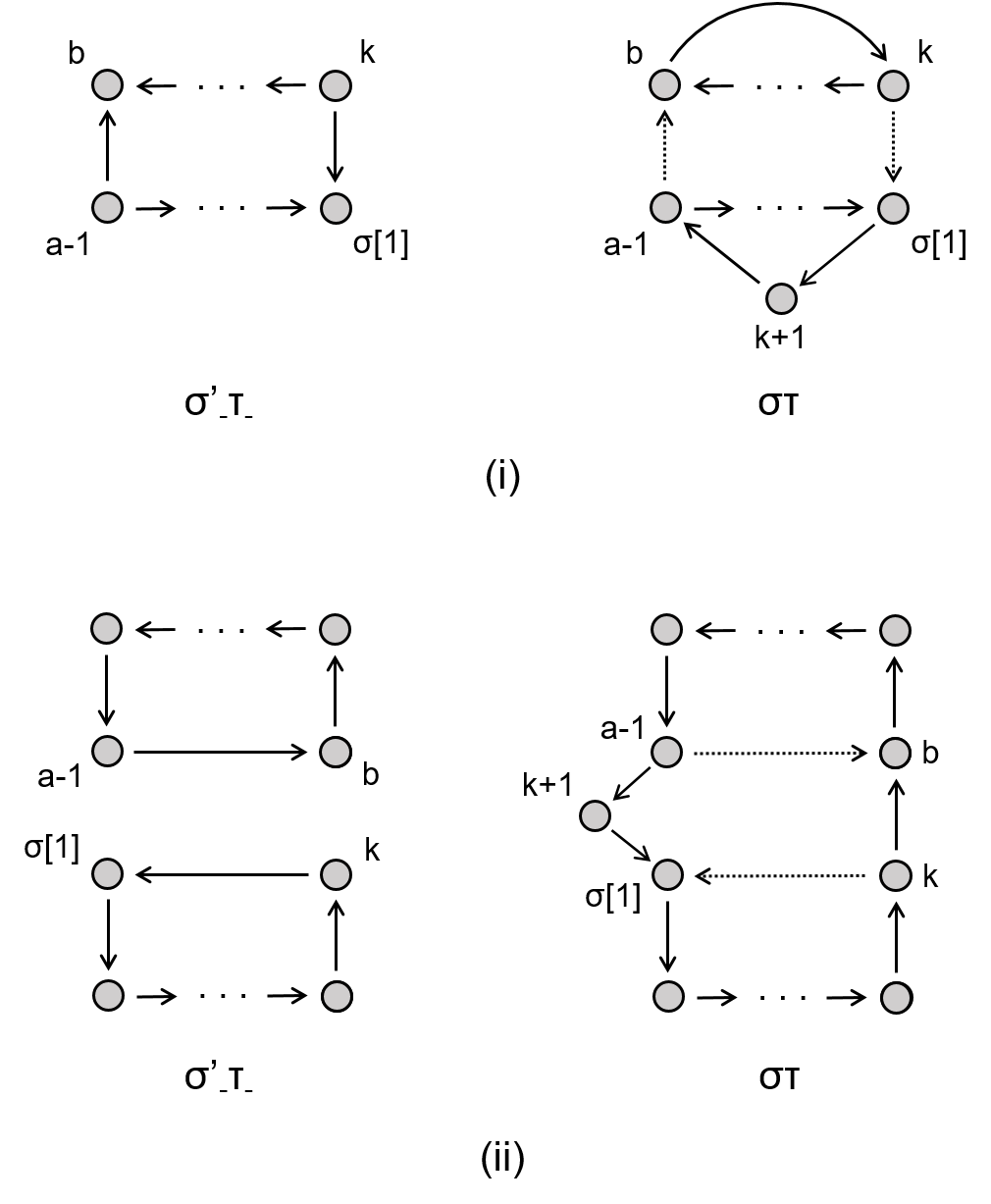}
    \caption{Comparison between $\sigma'_-\tau_-$ and $\sigma\tau$, when $k+1$ is in a cycle of length $>1$ in $\sigma$. Dashed arrows represent the mappings that are in $\sigma'_-\tau_-$ but no longer there in $\sigma\tau$, and identical cycles are not shown.  There are two possible cases: (i) The relevant elements $a-1, b, \sigma[1], k$ belong to the same cycle in $\sigma'_-\tau_-$. In $\sigma\tau$, this cycle is broken into two, so $\sigma\tau$ has one more cycle than $\sigma'_-\tau_-$; (ii) $ \{a-1,~ b\} $ and $ \{k,~\sigma[1]\} $ belong to two cycles in $\sigma'_-\tau_-$. In $\sigma\tau$, these two cycles are combined as one with element $k+1$ inserted, so $\sigma\tau$ has one less cycle than $\sigma'_-\tau_-$.}
    \label{fig:cycle2}
\end{figure}
	So $\xi(\sigma\tau)=\xi(\sigma'_-\tau_-)-1$.
\end{enumerate}
\item $ a=1 $. Then $ \sigma'_-\tau_- $ and $ \sigma\tau $ act identically on $ \{1,\dots,k\} $ and in addition $ \sigma\tau[k+1]=k+1 $. So $ \xi(\sigma\tau)=\xi(\sigma'_-\tau_-)+1 $.
\end{enumerate}

In conclusion, $\xi(\sigma\tau)$ can only increase by one or decrease by one as compared to $\xi(\sigma'_-\tau_-)$,
so $ \xi(\sigma)+\xi(\sigma\tau)=\xi(\sigma_-)+\xi(\sigma'_-\tau_-)\pm 1\leq k+2 $ in either case. Check.

\end{enumerate}

Lastly, consider $k=1$. The only element of $S_1$ is $(1)$, and $\xi((1))+\xi((1)(1))=2\leq k+1$, so the statement trivially holds. This completes our proof. 

\section{Bounds on the Catalan numbers}\label{app:cat}
It is well known that the  Catalan number $\cat_k=(2k)!/[k!(k+1)!]$ is approximated by ${4^k}/{\sqrt{\pi}k^{3/2}}$ when $k$ is large. To make this statement more precise, here we provide both lower and upper bounds for $\cat_k$. 
\begin{lem}\label{lem:CatalanBound}
	The  Catalan number $\cat_k$ satisfies 
	\begin{equation}
	\frac{4^k}{\sqrt{\pi}(k+1)^{3/2}}<\cat_k <\frac{4^k}{\sqrt{\pi}k^{3/2}}\quad \forall k\geq1,
	\end{equation}
	where $k$ is not necessarily an integer. 
\end{lem}

\begin{proof}
	The basis of our proof is the following Stirling approximation formula
	\begin{equation}
	\sqrt{2\pi}k^{k+\frac{1}{2}}\rme^{-k}\leq k!\leq \sqrt{2\pi}k^{k+\frac{1}{2}}\rme^{-k}\rme^{\frac{1}{12k}}. 
	\end{equation}
	As an implication,
	\begin{align}
	\cat_k&\leq \frac{\sqrt{2\pi}(2k)^{2k+\frac{1}{2}}\rme^{-2k}\rme^{\frac{1}{24k}}}{\sqrt{2\pi}k^{k+\frac{1}{2}}\rme^{-k} \sqrt{2\pi}(k+1)^{k+\frac{3}{2}}\rme^{-k-1}}=
	\frac{2^{2k+\frac{1}{2}}
		k^k\rme^{1+\frac{1}{24k}}}{ \sqrt{2\pi}(k+1)^{k+\frac{3}{2}} }\nonumber\\
	&=
	\frac{2^{2k}
		\rme^{1+\frac{1}{24k}}}{ \sqrt{\pi}k^{\frac{3}{2}}(1+\frac{1}{k})^{k+\frac{3}{2}} }< \frac{2^{2k}
		\rme^{\frac{1}{24k}}}{ \sqrt{\pi}k^{\frac{3}{2}}(1+\frac{1}{k}) }< \frac{4^{k}}{ \sqrt{\pi}k^{\frac{3}{2}}}.
	\end{align}
	Here  the second inequality follows from the inequality 
	\begin{equation}
	\Bigl(1+\frac{1}{k}\Bigr)^{k+\frac{1}{2}}> \rme,
	\end{equation}
	note that the left hand side is monotonically decreasing with $k$ and approaches  $\rme$ in the limit $k\rightarrow \infty$.

	On the other hand,
	\begin{align}
	\cat_k&\geq \frac{\sqrt{2\pi}(2k)^{2k+\frac{1}{2}}\rme^{-2k}}{\sqrt{2\pi}k^{k+\frac{1}{2}}\rme^{-k}\rme^{\frac{1}{12k}} \sqrt{2\pi}(k+1)^{k+\frac{3}{2}}\rme^{-k-1} \rme^{\frac{1}{12(k+1)}}}=
	\frac{2^{2k}
		k^k\rme}{ \sqrt{\pi}(k+1)^{k+\frac{3}{2}} \rme^{\frac{1}{12k}+ \frac{1}{12(k+1)}}} \nonumber\\
	&=
	\frac{4^{k}\rme}{ \sqrt{\pi}(k+1)^{\frac{3}{2}} (1+\frac{1}{k})^k \rme^{\frac{1}{12k}+ \frac{1}{12(k+1)}}}>\frac{4^{k}}{ \sqrt{\pi}(k+1)^{\frac{3}{2}} }. 
	\end{align}
	Here the last inequality follows from the inequality
	\begin{equation}
	\Bigl(1+\frac{1}{k}\Bigr)^k \rme^{\frac{1}{12k}+ \frac{1}{12(k+1)}}< \rme.
	\end{equation}
	To confirm this claim, we shall prove the equivalent inequality 
	\begin{equation}
	f(k):=\ln\left[\Bigl(1+\frac{1}{k}\Bigr)^k \rme^{\frac{1}{12k}+ \frac{1}{12(k+1)}}\right]< 1.
	\end{equation}
	The first and second derivatives of $f(k)$ read
	\begin{align}
	f'(k)&=\ln\Bigl(\frac{k+1}{k}\Bigr)-\frac{1}{k+1}-\frac{1}{12k^2}-\frac{1}{12(k+1)^2},\nonumber\\
	f''(k)&=-\frac{1}{k(k+1)}+\frac{1}{(k+1)^2}+\frac{1}{6k^3}
	+\frac{1}{6(k+1)^3}=-\frac{1}{6k^3(k+1)^3}(4k^3+3k^2-3k-1)<0.
	\end{align}
	Since $f''(k)$ is negative,  $f'(k)$ is monotonically decreasing, which implies that $f'(k)>0$ given
	that $\lim_{k\rightarrow \infty}f'(k)=0$. Consequently, $f(k)$ is monotonically increasing,
	which confirms our claim $f(k)<1$
	given that $\lim_{k\rightarrow \infty}f(k)=1$.
\end{proof}

The following two corollaries are easy consequences of Lemma~\ref{lem:CatalanBound}, though it is straightforward to prove them directly.
\begin{cor}\label{cor:CatalanMono}
	$\cat_k\leq \cat_{k+1}$ for any positive integer $k$. 
\end{cor}
\begin{proof}
	The corollary holds for $k=0,1$ by direct calculation. When $k\geq2$, Lemma~\ref{lem:CatalanBound} implies that
	\begin{equation}
	\frac{\cat_{k+1}}{\cat_k}\geq\frac{4 k^{3/2}}{(k+2)^{3/2}}\geq \frac{4}{2^{3/2}}=\sqrt{2}>1,
	\end{equation}
	which confirms the corollary.
\end{proof}

\begin{cor}\label{cor:CatalanSupMul}
	$\cat_j\cat_k< \cat_{j+k}$ for arbitrary positive  integers $j,k$.
\end{cor}
\begin{proof}
	The corollary holds when $j=1$ or $k=1$ according to 	Corollary~\ref{cor:CatalanMono}, given that $\cat_1=1$. When $j,k\geq2$, Lemma~\ref{lem:CatalanBound} implies that
	\begin{equation}
	\frac{\cat_j\cat_k}{\cat_{j+k}}< \frac{(j+k+1)^{3/2}}{\sqrt{\pi}j^{3/2}k^{3/2}}<1.
	\end{equation}
\end{proof}

\section{Bounds on the M\"obius function}\label{app:moeb}
Recall the definition of the M\"obius function,
\begin{equation}
\mathrm{Moeb}(\sigma):= \prod_{j=1}^k (-1)^{|C_j|}\cat_{|C_j|}=(-1)^{|\sigma|}\prod_{j=1}^k \cat_{|C_j|}.
\end{equation}

\begin{lem}\label{lem:MobiusBound}
	\begin{equation}
	1\leq |\mathrm{Moeb}(\sigma)|\leq \cat_{|\sigma|}< \frac{4^{|\sigma|}}{\sqrt{\pi}|\sigma|^{3/2}}\quad \forall |\sigma|\geq 1.
	\end{equation}
	The lower bound is saturated iff $\sigma$ is the identity or a product of disjoint transpositions. 
	The upper bound $|\mathrm{Moeb}(\sigma)|\leq \cat_{|\sigma|}$ is saturated iff $\sigma$ is a cycle of length $|\sigma|+1$.
\end{lem}

\begin{proof}
	The lemma holds when $\sigma$ is the identity. Otherwise,
	suppose $\sigma$ has disjoint cycle decomposition $\sigma=C_1 C_2\cdots  C_k$,  where $C_j$ for $1\leq j\leq k$ are nontrivial cycles. Then 
	\begin{equation}
	|\mathrm{Moeb}(\sigma)|=\prod_{j=1}^k \cat_{|C_j|}\geq 1
	\end{equation}	
	given that 	$\cat_{|C_j|}\geq1$ for all $j$.
	The inequality is saturated iff $|C_j|=1$ for all $j$, that is, $\sigma$ is  a product of disjoint transpositions. On the other hand,
	\begin{equation}
	|\mathrm{Moeb}(\sigma)|=\prod_{j=1}^k \cat_{|C_j|}\leq \cat_{\sum_ j|C_j|}=\cat_{|\sigma|}<\frac{4^{|\sigma|}}{\sqrt{\pi}|\sigma|^{3/2}},
	\end{equation}
	where the two inequalities follow from  Corollary~\ref{cor:CatalanSupMul} and Lemma~\ref{lem:CatalanBound}, respectively. The first inequality is saturated when $k=1$, but is strict whenever $k\geq 2$. So the upper bound $|\mathrm{Moeb}(\sigma)|\leq \cat_{|\sigma|}$ is saturated iff $\sigma$ is a cycle of length $|\sigma|+1$.
\end{proof}

\section{Bounds on the Weingarten function}\label{app:weingarten}
The following theorem is reproduced from \cite{CollM17},
\begin{thm}\label{thm:WgBoundCM}
	When $d>\sqrt{6}k^{7/4}$, any $\sigma\in S_k$ satisfies 
	\begin{equation}
	\frac{1}{1-\frac{k-1}{d^2}}\leq \frac{d^{k+|\sigma|}\wg(\sigma,d)}{\mathrm{Moeb}(\sigma)}\leq a_k:=\frac{1}{1-\frac{6k^{7/2}}{d^2}}.
	\end{equation}
\end{thm}

The following lemma is  an immediate consequence of Theorem~\ref{thm:WgBoundCM} and Lemma~\ref{lem:MobiusBound}. 
\begin{lem}\label{lem:WgBound}
	When $d>\sqrt{6}k^{7/4}$,  any $\sigma\in S_k$  satisfies
	\begin{equation}
	d^{k}|\wg(\sigma,d)|\leq \begin{cases}
	a_k(\frac{1}{d})^{|\sigma|} &|\sigma|=0,1,\\
	\min\left\{\frac{a_k(\frac{4}{d})^{|\sigma|}}{\sqrt{\pi}|\sigma|^{3/2} },\; \frac{a_k(\frac{4}{d})^{|\sigma|}}{8}\right\} &|\sigma|\geq2,
	\end{cases}
	\end{equation}
	where $a_k$ is defined in Theorem~\ref{thm:WgBoundCM}.
\end{lem}

\begin{lem}\label{lem:WgSumBound}
	Suppose $d>\sqrt{6}k^{7/4}$; then
	\begin{equation}
	\sum_{\sigma\in A_k}d^{k}\wg(\sigma,d)\leq\frac{a_k}{8}\left[7+\cosh\frac{2k(k-1)}{d}\right].
	\end{equation}
\end{lem}

\begin{proof}
	According to Lemma~\ref{lem:WgBound},
	\begin{align}
	&\sum_{\sigma\in A_k} d^{k}\wg(\sigma,d)\leq a_k + \sum_{\sigma\in A_k\; |\sigma|\geq 2} d^{k}\wg(\sigma,d)\leq a_k + \frac{a_k}{8}\sum_{\sigma\in A_k\; |\sigma|\geq 2} \Bigl(\frac{4}{d}\Bigr)^{|\sigma|}\nonumber\\
	&\leq \frac{7a_k}{8} + \frac{a_k}{8}\sum_{\sigma\in A_k} \Bigl(\frac{4}{d}\Bigr)^{|\sigma|}
	= \frac{7a_k}{8} + \frac{a_k}{8}\sum_{\sigma\in A_k} \Bigl(\frac{4}{d}\Bigr)^{k-\xi(\sigma)}
	\nonumber\\
	&= \frac{7a_k}{8} + \frac{a_k}{8}\Bigl(\frac{4}{d}\Bigr)^{k}\sum_{\sigma\in A_k} \Bigl(\frac{d}{4}\Bigr)^{\xi(\sigma)}\nonumber\\
	&= \frac{7a_k}{8} + \frac{a_k}{16}\Bigl(\frac{4}{d}\Bigr)^{k}\left[\sum_{\sigma\in S_k} \Bigl(\frac{d}{4}\Bigr)^{\xi(\sigma)}+\sum_{\sigma\in S_k}\Bigl(-\frac{d}{4}\Bigr)^{\xi(\sigma)}\right]\nonumber\\
	&= \frac{7a_k}{8} + \frac{a_k}{16}\Bigl(\frac{4}{d}\Bigr)^{k}\left[\prod_{j=0}^{k-1}\Bigl(\frac{d}{4}+j\Bigr) +\prod_{j=0}^{k-1}\Bigl(\frac{d}{4}-j\Bigr)  \right]\nonumber\\
	&= \frac{7a_k}{8} + \frac{a_k}{16}\left[\prod_{j=0}^{k-1}\Bigl(1+\frac{4j}{d}\Bigr) +\prod_{j=0}^{k-1}\Bigl(1-\frac{4j}{d}\Bigr)  \right]\nonumber\\
	&\leq\frac{7a_k}{8} + \frac{a_k}{16}\left[\prod_{j=0}^{k-1}\rme^{4j/d} +\prod_{j=0}^{k-1}\rme^{-4j/d}\Bigr)  \right]
	=\frac{7a_k}{8} + \frac{a_k}{16}\left[\rme^{\sum_{j=0}^{k-1}4j/d} +\rme^{-\sum_{j=0}^{k-1}4j/d}\Bigr)  \right]\nonumber\\
	&=\frac{7a_k}{8} + \frac{a_k}{16}\left[\rme^{2k(k-1)/d} +\rme^{-2k(k-1)/d}\right]=\frac{a_k}{8}\left[7+\cosh\frac{2k(k-1)}{d}\right].
	\end{align}
\end{proof}

\section{Bounds on the number of permutations with a given genus}\label{app:genus}
In this appendix, we provide an easy-to-use upper bound for the 
 number of  permutations with a given genus (Lemma~\ref{lem:NGpermutationT} below), which plays a crucial role in understanding  R\'enyi entanglement entropies of Haar random states as well as states drawn from designs. 
 
 The basis of our endeavor is the following theorem due to Goupil and Schaeffer \cite{GoupS98}.
\begin{thm}\label{thm:NGpermutationGS}
The number of permutations in the symmetric group $S_n$ with genus $g$ is given by
\begin{equation}
c_{g,n}=\frac{(n+1)_{2g}}{(n+1)2^{2g}}\sum_{g_1+g_2=g}\sum_{\substack{0\leq \ell_1\leq g_1\\ 0\leq \ell_2\leq g_2}} a_{g_1,\ell_1}a_{g_2,\ell_2}(n+1-2g)_{\ell_1+\ell_2}\binom{2n-2g-\ell_1-\ell_2}{n-2g_1-\ell_1},
\end{equation}
where $(n)_k:=n(n-1)\cdots, (n-k+1)$, $a_{0,0}=1$, $a_{g,0}=0$ for $g\geq1$, and 
\begin{equation}
a_{g,\ell}=\sum_{\substack{\gamma\vdash g,\;\ell(\gamma)=\ell \\ \gamma=1^{c_1}2^{c_2} g^{c_g}}}
\frac{1}{\prod_{j=1}^gc_j!(2j+1)^{c_j}}\quad  0<\ell\leq g.
\end{equation}
Here the summation runs over all partition  $\gamma$  of $g$, the expression $\gamma=1^{c_1}2^{c_2} g^{c_g}$ means that $\gamma$ has $c_j$ parts equal to $j$, and $\ell(\gamma)=\sum_j c_j$ denotes the number of parts of $\gamma$.
\end{thm}

In addition, we need two auxiliary lemmas.
\begin{lem}\label{lem:partitionS}
$a_{g,\ell}\leq 2^{-\ell}$ for all $0\leq \ell\leq g$.
\end{lem}
\begin{proof}
By definition, the lemma holds when $g=0$, or $g\geq1$ and $\ell=0$. Now suppose $0<\ell\leq g$; then 
\begin{align}
a_{g,\ell}&=\sum_{\substack{\gamma\vdash g,\;\ell(\gamma)=\ell \\ \gamma=1^{c_1}2^{c_2} g^{c_g}}}
\frac{1}{\prod_{j=1}^gc_j!(2j+1)^{c_j}} 
=\sum_{\substack{\gamma\vdash g,\;\ell(\gamma)=\ell \\ \gamma=1^{c_1}2^{c_2} g^{c_g}}}\left[
\frac{1}{\prod_{j=1}^gc_j!j^{c_j}} 
\prod_{j=1}^g \Bigl(\frac{j}{2j+1}\Bigr)^{c_j}\right]\nonumber\\
&\leq \sum_{\substack{\gamma\vdash g,\;\ell(\gamma)=\ell \\ \gamma=1^{c_1}2^{c_2} g^{c_g}}}\left[
\frac{1}{\prod_{j=1}^gc_j!j^{c_j}} 
\prod_{j=1}^g \Bigl(\frac{1}{2}\Bigr)^{c_j}\right]=\sum_{\substack{\gamma\vdash g,\;\ell(\gamma)=\ell \\ \gamma=1^{c_1}2^{c_2} g^{c_g}}}\left[
\frac{1}{\prod_{j=1}^gc_j!j^{c_j}} 
2^{-\sum_{j=1}^g c_j}\right]\nonumber\\
&=\sum_{\substack{\gamma\vdash g,\;\ell(\gamma)=\ell \\ \gamma=1^{c_1}2^{c_2} g^{c_g}}}
\frac{1}{\prod_{j=1}^gc_j!j^{c_j}} 
2^{-\ell(\gamma)}
=2^{-\ell}\sum_{\substack{\gamma\vdash g,\;\ell(\gamma)=\ell \\ \gamma=1^{c_1}2^{c_2} g^{c_g}}}
\frac{1}{\prod_{j=1}^gc_j!j^{c_j}} \leq 2^{-\ell}.
\end{align}
Here the last inequality can be derived as follows. Note that $\prod_{j=1}^gc_j!j^{c_j}$ is the order of the centralizer in $S_g$ of each element in the conjugacy class labeled by the partition $\gamma$. Therefore, $g!/\prod_{j=1}^gc_j!j^{c_j}$ is the number of elements in this conjugacy class, so that
\begin{equation}
\sum_{\substack{\gamma\vdash g \\ \gamma=1^{c_1}2^{c_2} g^{c_g}}}
\frac{g!}{\prod_{j=1}^gc_j!j^{c_j}}=g!,
\end{equation}
which amounts  to the identity
\begin{equation}
\sum_{\substack{\gamma\vdash g \\ \gamma=1^{c_1}2^{c_2} g^{c_g}}}
\frac{1}{\prod_{j=1}^gc_j!j^{c_j}}=1.
\end{equation}
As an immediate consequence,
\begin{equation}
\sum_{\substack{\gamma\vdash g,\;\ell(\gamma)=\ell \\ \gamma=1^{c_1}2^{c_2} g^{c_g}}}
\frac{1}{\prod_{j=1}^gc_j!j^{c_j}}\leq 1.
\end{equation}
\end{proof}

\begin{lem}\label{lem:BinomRatio}
Suppose $j,k,n$ are nonnegative  integers satisfying $j\leq n$, $k<2n$, and $k\leq n+j$. Then
\begin{equation}
\binom{2n-k}{n-j}\leq 2^{-k}\sqrt{\frac{n}{n-\lfloor k/2\rfloor}}\binom{2n}{n}.
\end{equation}
\end{lem}
\begin{proof}
Straightforward calculation shows that 
\begin{equation}
\binom{2n-k}{n-j}\leq \binom{2n-k}{n-\lfloor k/2\rfloor}=\binom{2n-k}{n-\lceil k/2\rceil}.
\end{equation}
So without loss of generality, we may assume that $j=\lfloor k/2\rfloor$. Then
\begin{align}
\frac{\binom{2n}{n}}{\binom{2n-k}{n-j}}&=\frac{(2n)!(n-j)!(n+j-k)!}{(2n-k)!n!n!}=\frac{2n(2n-1)\cdots (2n-k+1)}{[n(n-1) \cdots n-j+1] [n(n-1)\cdots (n+j-k+1)]}\nonumber\\
&=\frac{2^kn(n-\frac{1}{2})\cdots (n-\frac{k}{2}+\frac{1}{2})}{[n(n-1) \cdots n-j+1] [n(n-1)\cdots (n+j-k+1)]  }=2^kf,
\end{align}
where 
\begin{equation}
f=\frac{(n-\frac{1}{2})(n-\frac{3}{2})\cdots (n-j+\frac{1}{2})}{n(n-1)\cdots (n-j+1)  }.
\end{equation}
The square of $f$ can be bounded from below as follows,
\begin{align}
f^2&=\frac{(n-\frac{1}{2})^2(n-\frac{3}{2})^2\cdots (n-j+\frac{1}{2})^2}{ n^2(n-1)^2\cdots (n-j+1)^2 }\nonumber\\
&=\frac{1}{n}\times \frac{(n-\frac{1}{2})^2}{ n(n-1) }\times \cdots\times 
\frac{
	 (n-j+\frac{3}{2})^2}{ (n-j+2)  (n-j+1) } \times 
\frac{
	 (n-j+\frac{1}{2})^2}{ n-j+1 }\nonumber\\
&\geq \frac{1}{n} \times
	\frac{
		(n-j+\frac{1}{2})^2}{ n-j+1 }\geq \frac{n-j}{n}=\frac{n-\lfloor \frac{k}{2}\rfloor}{n}.
\end{align}
Therefore $f\geq \sqrt{\frac{n-\lfloor k/2\rfloor}{n}} $, from which the lemma follows. 
\end{proof}

\begin{lem}\label{lem:NGpermutationT}
	\begin{equation}
\frac{c_{g,n}}{c_{0,n}}	\leq \frac{(g+1)n^{3g}}{2^{6g}},\quad 
	\frac{c_{g,n}}{c_{0,n}}\leq \frac{2}{3}\left(\frac{n^3}{32}\right)^g\quad \forall 1\leq g\leq \frac{n-1}{2}.
	\end{equation}	
\end{lem}

\begin{proof}
Recall that $c_{0,n}=c_n=(2n)!/[n!(n+1)!]$.  The values of  $c_{1,n}, c_{2,n}$ can be computed explicitly according to Theorem~\ref{thm:NGpermutationGS}, with the result
\begin{align}
c_{1,n}&=\frac{n(n-1)}{6}\binom{2n-3}{n}=\frac{(2n-3)!}{6(n-2)!(n-3)!},\\
c_{2,n}&=\frac{(2n-5)!(5n^2-7n+6)}{720(n-3)!(n-5)!}. 
\end{align}
The coefficients $a_{g,\ell}$ necessary for deriving this result are given by 
\begin{equation}
a_{00}=1, \quad a_{1,1}=\frac{1}{3},\quad  a_{2,1}=\frac{1}{5},\quad a_{2,2}=\frac{1}{18}.
\end{equation}
As a consequence,
\begin{align}
\frac{c_{1,n}}{c_{0,n}}&=\frac{n(n+1)(n-1)(n-2)}{24(2n-1)}\leq \frac{n^3}{48},\\
\frac{c_{2,n}}{c_{0,n}}&=\frac{(n+1)n(n-1)(n-2)(n-3)(n-4)(5n^2-7n+6)}{5760(2n-1)(2n-3)}\leq \frac{n^6}{4608}. 
\end{align}
Therefore, Lemma~\ref{lem:NGpermutationT} holds when $g=1,2$. Now suppose $g\geq3$, so that  $n\geq 7$. According to Theorem~\ref{thm:NGpermutationGS}, we have 
\begin{align}
\frac{c_{g,n}}{c_{0,n}}&=\frac{(n+1)_{2g}}{2^{2g}}\sum_{g_1+g_2=g}\sum_{\substack{0\leq \ell_1\leq g_1\\ 0\leq \ell_2\leq g_2}} a_{g_1,\ell_1}a_{g_2,\ell_2}(n+1-2g)_{\ell_1+\ell_2}\frac{\binom{2n-2g-\ell_1-\ell_2}{n-2g_1-\ell_1}}{\binom{2n}{n}}\nonumber\\
&\leq \frac{(n+1)_{2g}}{2^{2g}}\sum_{g_1+g_2=g}\sum_{\substack{0\leq \ell_1\leq g_1\\ 0\leq \ell_2\leq g_2}} a_{g_1,\ell_1}a_{g_2,\ell_2}(n+1-2g)_{\ell_1+\ell_2}\times 2^{-(2g+\ell_1+\ell_2)}\sqrt{\frac{n}{n-g-\lfloor (\ell_1+\ell_2)/2\rfloor}}\nonumber\\
&\leq \frac{(n+1)_{2g}}{2^{4g}}\sum_{g_1+g_2=g}\sum_{\substack{0\leq \ell_1\leq g_1\\ 0\leq \ell_2\leq g_2}} 2^{-(\ell_1+\ell_2)} a_{g_1,\ell_1}a_{g_2,\ell_2} \frac{n(n+1-2g)_{\ell_1+\ell_2}}{n-g-\lfloor (\ell_1+\ell_2)/2\rfloor}\nonumber\\
&\leq \frac{(n+1)_{2g}}{2^{4g}}\sum_{g_1+g_2=g}\sum_{\substack{0\leq \ell_1\leq g_1\\ 0\leq \ell_2\leq g_2}} 4^{-(\ell_1+\ell_2)}n^{\ell_1+\ell_2} \frac{\max\{0,n+2-2g-(\ell_1+\ell_2)\}}{n-g-\lfloor (\ell_1+\ell_2)/2\rfloor}.
\end{align}
Here the first inequality follows from  Lemma~\ref{lem:BinomRatio}, and the last one from Lemma~\ref{lem:partitionS} and the fact that $a_{g,0}=0$ for $g>0$. The fraction at the end of the above equation is no larger than 1 given that $g\geq3$. Therefore,
\begin{align}
\frac{c_{g,n}}{c_{0,n}}&\leq \frac{(n+1)_{2g}}{2^{4g}}\sum_{g_1+g_2=g}\sum_{\substack{0\leq \ell_1\leq g_1\\ 0\leq \ell_2\leq g_2}} \left(\frac{n}{4}\right)^{-(\ell_1+\ell_2)} 
=\frac{(n+1)_{2g}}{2^{4g}}\sum_{g_1+g_2=g}
\frac{[\left(\frac{n}{4}\right)^{g_1+1}-1]}{\frac{n}{4}-1}\frac{[\left(\frac{n}{4}\right)^{g_2+1}-1]}{\frac{n}{4}-1}
\nonumber\\
&\leq \frac{(n+1)_{2g}}{2^{4g}} \frac{1}{(\frac{n}{4}-1)^2}  \sum_{g_1+g_2=g}
\left(\frac{n}{4}\right)^{g+2}=\frac{(n+1)_{2g}}{2^{4g}} \frac{(g+1)\left(\frac{n}{4}\right)^{g+2}}{(\frac{n}{4}-1)^2}  \nonumber\\
&=\frac{(g+1)n^{g+2}(n+1)_{2g}}{2^{6g}(n-4)^2} \leq \frac{(g+1)n^{3g-3} (n+1)(n-1)(n-2)(n-3)(n-4)}{2^{6g}(n-4)^2}\nonumber\\ &\leq \frac{(g+1)n^{3g}}{2^{6g}}.
\end{align}
This result confirms the first inequality in Lemma~\ref{lem:NGpermutationT} in the remaining case $g\geq 3$, which in turn implies the second inequality in the lemma.
\end{proof}


\section{Partially scrambling unitary}\label{app:partial}

Here we analyze the partially scrambling unitary model proposed in \cite{ding}, which can lead to a large separation between von Neumann and R\'enyi-2 entanglement entropies and tripartite information in the Choi state setting. More explicitly, let $\tilde{U}$ be a unitary that perfectly scrambles on almost the whole space besides a small subspace. Then, on the one hand, $\tilde{U}$ still has nearly maximal $-I_3$ due to continuity; while on the other hand, $-I_3^{(2)}$ can be gapped from maximum by $\Theta(\log d)$.
However, we find that this model is not likely to provide strict separations between R\'enyi entropies of order $\geq 2$.

The generalized partially scrambling unitary is defined as follows. Given $\alpha$, define
\begin{equation}
\tilde{U}\ket{mo}=\begin{cases}
    U_S\ket{mo}       & \quad 0\leq m,o<D\\
    \ket{mo}  & \quad \text{otherwise}\\
  \end{cases}
\end{equation}
where $U_S$ is $\alpha$-scrambling, and $D\leq \sqrt{d}$ controls the size of this $\alpha$-scrambling subspace (labeled by subscript $S$).
Then the Choi state of $\tilde{U}$ is
\begin{equation}
\ket{\tilde{U}}=\frac{D}{\sqrt{d}}\ket{U_S}_{A_SB_SC_SD_S}+\frac{1}{\sqrt{d}}\sum_{D\leq m<\sqrt{d}\wedge D\leq o<\sqrt{d}}\ket{mo}_{AB}\otimes\ket{mo}_{CD}.
\end{equation}

The question is whether there exists some $D$ that can lead to separations between higher R\'enyi entropies associated with this Choi state, say $\alpha$ and $\alpha'$, $\alpha'>\alpha\geq 2$. To establish such separations, we need to show a large ($\Theta(\log d)$) gap between R\'enyi-$\alpha'$ entropies and the maximum for some small $D$, as well as upper bound the difference between R\'enyi-$\alpha$ entropies and the maximum by continuity.  The gap side can work out by directly generalizing the corresponding calculation in \cite{ding}: Let $\beta = \log(\sqrt{d}-D)/\log{\sqrt{d}}$. Then $\log{d}-S_R^{(\alpha')}(\mathrm{tr}_{BD}\ketbra{\tilde{U}}{\tilde{U}})=\Theta(\log{d})$ as long as $\beta$ is a positive constant.
However, we find that the continuity bound for unified entropies can only give trivial results on the continuity side: 
\begin{lem}[Generalized Fannes' inequality \cite{Rastegin2011}]
Let $\rho$ and $\rho'$ be density operators in Hilbert space of dimension $d$. Denote $\epsilon=D_{\mathrm{tr}}(\rho,\rho')$. For $\alpha>1$ and $s\geq 0$:
\begin{equation}
|S^{(\alpha)}_s(\rho) - S^{(\alpha)}_s(\rho')| \leq \chi_s [\epsilon^\alpha\log_\alpha(d-1) + H^{(\alpha)}(\epsilon, 1-\epsilon)],
\end{equation}
where $\chi_s = 1$ for $s\geq 1$, and $\chi_s=d^{2(\alpha-1)}$ for $s=0$. $H^{(\alpha)}$ denotes the $\alpha$ binary entropy.
\end{lem}
It can be seen that this generalized Fannes' bound for R\'enyi entropies grows with the dimension $d$ for $\alpha>1$, which indicates that even a tiny non-scrambling subspace may perturb the R\'enyi entropies drastically. Indeed, some simple scaling analysis can confirm that this bound is trivial even for the R\'enyi-2 entropy. Notice that $\epsilon=D_{\mathrm{tr}}(\mathrm{tr}_{BD}\ketbra{\tilde{U}}{\tilde{U}},I)\leq O(D/\log{\sqrt{d}})= O(d^{(\beta -1)/2})$. Then it must hold that $2(\alpha -1 )+\alpha(\beta-1)/{2}<0$ so that $\log d-S_R^{(\alpha)}(\mathrm{tr}_{BD}\ketbra{\tilde{U}}{\tilde{U}})=o(\log d)$. This gives $\beta<-3+4/\alpha$, which has no overlap with the $\beta>0$ solution on the gap side when $\alpha\geq 2$.  Equivalently, by plugging in $\beta>0$ we can solve that the desired separation can exist when $\alpha'<2$. Summarizing, in order to have a nontrivial bound on R\'enyi entropies $D$ needs to be $o(1)$, which is meaningless.
This is hardly surprising: one expects that R\'enyi entropies are very sensitive, especially in the near-maximum regime, due to the logarithm.  In fact, we are able to obtain a large gap on the $\alpha'$ side basically because of such exponential sensitivity. Suppose we consider $s>0$ entropies instead. Then the continuity bound is strong since $\chi_s=1$, but it becomes hard to find a gap on the other side. There is a fundamental tradeoff between sensitivity and robustness in these unified entropies. In conclusion, we believe that partially scrambling unitaries are not likely to produce separations between generalized entropies in the Choi model.   

\section{Proof of Lemma~\ref{lem:AveRootNorm}}\label{app:lem24}

To prove Lemma~\ref{lem:AveRootNorm}, we need to introduce several auxiliary concepts and lemmas. 
An $m\times s$ matrix $G$ is a (standard) Gaussian random matrix if the entries of $G$ are i.i.d.\ standard Gaussian random variables (with mean 0 and variance 1). It is a complex Gaussian random matrix if its real part and imaginary part are independent Gaussian random matrices.
\begin{lem}\label{lem:GaussianAN}
	Suppose $G$ is a standard $m\times s$ real  Gaussian random matrix. Then 
	\begin{equation}
	\bbE\|G\|\leq \frac{\sqrt{2}\Gamma\left(\frac{m+1}{2}\right)}{\Gamma\left(\frac{m}{2}\right)}+\frac{\sqrt{2}\Gamma\left(\frac{s+1}{2}\right)}{\Gamma\left(\frac{s}{2}\right)}\leq \sqrt{m}+\sqrt{s}.
	\end{equation}
\end{lem}
	Usually this lemma is stated without the intermediate term, as it appears in \cite{DaviS01}. However,  the first inequality is essential to achieve our goal. Fortunately, this inequality is already implied by the proof in \cite{DaviS01}. Note that $\sqrt{2}\Gamma\left(\frac{m+1}{2}\right)/\Gamma\left(\frac{m}{2}\right)$ is the average norm of a vector composed of $m$ iid standard Gaussian random variables, while $\sqrt{m}$ is the root mean square norm. This observation implies the second inequality in the lemma, which is nearly tight when $m, s$ are large. 

\begin{lem}\label{lem:GaussianANC}
	Suppose $G$ is a standard $m\times s$ complex   Gaussian random matrix. Then 
	\begin{equation}
	\bbE\|G\|\leq \frac{2\sqrt{2}\Gamma\left(\frac{m+1}{2}\right)}{\Gamma\left(\frac{m}{2}\right)}+\frac{2\sqrt{2}\Gamma\left(\frac{s+1}{2}\right)}{\Gamma\left(\frac{s}{2}\right)}\leq 2\sqrt{m}+2\sqrt{s}.
	\end{equation}
\end{lem}
This lemma is an immediate consequence of the triangle inequality and Lemma~\ref{lem:GaussianAN} applied to the real and imaginary parts of $G$.

\begin{lem}\label{lem:GaussianUniform}
	\begin{equation}
	\bbE \|\rho_A\|^a=\frac{\Gamma(k)\bbE\|G\|^{2a}}{2^a\Gamma(k+a)} \quad\forall a\geq0,
	\end{equation}
	where $G$ is a complex (real) Gaussian random matrix of size $d_A\times d_B$ and $k=d_Ad_B$ ($k=d_Ad_B/2$ in the real case).	
\end{lem} 
\begin{proof}
	It is well known that $G/\|G\|_2$ considered as a unit vector in $\mathcal{H}=\mathcal{H}_A\otimes \mathcal{H}_B$ is distributed uniformly. In addition, the spectrum of $G/\|G\|_2$ is independent of  the Frobenius norm $\|G\|_2=\sqrt{\tr\{GG^\dag\}}$. Therefore,
	\begin{equation}
	\bbE\|G\|^{2a}=\bbE[\tr\{GG^\dag\}]^a \bbE \left\|\frac{G}{\|G\|_2}\right\|^{2a}=  \bbE[\tr\{GG^\dag\}]^a	\bbE \|\rho_A\|^a=\frac{2^a\Gamma(k+a)}{\Gamma(k)}\bbE \|\rho_A\|^a,
	\end{equation}
	from which the lemma follows. 
	Here the last equality in the above equation follows from the fact that $\tr\{GG^\dag\}$ obeys $\chi^2$-distribution with $2k$-degrees of freedom and pdf. 
	\begin{equation}
	f(x)=\frac{x^{k-1}\rme^{-x/2}}{2^k\Gamma(k)},
	\end{equation}
	which satisfies 
	\begin{equation}
	\int x^a f(x)\rmd x=\frac{2^a\Gamma(k+a)}{\Gamma(k)} \quad \forall a\geq0.
	\end{equation}
\end{proof}

\begin{proof}
	According to Lemmas~\ref{lem:GaussianUniform} and \ref{lem:GaussianAN}, in the real case, we have
	\begin{align}
	&\bbE \sqrt{\|\rho_A\|}=\frac{\Gamma\left(\frac{d_Ad_B}{2}\right)\bbE\|G\|}{\sqrt{2}\Gamma(\frac{d_Ad_B+1}{2})}\leq \frac{\Gamma\left(\frac{d_Ad_B}{2}\right)}{\sqrt{2}\Gamma(\frac{d_Ad_B+1}{2})}\left(\frac{\sqrt{2}\Gamma\left(\frac{d_A+1}{2}\right)}{\Gamma\left(\frac{d_A}{2}\right)}+\frac{\sqrt{2}\Gamma\left(\frac{d_B+1}{2}\right)}{\Gamma\left(\frac{d_B}{2}\right)}\right)\nonumber\\
	=&\frac{\gamma(d_B)}{\gamma(d_A d_B)}\frac{1}{\sqrt{d_A}}+\frac{\gamma(d_A)}{\gamma(d_A d_B)}\frac{1}{\sqrt{d_B}}\leq \frac{1}{\sqrt{d_A}}+\frac{1}{\sqrt{d_B}},
	\end{align}
	where $\gamma(m):=\Gamma(\frac{m+1}{2})/(\sqrt{m}\Gamma(\frac{m}{2})$, and  the last inequality follows from the fact that  $\gamma(m)$ is monotonic increasing with $m$ for $m\geq1$. This conclusion is intuitive if we  observe that $\gamma(m)$ is equal to the ratio of the mean length over the root mean square length of a standard Gaussian random vector with $m$ components. To derive an analytical proof, we can compute the log-derivative of $\gamma(m)$ with respect to $m$, note that the definition of $\gamma(m)$ can be extended to positive real numbers. Straightforward calculations shows that
	\begin{equation}
	\frac{\rmd \ln \gamma(m)}{\rmd m}=\frac{1}{2}\left[\psi^{(0)}\Bigl(\frac{m+1}{2}\Bigr)-\psi^{(0)}\Bigl(\frac{m}{2}\Bigr)-\frac{1}{m}\right]\geq \frac{1}{4}\left[\psi^{(0)}\Bigl(\frac{m+2}{2}\Bigr)-\psi^{(0)}\Bigl(\frac{m}{2}\Bigr)-\frac{2}{m}\right]=0.
	\end{equation}
	Here $\psi^{(0)}$ denotes the digamma function (instead of a ket), the inequality follows from the concavity of $\psi^{(0)}$, and the last equality follows from the identity $\psi^{(0)}(x+1)=\psi^{(0)}(x)+\frac{1}{x}$. 
	
	In the complex case,  Lemmas~\ref{lem:GaussianUniform} and \ref{lem:GaussianANC} imply that 
	\begin{align}
	&\bbE \sqrt{\|\rho_A\|}\leq \sqrt{2}
	\left(\frac{\gamma(d_B)}{\gamma(2d_A d_B)}\frac{1}{\sqrt{d_A}}+\frac{\gamma(d_A)}{\gamma(2d_A d_B)}\frac{1}{\sqrt{d_B}}\right)\leq\sqrt{2}\left( \frac{1}{\sqrt{d_A}}+\frac{1}{\sqrt{d_B}}\right),
	\end{align}
	where the second inequality follows from the monotonicity of $\gamma(\cdot)$, as in the real case.
\end{proof}

\end{spacing}


\end{document}